\documentclass[a4paper,11pt]{article}
\usepackage[utf8]{inputenc}
\usepackage[T1]{fontenc}
\usepackage{textcomp}
\usepackage{fancyhdr}
\usepackage{fullpage}
\usepackage{lmodern}
\usepackage{amsmath,amssymb,bm,bbm}
\usepackage{algorithm}
\usepackage{algorithmic}
\usepackage{float}
\usepackage{mathtools}
\usepackage{tikz}
\usetikzlibrary{patterns}
\usepackage{graphicx}
\usepackage{extarrows}
\usepackage{caption}
\usepackage{graphicx, subfig}
\usepackage{titling}
\usepackage{yhmath}
\usepackage{mathrsfs}
\usepackage{cite}
\usepackage{footnote}
\usepackage{amsthm}
\usepackage{color}
\usepackage{slashed}
\usepackage{verbatim}
\usepackage[colorlinks=true,linkcolor=red,citecolor=blue]{hyperref}
\usepackage{cleveref}
\usepackage{amsfonts,euscript}
\usepackage{latexsym, multicol}
\usepackage{epstopdf}
\usepackage{psfrag}
\usepackage{mathrsfs}
\usepackage{xcolor}
\usepackage{comment}
\usepackage{authblk}
\usepackage{indentfirst}
\usepackage{lipsum}
\DeclareFontFamily{U}{mathx}{\hyphenchar\font45}
\DeclareFontShape{U}{mathx}{m}{n}{
<5> <6> <7> <8> <9> <10>
<10.95> <12> <14.4> <17.28> <20.74> <24.88>
mathx10}{}
\DeclareSymbolFont{mathx}{U}{mathx}{m}{n}
\DeclareFontSubstitution{U}{mathx}{m}{n}
\DeclareMathAccent{\widecheck}{0}{mathx}{"71}
\numberwithin{equation}{section}
\allowdisplaybreaks[3]
\def\fo{\mathring{f}}
\newcommand{\f}{{\sqrt{f}}}
\def\psio{\mathring{\psi}}
\def\phio{\mathring{\phi}}
\def\nabo{\mathring{\nab}}
\def\go{\mathring{g}}
\def\ko{\mathring{k}}
\def\Fo{\mathring{F}}

\def\Xt{\widetilde{X}}
\def\Xc{\widecheck{X}}
\def\vphi{{\hat{n}}}

\def\Bt{\widetilde{B}}

\def\Hdot{\dot{H}}
\def\Ft{\widetilde{F}}

\def\kl{{k_*}}
\def\io{\iota}
\def\kc{\widecheck{k}}
\def\tpsic{\widecheck{e_0\psi}}

\def\bo{\square}
\def\et{\widetilde{e}}

\def\kt{\widetilde{k}}
\def\Ll{\mathbb{L}}
\def\gt{\widetilde{{\g}}}
\def\Hh{\mathbb{H}}
\def\ev{{\vec{e}}}
\def\Dd{\mathbb{D}}
\def\nt{\widetilde{n}}

\def\qIt{\widetilde{q_I}}
\def\qJt{\widetilde{q_J}}
\def\qKt{\widetilde{q_K}}
\def\qit{\widetilde{q_i}}

\def\ec{\widecheck{e}}
\DeclareMathOperator{\dvol}{dvol}
\renewcommand{\c}{\cdot}
\newcommand{\DD}{\mathcal{D}}
\newcommand{\pr}{\partial}
\newcommand{\bg}{\mathbf{g}}

\def\gat{\widetilde{\ga}}
\def\psit{\widetilde{\psi}}

\newcommand{\M}{\mathcal{M}}
\newcommand{\D}{\mathbf{D}}
\newcommand{\Ric}{{\ric}}
\DeclareMathOperator{\sRic}{Ric}
\newcommand{\g}{\bg}
\newcommand{\R}{{\mathbf{R}}}

\def\La{\Lambda}

\def\la{\lambda}
\def\Tt{\mathbb{T}}

\def\dq{\de q}
\def\MM{\mathcal{M}}
\def\Si{\Sigma}
\def\si{\sigma}
\def\ga{\gamma}
\def\Ga{\Gamma}
\def\a{\alpha}
\def\b{\beta}
\def\ka{\kappa}
\def\de{\delta}

\def\nab{\nabla}

\def\Gab{\DD_b}
\def\Gag{\DD_g}
\def\Gaw{\DD_w}
\DeclareMathOperator{\diag}{diag}
\def\ep{\varepsilon}
\def\les{\lesssim}
\def\TT{{\mathcal{T}}}
\def\jp{\langle p\rangle}

\DeclareMathOperator{\err}{Err}

\DeclareMathOperator{\supp}{supp}

\DeclareMathOperator{\tr}{tr}
\DeclareMathOperator{\sdiv}{div}

\def\ric{\mathbf{Ric}}
\newtheorem{thm}{Theorem}[section]
\newtheorem{prop}[thm]{Proposition}
\newtheorem{lem}[thm]{Lemma}

\newtheorem{rk}[thm]{Remark}
\newtheorem{df}[thm]{Definition}
\title{Stability of Big Bang Singularity\\ for the Einstein--Maxwell--Scalar Field--Vlasov System\\ in the Full Strong Sub-Critical Regime}
\author{Xinliang An, Taoran He, Dawei Shen}
\begin{document}
\maketitle
\begin{abstract}
In $3+1$ dimensions, we study the stability of Kasner solutions for the Einstein--Maxwell--scalar field--Vlasov system. This system incorporates gravity, electromagnetic, weak and strong interactions for the initial stage of our universe. Due to the presence of the Vlasov field, various new challenges arise. By observing detailed mathematical structures and designing new delicate arguments, we identify a new \textit{strong sub-critical regime} and prove the nonlinear stability with Kasner exponents lying in this full regime. This extends the result of Fournodavlos--Rodnianski--Speck \cite{FRS} from the Einstein--scalar field system to the physically more complex system with the Vlasov field.
\end{abstract}
\tableofcontents
\section{Introduction}
The presence of the Big Bang singularity poses fundamental mathematical and physical questions about the nature of our early universe. The interaction of matters and its geometric implications are quite mysterious. Mathematically, for Einstein field equations, a particularly interesting class of cosmological solutions describing the initial Big Bang singularities are the Kasner spacetimes, which provide exact, anisotropic models for spacetime dynamics near the Big Bang singularities.

In recent years, encouraging progress has been made toward proving the stability of Big Bang singularities. Without symmetry assumptions, pioneering explorations on the stable big bang formation were carried out by Rodnianski--Speck, Speck \cite{RSscalar,RSstiff,RSvacuum,Speck}. Remarkably, in \cite{FRS} Fournodavlos--Rodnianski--Speck demonstrated the nonlinear stability of Kasner solutions to the Einstein--scalar field system for the entire sub-critical regime. Later, in a notable result \cite{BOZ:subcrit}, Beyer--Oliynyk--Zheng established a localization-allowed version of \cite{FRS} using the Fuchsian method. See also a new localized construction \cite{AF} by Athanasiou--Fournodavlos. For the Fuchsian approach applying to the Einstein vacuum equations with cosmological constant under polarized $\mathbb{T}^2$ symmetry, we refer to Ames--Beyer--Isenberg--Oliynyk  \cite{ABIO:2021_Royal_Soc,ABIO:2021}. We also would like to mention another recent inspiring work \cite{GOPR} by Groeniger--Petersen--Ringstrom, where they extended the result of \cite{FRS} to allow for non-vanishing scalar potential and a broader class of initial data within the non-perturbative regime of Kasner-like solutions. Meanwhile, Fajman--Urban \cite{FU22} studied the stability of the (homogeneous) Friedman--Lema\^{i}tre--Robertson--Walker (FLRW) solutions for the Einstein--scalar field system  and in \cite{FU24} they achieved the past stability of FLRW solutions to the Einstein--scalar field--Vlasov system. We also note that, when incorporating the Vlasov matter field (for both massive and massless cases), Urban \cite{Urban} established the past stability of FLRW solutions in $1+2$ dimensions. It is also worth mentioning that Boothroyd \cite{Boothroyd} studied the big bang stability of the Einstein--Maxwell--scalar field system (no Vlasov) in $\mathbb{T}^2$-symmetry.
\vspace{2mm}

In this paper, along the line of \cite{FRS} we aim to establish the nonlinear stability of (anisotropic) Kasner solutions for a more physically complicated system in a regime as large as possible. In particular, here we study the stable big formation for the Einstein--Maxwell--scalar field--Vlasov system. 
In the Big Bang setting, the Maxwell field accounts for the electromagnetic interaction, the scalar field models the weak interaction of neutrinos. In this paper, we employ the massless Vlasov field to represent the quark–gluon plasma (QGP) in the earliest stages after the Big Bang, which governs by the strong interactions.\footnote{Here for mathematical simplicity we treat the particles as collisionless. We also remark that in physics literature, when studying the beginning stage of the Big Bang singularity, massless (instead of massive) particles are more commonly used.}
Our studied comprehensive system thus reflects the four fundamental forces: gravity, electromagnetism, weak interaction, and strong interaction.
\vspace{2mm}

Toward proving the nonlinear stability results for the corresponding ``entire'' sub-critical regime, the Vlasov imposes many new challenges. By making new observations of the Vlasov-related system, we introduce a new concept \textit{strong sub-critical regime}. Our main theorem establishes the nonlinear stability of Kasner solutions in the entire \textit{strong sub-critical regime} for the Einstein--Maxwell--scalar field--Vlasov system. We prove that, for initial data sufficiently close to a Kasner solution with Kasner components lying in the strong sub-critical regime, the dynamical solution is well-controlled and exhibits Kasner-type curvature blow-ups. This extends the main result of \cite{FRS} from the Einstein--scalar field system to the Einstein--Maxwell--scalar field--Vlasov system. In particular, for the Einstein--Maxwell--scalar field system without the Vlasov field, we obtain and recover the main conclusion of \cite{FRS} in the entire sub-critical regime.
\vspace{2mm}

In this paper, our key innovations are our treatments of the Vlasov field. By operating conservation laws directly and by employing weighted energy estimates, we establish sharp lower-order and higher-order estimates for the Vlasov field, and we also allow the perturbations of the Vlasov field to be with non-compact support in the mass shell.
\subsection{The Einstein--Maxwell--Scalar Field--Vlasov System}
In this paper, we study the $3+1$ dimensional Lorentz manifold $(\MM,\g)$ and our main goal is to investigate the stable big bang formation for the below Einstein--Maxwell--scalar field--Vlasov system (EMSVS):
\begin{align}\label{EMSV}
\Ric_{\mu\nu}-\frac12\R\g_{\mu\nu}=\D_\mu\psi\D_\nu\psi-\frac{1}{2}\g_{\mu\nu}\D_\a\psi\D^\a\psi+2\left(F_{\mu\a}{F_{\nu}}^\a-\frac{1}{4}\g_{\mu\nu}F_{\a\b}F^{\a\b}\right)+T^{(V)}_{\mu\nu}.
\end{align}
Here we have\footnote{We use $P(t,x)$ to denote the mass shell at $(t,x)$. See Section \ref{sssecmassshell} for more explanations.}
\begin{align}\label{dfT}
    T^{(V)}_{\mu\nu}\coloneqq\int_{P(t,x)}fp_\mu p_\nu\dvol
\end{align}
and it denotes the energy-momentum tensor of a massless Vlasov field. 

Contracting \eqref{EMSV} with $\g^{\mu\nu}$, we deduce that the scalar curvature $\R$ of the spacetime $(\M,\g)$ satisfies
\begin{align*}
    \R=\D_\a\psi\D^\a\psi.
\end{align*}
Injecting it into \eqref{EMSV}, we infer that the Einstein--Maxwell--scalar field--Vlasov field equations can be rewritten as
\begin{align}
\begin{split}\label{EMS}
    \Ric_{\mu\nu}&=\D_\mu\psi\D_\nu\psi+2F_{\mu\a}{F_{\nu}}^\a-\frac{1}{2}\g_{\mu\nu}F_{\a\b}F^{\a\b}+T^{(V)}_{\mu\nu}.
\end{split}
\end{align}
Regarding the matter fields, we have the following equations:
\begin{itemize}
\item Wave equation for the scalar field $\psi$:
\begin{align}\label{Wave}
    \bo_\g\psi=0.
\end{align}
\item The Maxwell equations for the electromagnetic field $F$:
\begin{align}\label{Maxwell}
    \D_{\a}F^{\a\b}=0,\qquad\qquad \D_{[\a}F_{\b\ga]}=0.
\end{align}
\item The (massless) Vlasov equation for the distribution function $f$:
\begin{align}\label{Vlasov}
    X(f)=0,
\end{align}
where $X\in\Ga(TT\M)$ denotes the \emph{geodesic spray}, i.e., the generator of the geodesic flow of $(\M,\g)$.
\end{itemize}
The system \eqref{EMS}--\eqref{Vlasov} admits a well-posed initial value formulation and sufficiently regular initial data yield unique solutions. An initial data set for the system \eqref{EMS}--\eqref{Vlasov} consists of a septuplet $(\Si_1,\go,\ko,\psio,\phio,\Fo,\fo)$, where $\go$ is a Riemannian metric on $\Si_1$, $\ko$ is a symmetric two-tensor, $(\psio,\phio)$ is a pair of scalar functions, $\Fo$ is a $2$-form and $\fo$ is a scalar function defined on the tangent bundle $T\Si_1$. This septuplet satisfies
\begin{align*}
\psi|_{t=1}&=\psio, &\pr_t\psi|_{t=1}&=\phio,\\
F_{\mu\nu}|_{t=1}&=\Fo_{\mu\nu}, &  f|_{t=1}&=\fo.
\end{align*}
We note that the admissible geometric initial data must satisfy the Hamiltonian and momentum constraint equations, which take the form of
\begin{align}
    \mathring{R}-|\ko|^2+(\tr\ko)^2&=\phio^2+|\nabo\psio|^2+4\left(\Fo_{0C}\Fo_{0C}+\frac{1}{4}\Fo_{\a\b}\Fo^{\a\b}\right)+2T_{00},\label{Hamiltonconstraint}\\
    (\mathring{\sdiv}\ko)_I-\nabo_I\tr\ko&=-\phio\left(\nabo_I\psio\right)-2\Fo_{0C}\Fo_{IC}-T_{0I}.\label{momentumconstraint}
\end{align}
Here, $\nabo$ and $\mathring{R}$ are the Levi--Civita connection and the scalar curvature of $\go$ respectively, and we use $C, I, J, K$ to denote indices $1,2,3$ for spatial variables. Moreover, the initial data of the electromagnetic field $\Fo$ satisfy
\begin{align}
\nabo_I\Fo_{JK}+\nabo_J\Fo_{KI}+\nabo_K\Fo_{IJ}&=0\qquad\qquad \text{with} \quad I\ne J, J \ne K, K\ne I.\label{Maxwellelliptic}
\end{align}
\subsection{Kasner Solutions and Strong Sub-Critical Condition}
This paper is to explore the stability of curvature-blowup phenomena for a large class of the Kasner solutions on $(0,\infty)\times\Tt^3$. These solutions take the form
\begin{align}
\begin{split}\label{kasnerg}
    \gt&=-dt\otimes dt+\sum_{I=1,2,3}t^{2\qIt}dx^I\otimes dx^I,\\
    \widetilde{\psi}&=\widetilde{B}\log t,\qquad\quad \widetilde{F}=0,\qquad\quad \widetilde{f}=0.
\end{split}
\end{align}
Here $\{\widetilde{q_I}\}_{I=1,2,3}$ and $\widetilde{B}$ are constants that satisfy the algebraic relations
\begin{align}\label{kasnerq}
    \sum_{I=1}^3\qIt=1,\qquad\quad\sum_{I=1}^3\qIt^2=1-\Bt^2.
\end{align}
We refer to $\{\widetilde{q_I}\}_{I=1,2,3}$ as the Kasner exponents.

We note that the constraints in \eqref{kasnerq}, which arise from the constant mean curvature (CMC) condition $\tr\kt=-\frac{1}{t}$ and the Hamiltonian constraint \eqref{Hamiltonconstraint}, ensure that \eqref{kasnerg} are indeed solutions to the Einstein--Maxwell--scalar field--Vlasov system \eqref{EMS}.
\begin{df}\label{sscdf}
We say that a Kasner solution \eqref{kasnerg} to the system \eqref{EMS} on $(0,\infty)\times\Tt^3$ satisfies the \textbf{strong sub-critical condition} (or \textbf{strong stability condition}), if its Kasner exponents satisfy
\begin{align}\label{strongstability}
    \max_{I,J,K=1,2,3}\{\qIt+\qJt-\qKt\}<1.
\end{align}
\end{df}
\begin{rk}
Note that the condition \eqref{strongstability} is slightly more restrictive than the so-called \emph{sub-critical condition} employed by Fournodavlos--Rodnianski--Speck in \cite{FRS}: 
\begin{align}\label{kasnerstability}
    \max_{\substack{I,J,K=1,2,3 \\ I\ne J}}\{\qIt+\qJt-\qKt\}<1.
\end{align}
To see the difference between \eqref{strongstability} and \eqref{kasnerstability}, the admissible regions for the exponents $\widetilde{q_M}\coloneqq \max\limits_{{I=1,2,3}}{\qIt}$ and $\widetilde{q_m}\coloneqq \min\limits_{{I=1,2,3}}{\qIt}$ are portrayed as in \Cref{fig:Kasnerexp}. Here the regime of strong sub-criticality \eqref{strongstability} corresponds to the gray region, while the sub-critical condition \eqref{kasnerstability} includes both the gray and hatched region.
\begin{figure}[htp]
    \centering
    \includegraphics[width=0.5\linewidth]{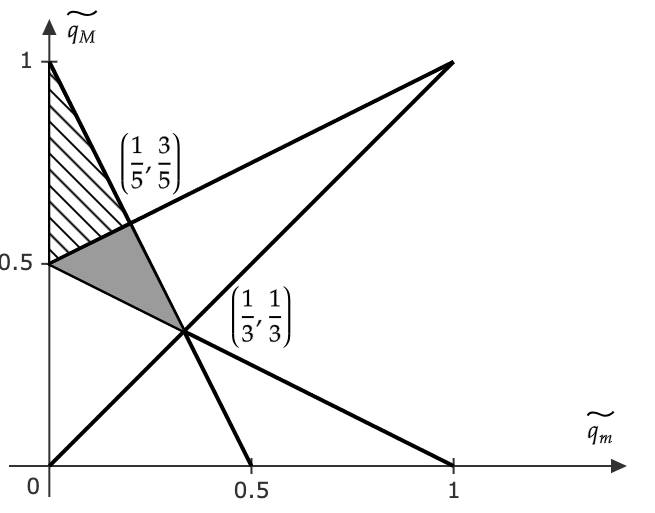}
    \caption{Admissible Regions for \eqref{strongstability} and \eqref{kasnerstability}}
    \label{fig:Kasnerexp}
\end{figure}
\end{rk}
\subsection{Main Theorem}
In this section, we state a rough version of our main theorem. The explicit statement is referred to \Cref{mainestimates}.
\begin{thm}[main theorem (rough version)]\label{Rough Main thm 1}
We study the Einstein--Maxwell--scalar field--(massless) Vlasov system \eqref{EMS} on the slab $(0,1]\times\Tt^3$. Let $(\Si_1,g,k,\psi,\mathring{\phi},\Fo,\fo)$ be its initial data set, that is close to a Kasner solution \eqref{kasnerg} with Kasner exponents satisfying the strong sub-critical condition \eqref{strongstability}. Then, for this system, there exists a unique solution and it obeys quantitative stability estimates provided in Theorem \ref{mainestimates}.
\end{thm}
\begin{rk}
   The proof of \Cref{Rough Main thm 1} crucially builds upon our sharp controls of both lower-order and higher-order derivatives of the Vlasov field. These sharp controls enable us to demonstrate the stable big bang formation for EMSVS in the full strong sub-critical regime.
\end{rk}
\begin{rk}
    Owing to our new method for handling the Vlasov equation, in \Cref{Rough Main thm 1} the initial perturbation for the Vlasov field is not necessarily compactly supported in the mass shell $P(1, x)$.
\end{rk}
When taking the Vlasov field $f$ to be identically $0$, adapting our arguments in this paper as well as the approach in \cite{FRS}, we can also show the nonlinear stability of Kasner solutions to the Einstein--Maxwell--scalar field system in the full sub-critical regime. We summarize the conclusion here:
\begin{prop}
Consider the Einstein--Maxwell--scalar field system, i.e., \eqref{EMS} with $f=0$ on the slab $(0,1]\times\Tt^3$. Let  $(\Si_1,g,k,\psi,\mathring{\phi},\Fo)$ be its initial data set that is close to a Kasner solution \eqref{kasnerg} with Kasner exponents satisfying the stability condition \eqref{kasnerstability} (with $\fo=0$). Then, for this system there exists a unique solution on $(0,1]\times\Tt^3$ and it satisfies quantitative stability bounds.
\end{prop}
\subsection{Main Difficulty and New Ingredients}\label{overview}
This section is devoted to highlighting the key ideas and new ingredients presented in our proof of \Cref{Rough Main thm 1}. 
\subsubsection{Strong Sub-Critical Condition and AVTD Behavior}
Our results yield the dynamical stability of the Kasner Big Bang singularity for the Einstein--Maxwell--scalar field--Vlasov system if the exponents of the background Kasner solution verify the strong sub-critical condition \eqref{strongstability}. Based on the hyperbolic estimates established for various geometric quantities, we are able to show that these perturbed spacetimes converge to Kasner-like\footnote{See details in Proposition \ref{Kasnerlimit}.} Big-Bang solutions when they are approaching the Big-Bang singularity.

The condition \eqref{strongstability} is referred to as the \textit{strong sub-critical condition} or the \textit{strong stability condition} for the following two reasons:
\begin{itemize}
\item The requirement \eqref{strongstability} is more restrictive than the stability condition \eqref{kasnerstability} employed in \cite{FRS}. For example, the close-to-endpoint case $(\widetilde{q_1}, \widetilde{q_2}, \widetilde{q_3})=(\epsilon, \epsilon, 1-2\epsilon)$ with $0<\epsilon\ll 1$ verifies the stability condition \eqref{kasnerstability}, but is excluded in our strong sub-critical regime.
\item The strong stability condition \eqref{strongstability} is necessary in order to establish the stability of Big Bang formation when taking into account the effect of the massless Vlasov field which describes the strong interaction. We point out that adding only the Maxwell field would keep the same sub-critical regime \eqref{kasnerstability} and lead to the same conclusion, namely, all sub-critical Kasner solutions are stable for the Einstein--Maxwell--scalar field system.
\end{itemize}
In our proof, we will guarantee that all geometric quantities exhibit the so-called \emph{asymptotically velocity term dominated} (AVTD) behaviors. We now clarify its meaning. For the sake of illustration, we first introduce the following key parameters.\footnote{See also Section \ref{seckey} for more precise definitions.} According to \eqref{strongstability}, we pick parameters $q, \si\in (0, 1)$ such that
\begin{align*}
    \max_{I,J,K=1,2,3}\big\{1-\qIt,\qIt+\qJt-\qKt\big\}=q-4\si=1-6\si.
\end{align*}
Notice that the parameter $\si>0$ allows us to lose a $\si$-room for controlling the blow-up rate, which plays a notable role in the derivation of the sharp lower-order estimates.

We also define
\begin{align}\label{dfqintro}
    q_M\coloneqq\max_{I=1,2,3}\{\qIt\}+\si,\qquad\quad \dq\coloneqq\max_{I,J=1,2,3}\{\qIt-\qJt\}+\si.
\end{align}
The AVTD behavior manifests that, for any tensorfield $h$, the blow-up speed of the time derivative of $h$ surpasses that of the spatial derivatives of $h$. Quantitatively, in our paper, we have
\begin{equation}\label{AVTDintro}
    e_0h\sim t^{-1}h,\quad\qquad \ev h\sim t^{-q_M-\dq}h,
\end{equation}
where $e_0$ and $\ev$ denote the normalized time derivative and the normalized spatial derivatives, respectively. Observing that
\begin{align}\label{AVTDintr}
    q_M+\dq=q-2\si<q<1,
\end{align}
we thus infer from \eqref{AVTDintro} that the time derivative of $h$ is more singular compared to the spatial derivatives of $h$, which reflects the AVTD behavior.
\subsubsection{Lower-Order Estimates and Top-Order Estimates}
The proof of Theorem \ref{Rough Main thm 1} mainly relies on deriving the desired estimates for both lower-order derivatives and higher-order derivatives of all geometric quantities and matter fields. Specifically, for the connection coefficients $k,\ga$ defined in \eqref{Def:k} and \eqref{gaIJK}, the scalar field $\psi$ and the Maxwell field $F$, we aim to establish the following bounds\footnote{See \Cref{secortho} and \Cref{seckey} for the definitions of $\kc$, $\ga, \tpsic,F$ and those of $q$ and $\kl$. We also note that the homogeneous Sobolev space $\Hdot^\kl(\Si_t)$ is introduced in Section \ref{secSobolev}.}
\begin{align}
    t\|\kc\|_{W^{1, \infty}(\Si_t)}+t^q\|\ga\|_{W^{1,\infty}(\Si_t)}+t^q\|\ev\psi\|_{W^{1,\infty}(\Si_t)}+t\|\tpsic\|_{W^{1,\infty}(\Si_t)}+t^q\|F\|_{W^{1,\infty}(\Si_t)}&\les\ep_0, \label{intro:Est:otherlow}\\
t^{A_*+1}\Big(\|k\|_{\Hdot^\kl(\Si_t)}+\|\ga\|_{\Hdot^\kl(\Si_t)}+\|\vec{e}\psi\|_{\dot{H}^\kl(\Si_t)}+\|e_0\psi\|_{\dot{H}^\kl(\Si_t)}+\|F\|_{\Hdot^\kl(\Si_t)}\Big)&\les\ep_0. \label{intro:Est:otherhigh}
\end{align}
Here $A_*$ and $\kl$ represent large enough fixed constants and $\ep_0>0$ is a small constant, which measures the size of perturbations of the initial data.

Regarding the Vlasov field $f$, we will establish
\begin{equation}\label{intro:Est:f}
    t^{\frac{1+q_M}{2}}\|\TT\|_{L^{\infty}(\Si_t)}+t^{A_*+\de q+\frac{1+q_M}{2}}\|\TT^{(k_*)}\|_{L^2(\Si_t)}\les \ep_0.
\end{equation}
Here $\TT$ and $\TT^{(k_*)}$ denote the second moment of $\sqrt{f}$ and its $k_*$-th order derivatives in the phase space. See \Cref{Subsec:defTT} for their explicit definitions.

As a consequence, employing the interpolation inequality, we also derive
\begin{equation}\label{intro:Est:low+1}
    \begin{aligned}
        &t\|\kc\|_{W^{2, \infty}(\Si_t)}+t^q\|\ga\|_{W^{2,\infty}(\Si_t)}+t^q\|\ev\psi\|_{W^{2,\infty}(\Si_t)}\\&+t\|\tpsic\|_{W^{2,\infty}(\Si_t)}+t^q\|F\|_{W^{2,\infty}(\Si_t)}+t^{\frac{1+q_M}{2}}\|\TT\|_{W^{1, \infty}(\Si_t)}\les \ep_0 t^{-\de q}.
    \end{aligned}
\end{equation}
\begin{rk}\label{rkA*}
    To conclude the proof and to validate the estimates \eqref{intro:Est:otherlow}, \eqref{intro:Est:otherhigh} and \eqref{intro:Est:f}, the order of constant choosing  is crucial. We first select the parameter $A_*$ to be sufficiently large, and then pick the number for top regularity $\kl\in \mathbb{N}$ such that $\frac{A_*}{\kl}\approx \dq$.\footnote{See \eqref{dfkl} for the precise choice of $A_*$ and $\kl$.} Finally, we let the size of the initial perturbation $\ep_0$ be sufficiently small relative to $A_*$ and $\kl$. It is worth noting that in \cite{FRS} Fournodavlos--Rodnianski--Speck are able to select $\frac{A_*}{\kl}$ to be arbitrarily small, whereas in our current paper, we only have that this ratio is close to $\de q$ (not necessarily small) due to the influence of the Vlasov field and we have to work with it. See \Cref{Subsec:Control of Vlasov Field} for more explicit explanations.
\end{rk}
\begin{rk}
    For geometric quantities and matter fields except the Vlasov field, the lower-order estimate \eqref{intro:Est:otherlow} is almost sharp, while the higher-order bound \eqref{intro:Est:otherhigh} is much more singular compared to their exact blow-up behaviors. On the other side, the lower-order $L^{\infty}$ estimate and the higher-order energy estimate of the Vlasov field are both optimal with respect to the power of $t$. Capturing the optimal blow-up rates of the Vlasov field is the key to our entire proof.
\end{rk}
\begin{rk}
Notice that compared to the lower-order estimates in \eqref{intro:Est:otherlow} and \eqref{intro:Est:f}, when adding one more derivative, there is a loss of factor $t^{-\de q}$ in \eqref{intro:Est:low+1}. This presents new difficulties when incorporating the Vlasov field. In contrast, in \cite{FRS} when adding the derivative the additional factor $t^{-A_*/k_*}$ is negligible since $A_*/k_*$ is small there.  This is the main reason why we cannot show the (past) stability of Kasner solutions for the full sub-critical regime as in \cite{FRS} when taking into account the Vlasov field.
\end{rk}
\begin{rk}
In fact, it is rather challenging to prove our main result \Cref{Rough Main thm 1} for the entire strong sub-critical regime. This is only achieved by our careful tracking of the optimal estimates for the Vlasov field as established in \eqref{intro:Est:f} and \eqref{intro:Est:low+1}, which are critical for ensuring the AVTD behavior. Meanwhile, these key estimates in \eqref{intro:Est:f}, \eqref{intro:Est:low+1} heavily rely on our key observations on the structure of the Vlasov equations and our utilization of weighted energy estimates and the equation for the conservation law in a new way. Further details are provided in \Cref{Subsec:Control of Vlasov Field}.
\end{rk}
\subsubsection{Control of Spacetime Geometry and Matter Fields Except Vlasov Field}\label{Subsec:EstexceptVlasov}
In this subsection, we demonstrate the ideas for estimating the geometric quantities and matter fields except for the Vlasov field. The arguments in this subsection are inspired by \cite{FRS} and we extend their approach to control the Maxwell field as well. Employing \eqref{EMS}, \eqref{Wave} and \eqref{Maxwell}, the evolution equations for these quantities can be written as follows:\footnote{Throughout Section \ref{overview}, we use $\cdots$ to denote the error terms.}
\begin{align}
     \pr_t\kc_{IJ}+\frac{1}{t}\kc_{IJ}&=e_C(\ga_{IJC})-e_I(\ga_{CJC})+\cdots, \nonumber\\
     \pr_tS_{IJK}+\frac{\qIt+\qJt-\qKt}{t}S_{IJK}&=\ev k+\cdots, \label{Eqn:S}\\  
       e_0(\widecheck{e_0\psi})+\frac{1}{t}\widecheck{e_0\psi}&=e_C(e_C\psi)+\cdots, \nonumber\\
    e_0(e_I\psi)+\frac{\qIt}{t}(e_I\psi)&=e_I(e_0\psi)+\cdots, \nonumber\\
    \pr_t(F_{0I})+\frac{1-\qIt}{t}F_{0I}&=e_C(F_{CI})+\cdots,\nonumber\\
\pr_t(F_{IJ})+\frac{\qIt+\qJt}{t}F_{IJ}&=e_I(F_{0J})-e_J(F_{0I})+\cdots.
\end{align}
Here $S_{IJK}\coloneqq\ga_{IJK}+\ga_{JKI}$.\footnote{The original evolution equation of $\ga$ is
\begin{equation*}
    \pr_t (\ga_{IJK})=e_K( k_{JI})-e_J(k_{KI})+\cdots.
\end{equation*}
This rewriting \eqref{Eqn:S} with $S_{IJK}$ is a key ingredient in \cite{FRS} by Fournodavlos--Rodnianski--Speck.}

We notice that the above system of equations exhibits the following important features:
\begin{enumerate}
    \item For $\kc$ and $\widecheck{e_0\psi}$, the coefficient in front of the linear term on the left is $1$, while for other quantities, the corresponding coefficient there is at most $q < 1$. This is consistent with the corresponding sharp lower-order estimates in \eqref{intro:Est:otherlow}. These estimates are obtained utilizing the standard evolution lemma.
    \item All quantities except the Vlasov field obey the structured equations of the \textit{Bianchi pairs} $(A, B)$, i.e., they satisfy schematically
    \begin{equation}\label{Bianchipair}
        \pr_t A=DB+\cdots, \qquad \qquad \pr_t B=-D^*A+\cdots
    \end{equation}
    with $D$ being an differential operator on $\Si_t$ and $D^*$ representing its $L^2$-dual. By integrating $\pr_t(|A|^2+|B|^2)$ over $\Si_t$, we obtain the following differential identity
    \begin{equation*}
        \pr_t \left(\int_{\Si_t}|A|^2+|B|^2\right)=2\int_{\Si_t} A\c DB-B\c D^*A+\cdots=\cdots.
    \end{equation*}
    Note that there is the exact cancellation $2\int_{\Si_t}A\c DB-B\c D^*A=0$, which avoids the loss of derivatives. This allows us to implement the $t$-weighted energy estimates to control the top orders.
    \item The error terms in $\cdots$ on the right satisfy the null structures. Roughly speaking, there is no $\Gaw\c\Gaw$ term appearing in $\cdots$, where $\Gaw\in\{\kc, \widecheck{e_0\psi} \}$.
\end{enumerate}
In below, we demonstrate an example on how to derive the lower-order and higher-order estimates for $k$ and $\ga$. Regarding the lower-order estimates, we use
\begin{align*}
    \pr_t\kc_{IJ}+\frac{1}{t}\kc_{IJ}&=O(\ev\ga)+\cdots,\\
    \pr_tS_{IJK}+\frac{\qIt+\qJt-\qKt}{t}S_{IJK}&=O(\ev k)+\cdots.
\end{align*}
From \eqref{AVTDintro} and \eqref{intro:Est:otherlow} we have
\begin{align*}
    \ev\ga\sim \ep_0t^{-q-q_M-\dq},\qquad\quad \ev k\sim\ep_0t^{-1-q_M-\dq}.
\end{align*}
Combining with \eqref{dfqintro}, \eqref{AVTDintr} and direct integrations, we deduce
\begin{align*}
    \|t\kc\|_{L^\infty(\Si_t)}&\les\ep_0+\int_t^1s\|\ev\ga\|_{L^\infty(\Si_s)}ds\les \ep_0+\ep_0\int_t^1 s^{1-2q+2\si}ds\les \ep_0,\\
    \|t^qS_{IJK}\|_{L^\infty(\Si_t)}&\les\ep_0+\int_t^1s^q\|\ev k\|_{L^\infty(\Si_s)}ds\les\ep_0+\ep_0\int_t^1s^{-1+2\si}ds\les\ep_0.
\end{align*}
Observing that the components of $\ga$ are linear combinations of $S_{IJK}$, we hence obtain
\begin{align}
    t\|\kc_{IJ}\|_{L^\infty(\Si_t)}&\les\ep_0,\label{kintrodecay}\\
    t^q\|\ga_{IJK}\|_{L^\infty(\Si_t)}&\les\ep_0. \nonumber
\end{align}
We remark that the estimates in \eqref{kintrodecay} are of particular importance, since they are needed for controlling various borderline terms in the energy estimates.

We proceed to control the top-order derivatives of $(k, \ga)$ and a Bianchi-pair structure will present. By commuting the evolution equations of $k$ and $\ga$ with $\pr^\io$ for $|\io|=\kl$, we obtain
\begin{align}
    \pr_t (\pr^\io k_{IJ})+\frac{1}{t}(\pr^\io k_{IJ})&=e_C(\pr^\io\ga_{IJC})-e_I(\pr^\io\ga_{CJC})+\cdots,\label{ekintro}\\
    \pr_t (\pr^\io \ga_{IJK})&=e_K(\pr^\io k_{JI})-e_J(\pr^\io k_{KI})+\cdots,\label{egaintro}
\end{align}
which obey the Bianchi pair structure as stated in \eqref{Bianchipair}. We then apply the $t^p$-weighted energy estimates for $(k,\ga)$.\footnote{The $t^p$--weighted estimates for $(k,\ga)$ used here are analogous to the $r^p$--weighted estimates for Bianchi pairs. See the discussions in Section 4.1 of \cite{Shen22}.
} We first multiply \eqref{ekintro} by $2t^{2A_*+2}(\pr^\io k_{IJ})$ and \eqref{egaintro} by $t^{2A_*+2}(\pr^\io\ga_{IJK})$. By adding the two identities, taking the sum for $I,J,K=1,2,3$ and integrating it on $\Si_t$, we obtain\footnote{We also utilize the momentum constraint equation \eqref{momentumconstraint}, which gives $e_C(k_{CI})=\cdots$.  See Proposition \ref{estkga} for more details.}
\begin{align}
\begin{split}\label{prtkgaintro}
&\pr_t\left(\int_{\Si_t}t^{2A_*+2}\sum_{I,J}|\pr^\io k_{IJ}|^2+\frac{1}{2}t^{2A_*+2}\sum_{I,J,K}|\pr^\io\ga_{IJK}|^2\right)\\
=&2A_*\sum_{I,J}\int_{\Si_t}t^{2A_*}|\pr^\io k_{IJ}|^2+(2A_*+2)\sum_{I,J,K}\int_{\Si_t}t^{2A_*+2}|\pr^\io\ga_{IJK}|^2+\cdots.
\end{split}
\end{align}
Then, conducting the integration for \eqref{prtkgaintro} from $t$ to $1$, we deduce
\begin{align*}
    &\sum_{I,J}\int_{\Si_t}t^{2A_*+2}|\pr^\io k_{IJ}|^2+\sum_{I,J,K}\int_{\Si_t}t^{2A_*+2}|\pr^\io\ga_{IJK}|^2\\
    +&A_*\sum_{I,J}\int_t^1\int_{\Si_t}t^{2A_*}|\pr^\io k_{IJ}|^2 ds+A_*\sum_{I,J,K}\int_t^1\int_{\Si_t}t^{2A_*+2}|\pr^\io\ga_{IJK}|^2 ds\\
    \les&\ep_0^2+\sum_{I,J}\int_t^1\int_{\Si_t}t^{2A_*}|\pr^\io k_{IJ}|^2 +\cdots,
\end{align*}
where the constant involved in $\les$ is independent of $A_*$ and $k_*$. Taking $A_*$ to be large enough, we can absorb the borderline bulk term  $\sum_{I,J}\int_t^1\int_{\Si_t}t^{2A_*}|\pr^\io k_{IJ}|^2$ on the right. 
Thus, we obtain the desired estimate 
\begin{equation*}
t^{A_*+1}\Big(\|k\|_{\Hdot^\kl(\Si_t)}+\|\ga\|_{\Hdot^\kl(\Si_t)}\Big)\les\ep_0.
\end{equation*}
\subsubsection{Control of the Vlasov Field}\label{Subsec:Control of Vlasov Field}
The treatment of the Vlasov field is quite different from the previous arguments for $ k, \ga, \psi$ and $F$, especially for the top-order energy estimates. This is because the Vlasov equation \eqref{Vlasov} does not exhibit the structure of the Bianchi pairs and an approach as above does not work. Here we develop a new method to deal with the Vlasov equation, which also enables us to allow the non-compact initial perturbation for the Vlasov field in the phase space.

Our new aim here is to prove the following estimates:
\begin{align}
    \|T\|_{L^\infty(\Si_t)}&\les\ep_0 t^{-1-q_M}, \label{IntroEst:T}\\
    \max_{|\io_1|+|\io_2|\leq\kl}\left\|(p^0)^\frac{1}{2}\f^{(\io_1,\io_2)}\right\|_{L^2(T\Si_t)}&\les\ep_0 t^{-A_*-\dq-\frac{q_M+1}{2}}, \label{IntroEst:f}
\end{align}
where for notational simplicity we write $T=T^{(V)}$ and denote\footnote{Here, $p\pr_p$ denotes all the vectorfield in the form of $p^J\pr_{p^K}$ with $J,K=1,2,3$ and $\pr\coloneqq\{\pr_{x^1},\pr_{x^2},\pr_{x^3}\}$.}
$$
\f^{(\io_1,\io_2)}\coloneqq\pr^{\io_1}(p\pr_p)^{\io_2}\f.
$$

We start with deriving the $L^{\infty}$ bound for $T$. Instead of utilizing the Vlasov equation \eqref{Vlasov}, we employ the conservation law for the Vlasov part of the energy momentum tensor $T$, i.e.,
 \begin{align}\label{conservationintro1}
    \D_\mu T^{\mu\nu}=0,
\end{align}
 and utilize the non-negativity of the diagonal entries of $T$.
 
 Specifically, from \eqref{dfT} we have that
\begin{align*}
    T_{\mu\mu}=\int_{P(t,x)}f(p_\mu)^2\dvol\geq 0,\qquad\quad T_{00}=\sum_{I=1}^3T_{II}.
\end{align*}
These imply
\begin{align}\label{Intro:Est sum TII}
    \sum_{I=1}^3\frac{\qIt}{t}T_{II}\leq \frac{\max_{I=1,2,3}\qIt}{t}\sum_{I=1}^3T_{II}\leq\frac{\max_{I=1,2,3}\qIt}{t}T_{00}.
\end{align}
    By rewriting \eqref{conservationintro1} as
\begin{equation*}
    \pr_t(T_{00})+\frac{1}{t}T_{00}+\sum_{I=1}^3\frac{\qIt}{t}T_{II}=e_C(T_{0C})+\cdots,
\end{equation*}
in view of \eqref{Intro:Est sum TII} we hence deduce
\begin{align*}
    \pr_t(T_{00})+\frac{1+\max_{I=1,2,3}\{\qIt\}}{t}T_{00}\geq e_C(T_{0C})+\cdots.
\end{align*}
Multiplying it by $t^{1+q_M}$ on both sides with $q_M=\max_{I=1,2,3}\{\qIt\}+\si$, we then obtain
\begin{align*}
    \pr_t(t^{1+q_M}T_{00})-\si t^{q_M}T_{00}\geq t^{1+q_M}e_C(T_{0C})+\cdots.
\end{align*}
Consequently, the integration of the above inequality from $t$ to $1$ yields
\begin{align*}
    t^{1+q_M}T_{00}+\si\int_t^1 s^{q_M} T_{00}ds\les\ep_0+\int_{t}^1 s^{1+q_M}\|\ev  \, T\|_{L^\infty(\Si_s)}ds\les\ep_0.
\end{align*}
Recalling from \eqref{dfT} that 
$$
|T_{\mu\nu}|\leq T_{00},\qquad\quad\text{for}\quad \mu,\nu=0,1,2,3,
$$
we hence derive the following desired lower-order estimate:
\begin{align*}
    |T|\les\ep_0t^{-1-q_M}.
\end{align*}
We note that, when examining the linear part of the null geodesic equations for $p^{\mu}$, the expected bound for the contribution of $p^{I}$ in $\supp_{P(t,x)} f$ is $t^{-\qIt}$ and the contribution of $p^0$ is like $t^{-q_M}$. As a result, the expected upper bound for $T$ is $\ep_0 t^{-q_M}\c t^{-\sum_I \qIt}=\ep_0 t^{-1-q_M}$,\footnote{Note that the volume form in the mass shell $P(t, x)$ is $\dvol=(p^0)^{-1}dp^1 dp^2 dp^3$.} indicating that the estimate \eqref{IntroEst:T} is sharp.
\vspace{2mm}

We proceed to establish the top-order estimate \eqref{IntroEst:f}. We consider a new form of the Vlasov equation, namely,\footnote{Since $f\ge0$, the square root of $f$ is well-defined.} the equation for $\f$, i.e.,
\begin{align}\label{sqfintro}
    X(\f)=0.
\end{align}
Notice that \eqref{sqfintro} can be expanded in the following form:
\begin{align}\label{sqfvlasovintro}
    \pr_t(\f)+\sum_{I=1}^3\frac{p^I}{p^0}t^{-\qIt}\pr_I(\f)-\sum_{I=1}^3\frac{\qIt}{t}p^I\pr_{p^I}(\f)=\cdots.
\end{align}
By a direct computation, for any $J,K\in\{1,2,3\}$ we have
\begin{align}\label{Intro:commute p pr p}
\left[p^J\pr_{p^K},\sum_{I=1}^3\frac{\qIt}{t}p^I\pr_{p^I}\right]=\frac{\qKt-\qJt}{t}p^J\pr_{p^K}.
\end{align}
We then commute \eqref{sqfvlasovintro} with $\pr^{\io_1}$ and $(p\pr_{p})^{\io_2}$. For $|\io_1|+|\io_2|\leq\kl$, we obtain
\begin{align}
\begin{split}\label{vlasovintrotop}
    &\pr_t\left(\f^{(\io_1,\io_2)}\right)+\frac{C_{\io_2}}{t}\f^{(\io_1,\io_2)}+\sum_{I=1}^3\frac{p^I}{p^0}t^{-\qIt}\pr_I\left(\f^{(\io_1,\io_2)}\right)
    -\sum_{I=1}^3\frac{\qIt}{t}p^I\pr_{p^I}\left(\f^{(\io_1,\io_2)}\right)=\cdots,
\end{split}
\end{align}
where $C_{\io_2}$ is a constant obeying $|C_{\io_2}|\leq |\io_2|\dq$. 

Multiplying \eqref{vlasovintrotop} by $2t^{2P}p^0\f^{(\io_1,\io_2)}$ (the choice of $P>0$ will be determined later), we deduce
    \begin{align*}
    &\pr_t\left(t^{P}p^0|\f^{(\io_1,\io_2)}|^2\right)+(2C_{\io_2}-P)t^{P-1}p^0|\f^{(\io_1,\io_2)}|^2\\
    +&\sum_{I=1}^3\frac{p^I}{p^0}t^{-\qIt}\pr_I\left(t^Pp^0|\f^{(\io_1,\io_2)}|^2\right)-\sum_{I=1}^3\frac{\qIt}{t}p^I\pr_{p^I}\left(t^{P}p^0|\f^{(\io_1,\io_2)}|^2\right)\\
    +&\sum_{I=1}^3\frac{\qIt}{t}\frac{(p^I)^2}{(p^0)^2}\left(t^Pp^0|\f^{(\io_1,\io_2)}|^2\right)=\cdots.
    \end{align*}
    Integrating it on $T\Si_t$, via integration by parts, we then obtain
    \begin{align*}
        \pr_t\left(t^{P}\int_{T\Si_t}p^0|\f^{(\io_1,\io_2)}|^2\right)+(2C_{\io_2}-P)t^{P-1}\int_{T\Si_t}p^0|\f^{(\io_1,\io_2)}|^2\geq\cdots.
    \end{align*}
Integrating from $t$ to $1$, for $P\coloneqq2A_*+2\dq+q_M+1>2C_{\io_2}$,\footnote{This can be ensured by taking $\frac{A_*}{\kl}\approx \dq$ and by applying the fact that $|C_{\io_2}|\leq|\io_2|\dq\leq\kl\dq$. See \eqref{dfkl} for the particular choice of $A_*$ and $\kl$.} we thus derive
\begin{align*}
t^{P}\int_{T\Si_t}p^0\left|\f^{(\io_1,\io_2)}\right|^2+\int_t^1s^{P-1}\int_{T\Si_s}p^0\left|\f^{(\io_1,\io_2)}\right|^2ds\les\ep_0^2.
\end{align*}
This implies the desired top-order estimate for $\sqrt{f}$, i.e.,
    \begin{align*}
        \max_{|\io_1|+|\io_2|\leq\kl}\int_{T\Si_t} p^0\left|\f^{(\io_1,\io_2)}\right|^2\les\frac{\ep_0^2}{t^P}=\frac{\ep_0^2}{t^{2A_*+2\dq+q_M+1}}.
    \end{align*}
   \begin{rk}\label{Rk:q}
        The above proof remains valid for the equation $X(f^q)=0$ with any power $q>0$. Here our above choice of $q=\frac12$ is consistent with the $\Hdot^{k_*}_x$-estimate for the energy-momentum tensor $T$. Specifically, to bound $\|T\|_{\Hdot^{k_*}_x}$ in terms of $\|f^q\|_{L^2_x L^2_p}$ and $\|(f^q)^{(k_*)}\|_{L^2_x L^2_p}$,\footnote{Here $\|\c\|_{L^2_x}$ and $\|\c\|_{L^2_x L^2_p}$ are $L^2$--norms on $\Si_t$ and on $T\Si_t$, and $f^{(k_*)}\coloneqq\{f^{(\io_1, \io_2)}: \ |\io_1|+|\io_2|\le k_* \}$.} by virtue of the heuristic that $p^0\sim t^{-q_M}$, by applying the Cauchy--Schwarz inequality, for $q\ge \frac12$ we obtain
        \begin{align*}
        \|T\|_{\Hdot^{k_*}_x}\les t^{-q_M}\|f^{1-q}\|_{L^{\infty}_xL^2_p}\|(f^q)^{(k_*)}\|_{L^2_xL^2_p}\les t^{-q_M} \|f^{q}\|^{\frac{1}{q}-1}_{L^{\infty}_xL^2_p}\|(f^q)^{(k_*)}\|_{L^2_xL^2_p}\c \sup\limits_{x\in\Si_t} V(t, x)^{1-\frac{1}{2q}}.
        \end{align*}
        Here $V(t, x)\coloneqq \int_{\supp_{P(t, x)} f} 1$ denotes the volume of $\supp f$ on the mass shell $P(t, x)$ and in principle it is hard to get the sharp bound for $V(t, x)$. By our method, via taking $q=\frac12$ we eliminate the need to control the size of $V(t, x)$ in deriving the estimate of $T$.
    \end{rk}
\begin{rk}
    Thanks to our new approach for controlling the Vlasov field, namely, our utilization of the weighted derivatives $p\pr_p$ and our choice of $q=\frac12$ as explained in \Cref{Rk:q}, we do not need to control the size of $p^{\mu}$ in $\supp_{P(t, x)} f$. At the level of initial data, we thus only need the $L^2$ control of $\f^{(k_*)}$ on the tangent bundle $T\Si_1$. This provides us the freedom that our initial perturbation of the Vlasov field is not necessarily restricted to a compact region of the mass shell $P(1,x)$.
\end{rk}
\begin{rk}
    The commutation formula \eqref{Intro:commute p pr p} suggests that  the top-order terms $(p^m \pr_{p^M})^{\kl}\f$ with $\widetilde{q_m}=\min\limits_{I} \widetilde{q_I}, \ \widetilde{q_M}=\max\limits_{I} \widetilde{q_I}$ are potentially the most singular. When estimating $(p^m \pr_{p^M})^{\kl}\f$, in order to absorb the linear term that comes from the commutation formula \eqref{Intro:commute p pr p}, we multiply the integrating factor $t^{2P}$ and we also need to impose that $P\ge 2\kl(\widetilde{q_M}-\widetilde{q_m})\approx 2\kl \de q$. Meanwhile, as the blow-up rate for the top-order energy estimates is $t^{-A_*+O(1)}$, we must require $P\le 2A_*+O(1)$, which implies that the parameters $A_*$ and $\kl$ need to obey
    \begin{equation*}
        A_*\ge \kl \de q+O(1).
    \end{equation*}
    Thus we cannot let $A_*/\kl$ to be arbitrarily small as the case in \cite{FRS}. In practice, we have to pick $A_*/\kl\approx \de q$, which indicates that our bounds of $\f$ are optimal for both lower-order estimates and top-order estimates. This is vastly different from the estimates for geometric quantities and for other matter fields.
\end{rk}
\subsubsection{Necessity of Strong Sub-Criticality for the Kasner Exponents}
Finally, we present the reason why we require the strong sub-critical condition as in \eqref{strongstability} to ensure the stability of Kasner solutions to the Einstein--Maxwell--scalar field--Vlasov system. We first observe that the evolution equation for $\pr k$ takes the form of
\begin{equation*}
    \pr_t(\pr\kc)+\frac{1}{t}(\pr\kc)=\pr T+\cdots.
\end{equation*}
Carrying out similar arguments as in \Cref{Subsec:EstexceptVlasov} and noting that $\pr T\sim \ep_0 t^{-1-q_M-\de q}$ (which is sharp as shown in \Cref{Subsec:Control of Vlasov Field}), we deduce
\begin{equation*}
\|t\pr \kc\|_{L^\infty(\Si_t)}\les\ep_0+\int_t^1s\|\pr T\|_{L^\infty(\Si_s)}ds +\cdots\les \ep_0+\ep_0\int_t^1 s^{-q_M-\de q}ds.
\end{equation*}
Therefore, to guarantee the integrability of $t^{-q_M-\de q}$ for $t\in(0, 1]$, we must require
\begin{equation*}
    q_M+\de q<1.
\end{equation*}
Recalling from \eqref{dfqintro} that 
\begin{equation*}
     q_M=\max_{I=1,2,3}\{\qIt\}+\si,\qquad\quad \dq=\max_{I,J=1,2,3}\{\qIt-\qJt\}+\si
\end{equation*}
and noting $\si>0$, we thus infer
\begin{equation*}
    1>\max_{I=1,2,3}\{\qIt\}+\max_{I,J=1,2,3}\{\qIt-\qJt\}=2\max_{I=1,2,3}\{\qIt\}-\min_{I=1,2,3}\{\qIt\}=\max_{I,J,K=1,2,3}\big\{\qIt+\qJt-\qKt\big\},
\end{equation*}
which is exactly the strong stability condition as in \eqref{strongstability}.

Although the necessity of the strong sub-critical condition has been explained in the preceding arguments, it is worth noting that proving stable Big Bang formation for all strong sub-critical Kasner solutions is still challenging. If the sharpness of the estimates for the Vlasov field is lost at any step, to ensure the AVTD behavior would require us to restrict the Kasner exponents  $\{\qIt\}_{I=1,2,3}$ to a proper subset of the strong sub-critical regime. In this paper, by leveraging new insights of the Vlasov equation and by employing the carefully designed weighted energy estimates, we ultimately succeed in establishing the desired result for the \textit{entire} strong sub-critical regime.
\subsection{Structure of the Paper}
\begin{itemize}
\item In Section \ref{secpreliminaries}, we introduce the geometry setup and derive the main equations. We also compute the precise values of geometric quantities for the exact Kasner solutions.
\item In Section \ref{secmain}, we state the main theorem and our bootstrap assumption.
\item In Section \ref{secbootass}, we prove the first consequences of the bootstrap assumption by using the interpolation inequality. These consequences are frequently used in the remaining sections of the paper.
\item In Section \ref{secelliptic}, we apply the maximum principle and derive energy estimates for the  elliptic equations to bound the lapse function.
\item In Section \ref{sectransport}, we control the lower-order $L^\infty$-norms of the geometric quantities and the matter fields by applying the transport estimates.
\item In Section \ref{secenergy}, we deduce the $L^2$-energy estimates to establish the top-order estimates of the geometric quantities and the matter fields.
\item In Section \ref{secconclusions}, we prove our main theorem and show the nonlinear stability of the Kasner Big Bang singularity for the Einstein--Maxwell--scalar field--Vlasov system.
\end{itemize}
\subsection{Acknowledgments}
The authors would like to thank Liam Urban for letting us know the helpful references \cite{Svedberg} and \cite{FU24}. XA is supported by MOE Tier 1 grants A-0004287-00-00, A-0008492-00-00 and MOE Tier 2 grant A-8000977-00-00.  TH acknowledges the support of NUS President Graduate Fellowship.
\section{Preliminaries}\label{secpreliminaries}
In this section, we introduce the geometric framework used to study perturbations of Kasner solutions based on the \textit{constant mean curvature} (CMC) foliation. This allows us to derive the corresponding reduced equations for the Einstein--Maxwell--scalar field--Vlasov system. This formalism is inspired by Fournodavlos--Rodnianski--Speck \cite{FRS}.
\subsection{Geometry Setup}
\subsubsection{Spacetime Metric}
Our spacetime $(\MM, \g)$ is equipped with the CMC-transported spatial coordinates on a slab $(t,x)\in(T,1]\times\mathbb{T}^3$ with $T\in [0, 1)$, where the spacetime metric takes the form of
\begin{equation}
    \g=-n^2dt\otimes dt+g_{ij}dx^i\otimes dx^j.
\end{equation}
Here $n>0$ is the lapse function,  $t$ is the time function and $g$ represents the induced (Riemannian) metric on the constant-time slice $\Si_t\coloneqq\{(s,x)\in(T,1]\times\mathbb{T}^3|\, s=t\}$. The spatial coordinates $\{x^i\}_{i=1,2,3}$ are said to be transported as $n^{-1}\pr_tx^i=0$, with $n^{-1}\pr_t$ being the future-directed unit normal to $\Si_t$. 

Since we frequently work with derivatives involving $n^{-1}(\pr_{t}, \pr_x) n$, for the sake of simplicity, instead of the lapse $n$, we introduce a modified lapse function.
\begin{align*}
    \vphi\coloneqq\log n.
\end{align*}
\subsubsection{The Orthonormal frame}\label{secortho}
Relative to $(t, x)$ coordinates on $(\MM, \g)$, in this paper we consider an associated orthonormal frame:
\begin{equation}
    e_0=n^{-1}\pr_t,\qquad e_I=e_I^c\pr_c,\qquad I=1,2,3,
\end{equation}
where $e_0$ is the future-directed unit normal to $\Si_t$, the spatial frame $\{e_I\}_{I=1,2,3}$ consists of $\Si_t$-tangent vectors that are normalized by
\begin{align*}
    g(e_I,e_J)=\de_{IJ},
\end{align*}
and the scalar functions $\{e^i_I\}_{i=1,2,3}$ are the components of $e_I$ relative to the transported spatial coordinates.

We now construct the desired spatial frame $\{e_I\}_{I=1,2,3}$ using the Fermi--Walker transport. Firstly, we pick an initial orthonormal spatial frame on $\Si_1$ with the help of the Gram--Schmidt process. Given this frame on $\Si_1$, we propagate it to the slab $(T,1]\times\mathbb{T}^3$ by solving the propagation equations:
\begin{align}\label{De0eI}
    \D_{e_0}e_I=(e_I\vphi)e_0,
\end{align}
where $\D$ is the Levi--Civita connection of $\g$. It is straightforward to check that
\begin{align*}
    \g(e_\a,e_\b)=\eta_{\a\b},\qquad \a,\b=0,1,2,3,\qquad e_I(t)=0,\qquad I=1,2,3
\end{align*}
with $\eta_{\a\b}=\mbox{diag}(-1, 1, 1, 1)$. We also have
\begin{align*}
    \g(\D_{e_\a}e_\b,e_\ga)=-\g(e_\b,\D_{e_\a}e_\ga).
\end{align*}
Note that a direct computation yields
\begin{align*}
\g\left(\D_{e_0}e_0,e_I\right)=-\g(e_0,\D_{e_0}e_I)=-\g(e_0,(e_I\vphi)e_0)=e_I\vphi.
\end{align*}
This further implies
\begin{align}\label{De0e0}
    \D_{e_0}e_0=(e_C\vphi)e_C.
\end{align}
\begin{rk}
The standard Fermi--Walker transport requires that
\begin{align*}
    \D_{e_0}e_I=(e_I\vphi)e_0-\g(e_I,e_0)(e_C\vphi)e_C.
\end{align*}
Compared to \eqref{De0eI} in this current paper, we omit the last term since we choose the initial frame $\{e_I\}$ to satisfy
\begin{equation*}
    \g(e_I, e_0)=0, \qquad\qquad \g(e_I, e_J)=\de_{IJ}.
\end{equation*}
Note this orthogonality property is preserved according to \eqref{De0eI}.
\end{rk}
\begin{rk}\label{rkconvention}
Throughout this paper, we use the Einstein summation for repeated indices $C,D,E$ and $c,i,j$. However, we will not use the Einstein summation convention for the indices $I,J,K$.
\end{rk}
\subsubsection{Second Fundamental Form and Curvature}
With the spatial frame $\{e_I\}_{I=1,2,3}$ as defined in \Cref{secortho}, we define the second fundamental form $k$ of $\Si_t$ as
\begin{align}\label{Def:k}
    k_{IJ}\coloneqq-\g(\D_{e_I}e_0,e_J).
\end{align}
This immediately implies
\begin{align}\label{DeIe0}
    \D_{e_I}e_0=-k_{IC}e_C.
\end{align}
We now normalize the time function $t$ according to the CMC condition:
\begin{align}\label{trk1t}
    \tr k\coloneqq k_{AA}=-\frac{1}{t}.
\end{align}
As a consequence, the condition \eqref{trk1t} leads to an elliptic equation for the lapse $n$, whose explicit form is given in \eqref{2.25}.

We also define the spatial connection coefficients of the frame $\{e_I\}_{I=1,2,3}$ as
\begin{align}\label{gaIJK}
    \ga_{IJK}\coloneqq\g(\D_{e_I}e_J,e_K)=g(\nab_{e_I}e_J,e_K),
\end{align}
with $\nab$ denoting the Levi--Civita connection of $g$. Using these definitions, we can write
\begin{align}\label{DeIeJ}
    \D_{e_I}e_J=-k_{IJ}e_0+\ga_{IJC}e_C,\qquad\quad\nab_{e_I}e_J=\ga_{IJC}e_C.
\end{align}
Differentiating the relation $\g(e_J,e_K)=\de_{JK}$ by $\D_{e_I}$, we deduce
\begin{align}\label{gaantisymmetric}
    \ga_{IJK}=-\ga_{IKJ}.
\end{align}
Furthermore,  we define the Riemann curvature $\R$, the Ricci curvature $\Ric$, and the scalar curvature $\R$, with respect to the spacetime metric $\g$ as follows:
\begin{align}
\begin{split}\label{dfcur}
    \R(e_\a,e_\b,e_\mu,e_\nu)&\coloneqq\g\left(\D_{e_\a}\D_{e_\b}e_\nu-\D_{e_\b}\D_{e_\a}e_\nu-\D_{[e_\a,e_\b]}e_\nu,e_\mu\right),\\
    \Ric(e_\a,e_\b)&\coloneqq\eta^{\mu\nu}\R(e_\a,e_\mu,e_\b,e_\nu),\\
    \R&\coloneqq\eta^{\mu\nu}\Ric(e_\mu,e_\nu).
\end{split}
\end{align}
And the curvature of tensors for the induced metric $g$ along $\Si_t$, namely its Riemann curvature $R$, Ricci curvature $\sRic$ and scalar curvature $R$, are analogous to the ones in \eqref{dfcur}.
\subsubsection{Mass Shell}\label{sssecmassshell}
Regarding the massless Vlasov field, the associated mass shell $P(t, x)\subseteq T_{(t,x)}\M$ is defined to be the set of future-pointing null vectors at the point $(t, x)\in \MM$, i.e.,
\begin{equation*}
    P(t,x)=\{p\in T_{(t,x)}\M: \ \g(p, p)=0\}.
\end{equation*}
And we define $P=\cup_{(t, x)\in \M} P(t, x)$.

For any vector $p\in T_{(t,x)}\M$, using the aforementioned orthonormal frame $\{e_\mu\}_{\mu=0,1,2,3}$, we can express
\begin{align*}
    p=p^\mu e_\mu.
\end{align*}
Thus, for any  $p\in P$ the following relation holds
\begin{align*}
    (p^0)^2=\sum_{I=1}^3(p^I)^2.
\end{align*}
We proceed to define the associated volume form on the mass shell $P(t, x)$, which is a null hypersurface in $T_{(t,x)\M}$. Utilizing the coordinates $(p^{\mu})$ with $\mu=0,1,2,3$, the spacetime metric on $\M$ induces a metric on $T_{(t,x)}\M$:
\begin{align*}
    -(dp^0)^2+\sum_{I=1}^3(dp^I)^2,
\end{align*}
which in turn defines a volume form on $T_{(t,x)}\M$:
\begin{align*}
    dp^0\wedge dp^1\wedge dp^2\wedge dp^3.
\end{align*}
Define the function $\La: T_{(t,x)}\M\to \mathbb{R}$
\begin{align*}
    \La(X)=\g(X,X)=-(p^0)^2+\sum_{I=1}^3(p^I)^2\qquad \text{for all} \quad X=p^\mu e_\mu\in T_{(t,x)}\M,
\end{align*}
which measures the length of the vector $X$. Then the canonical one-form normal to $P_{(t,x)}$ can be defined as the differential of $\La$:
\begin{align*}
    \La(X)=\g(X,X)=-(p^0)^2+\sum_{I=1}^3(p^I)^2,\qquad X=p^\mu e_\mu\in T_{(t,x)}\M.
\end{align*}
Hence, we define the volume form on $P(t,x)$ via
\begin{align*}
    \dvol\coloneqq\frac{1}{p^0}dp^1\wedge dp^2\wedge dp^3.
\end{align*}
This is the unique volume form on $P_{(t,x)}$ compatible with $-\frac{1}{2}d\La$ obeying
\begin{align*}
    -\frac{1}{2}d\La\wedge \left(\frac{1}{p^0}dp^1\wedge dp^2\wedge dp^3\right)=dp^0\wedge dp^1\wedge dp^2\wedge dp^3.
\end{align*}
With this choice, the energy momentum tensor as in \eqref{dfT} therefore takes the form
\begin{align}\label{expressionT}
    T_{\mu\nu}(t,x)=\int_{\mathbb{R}^3} \frac{f(t,x,p)p_{\mu}p_{\nu}}{p^0}dp^1dp^2dp^3.
\end{align}
\subsection{Main Reduced Equations}
In this subsection, we derive the differential equations  suited for analyzing perturbations of generalized Kasner solutions. Specifically, we have the following reduced Einstein--Maxwell--scalar field--Vlasov system relative to the aforementioned CMC-transported spatial coordinates $(t,x)$ and a Fermi--Walker transported orthonormal frame $\{e_{\mu}\}_{\mu=0,1,2,3}$.
\begin{prop}\label{basicequations}
The Einstein--Maxwell--scalar field--Vlasov system \eqref{EMSV}, \eqref{Wave}--\eqref{Vlasov} is equivalent to the below reduced system of equations for $k, \ga, e^i_I, \psi, F, f, \vphi$:
\begin{itemize}
\item \textbf{The evolution equations for $k,\ga$:}
\begin{align}
e_0(k_{IJ})&=-e_I(e_J\vphi)-(e_I\vphi)e_J\vphi+e_C(\ga_{IJC})-e_I(\ga_{CJC})-t^{-1}k_{IJ}+\ga_{IJC}e_C\vphi\nonumber\\
&-\ga_{DIC}\ga_{CJD}-\ga_{DDC}\ga_{IJC}-(e_I\psi)e_J\psi-2\left(F_{I\a}{F_J}^{\a}-\frac{1}{4}\de_{IJ}F_{\a\b}F^{\a\b}\right)-T_{IJ},\label{2.22a}\\
e_0(\ga_{IJK})&=e_K(k_{JI})-e_J(k_{KI})-\ga_{KJC}k_{CI}-\ga_{KIC}k_{JC}+\ga_{JKC}k_{IC}+\ga_{JIC}k_{KC}+k_{IC}\ga_{CJK}\nonumber\\
&-(e_J\vphi)k_{IK}+(e_K\vphi)k_{IJ}.\label{2.22b}
\end{align}
\item \textbf{The evolution equation for $e_I^i$:}
\begin{align}\label{2.23}
    e_0e_I^i=k_{IC}e_C^i.
\end{align}
\item \textbf{The wave equation for the scalar field $\psi$:}
\begin{align}\label{2.24}
    e_0(e_0\psi)=e_C(e_C\psi)-t^{-1}e_0\psi+(e_C\vphi)e_C\psi-\ga_{CCD}e_D\psi.
\end{align}
\item \textbf{The Maxwell equations for the electromagnetic field $F$:}
\begin{align}
e_0(F_{0I})+t^{-1}F_{0I}+k_{CI}F_{0C}&=e_C(F_{CI})+(e_C\vphi)F_{CI}-\ga_{CCD}F_{DI}-\ga_{CIB}F_{CD},\label{Maxwell1}\\
e_0(F_{IJ})+k_{IC}F_{JC}+k_{JC}F_{CI}
&=e_I(F_{0J})-e_J(F_{0I})-(\ga_{IJC}+\ga_{JCI})F_{0C}\nonumber \\
&+(e_I\vphi)F_{0J}+(e_J\vphi)F_{I0},\label{Maxwell2}\\
e_I(F_{JB})+e_J(F_{BI})+e_B(F_{IJ})&=(\ga_{IJC}+\ga_{JCI})F_{CB}+(\ga_{ICB}+\ga_{KIC})F_{CJ}\nonumber\\
&+(\ga_{JKC}+\ga_{BCJ})F_{CI}.\label{Maxwell3}
\end{align}
\item \textbf{The Vlasov equation for the distribution function $f$:}
\begin{align}\label{Vlasoveq}
e_0(f)+\frac{p^C}{p^0}e_C(f)-\frac{p^Dp^E}{p^0}\ga_{DEC}\pr_{p^C}(f)+p^Dk_{DC}\pr_{p^C}(f)-p^0e_C(\vphi)\pr_{p^C}(f)=0.
\end{align}
\item \textbf{The elliptic equation for the lapse $\vphi$:}
\begin{align}
\begin{split}\label{2.25}
e_C(e_C\vphi)-\frac{\vphi}{t^2}&=-(e_C\vphi)e_C\vphi+2e_C(\ga_{DDC})-(e_C\psi)e_C\psi+\frac{1-\vphi-e^{-\vphi}}{t^2}\\
&-\ga_{DEC}\ga_{CED}-\ga_{DDC}\ga_{EEC}+2\left(F_{0C}F_{0C}-\frac{1}{4}F_{\a\b}F^{\a\b}\right)-T_{00}.
\end{split}
\end{align}
\item \textbf{The Hamiltonian equation:}
\begin{align}
\begin{split}\label{2.26a}
&\quad\; 2e_C(\ga_{DDC})-\ga_{CDE}\ga_{EDC}-\ga_{CCD}\ga_{EED}-k_{CD}k_{CD}+t^{-2}\\
&=(e_0\psi)^2+(e_C\psi)e_C\psi+4\left(F_{0C}{F_{0C}}+\frac{1}{4}F_{\a\b}F^{\a\b}\right)+2T_{00},
\end{split}
\end{align}
and \textbf{the momentum constraint equation:}
\begin{equation}
e_Ck_{CI}=\ga_{CCD}k_{ID}+\ga_{CID}k_{CD}-(e_0\psi)e_I\psi-2F_{0C}F_{IC}-T_{0I}.\label{2.26b}
\end{equation}
\end{itemize}
\end{prop}
\begin{proof}
This proof is inspired by Section 2.1.5 of \cite{FRS}. Throughout this proof, we frequently employ the relations \eqref{De0eI}, \eqref{De0e0}, \eqref{DeIe0} and \eqref{DeIeJ} without further reference. We start with deriving the equation of $k$. A straightforward calculation gives
\begin{equation}\label{Eqn:R0I0J}
\begin{aligned}
&\quad\;\R(e_0,e_I,e_0,e_J)\\
&=\g\left(\D_{e_0}\D_{e_I}e_J-\D_{e_I}\D_{e_0}e_J-\D_{[e_0,e_I]}e_J,e_0\right)\\
&=\g\left(\D_{e_0}\D_{e_I}e_J-\D_{e_I}\D_{e_0}e_J-\D_{(e_I\vphi)e_0+k_{IC}e_C}e_J,e_0\right)\\
&=e_0\g(\D_{e_I}e_J,e_0)-\g(\D_{e_0}e_0,\D_{e_I}e_J)-e_I\g(\D_{e_0}e_J,e_0)+\g(\D_{e_0}e_J,\D_{e_I}e_0)+(e_I\vphi)e_J\vphi-k_{IC}k_{CJ}\\
&=e_0(k_{IJ})-(e_C\vphi)\ga_{IJC}+e_I(e_J\vphi)+(e_I\vphi)e_J\vphi-k_{IC}k_{CJ}.
\end{aligned}
\end{equation}
We proceed to rewrite \eqref{Eqn:R0I0J} with the assistance of the Einstein field equations. Recall the Gauss equation, namely,
\begin{align}\label{Gauss}
    \R(e_C,e_I,e_D,e_J)=R(e_C,e_I,e_D,e_J)+k_{CD}k_{IJ}-k_{CJ}k_{ID}.
\end{align}
Combining with \eqref{EMS} and noting $\tr k=k_{CC}=-t^{-1}$, we obtain
\begin{align}
\begin{split}\label{2.28}
    \R(e_0,e_I,e_0,e_J)&=-\Ric(e_I,e_J)+\R(e_C,e_I,e_C,e_J)\\
    &=-(e_I\psi)e_J\psi-2\left(F_{I\a}{F_J}^{\a}-\frac{1}{4}\de_{IJ}F_{\a\b}F^{\a\b}\right)-T_{IJ}\\
    &+\sRic(e_I,e_J)-t^{-1}k_{IJ}-k_{IC}k_{JC}.
\end{split}
\end{align}
Then we compute the components of the Ricci tensor of $g$ adapted to the spatial frame $\{e_I\}_{I=1,2,3}$ and derive that
\begin{align}
\begin{split}\label{sRiceIeJ}
\sRic(e_I,e_J)&=R(e_C,e_I,e_C,e_J)\\
&=g\left(\nab_{e_C}\nab_{e_I}e_J-\nab_{e_I}\nab_{e_C}e_J-\nab_{[e_C,e_I]}e_J,e_C\right)\\
&=e_Cg(\nab_{e_I}e_J,e_C)-g(\nab_{e_I}e_J,\nab_{e_C}e_C)-e_Ig(\nab_{e_C}e_J,e_C)+g(\nab_{e_C}e_J,\nab_{e_I}e_C)\\
&-g(\nab_{\ga_{CID}e_D-\ga_{ICD}e_D}e_J,e_C)\\
&=e_C(\ga_{IJC})-\ga_{IJD}\ga_{CCD}-e_I(\ga_{CJC})+\ga_{CJD}\ga_{ICD}-\ga_{CID}\ga_{DJC}+\ga_{ICD}\ga_{DJC}\\
&=e_C(\ga_{IJC})-\ga_{IJD}\ga_{CCD}-e_I(\ga_{CJC})-\ga_{CID}\ga_{DJC}.
\end{split}
\end{align}
Here we utilize \eqref{gaantisymmetric}. Substituting it into \eqref{2.28} and comparing with \eqref{Eqn:R0I0J}, we hence deduce
\begin{align*}
e_0(k_{IJ})+e_I(e_J\vphi)+(e_I\vphi)e_J\vphi&=-(e_I\psi)e_J\psi-2\left(F_{I\a}{F_J}^{\a}-\frac{1}{4}\de_{IJ}F_{\a\b}F^{\a\b}\right)-T_{IJ}+(e_C\vphi)\ga_{IJC}\\
&+e_C(\ga_{IJC})-\ga_{IJD}\ga_{CCD}-e_I(\ga_{CJC})-\ga_{CID}\ga_{DJC}-t^{-1}k_{IJ},
\end{align*}
which implies \eqref{2.22a}. 
\vspace{3mm}

Next we turn to derive the equation for $\ga$. Observe that
\begin{align*}
e_0(\ga_{IJK})=&\D_{e_0}\g(\D_{e_I}e_J,e_K)\\
=&\g(\D_{e_0}\D_{e_I}e_J,e_K)+\g(\D_{e_I}e_J,\D_{e_0}e_K)\\
=&\R(e_0,e_I,e_K,e_J)+\g(\D_{e_I}\D_{e_0}e_J,e_K)+\g(\D_{[e_0,e_I]}e_J,e_K)+\g(\D_{e_I}e_J,\D_{e_0}e_K)\\
=&\R(e_0,e_I,e_K,e_J)+e_I\g(\D_{e_0}e_J,e_K)-\g(\D_{e_0}e_J,\D_{e_I}e_K)\\
&+\g(\D_{\D_{e_0}e_I-\D_{e_I}e_0}e_J,e_K)+\g(\D_{e_I}e_J,\D_{e_0}e_K)\\
=&\R(e_K,e_J,e_0,e_I)-\g((e_J\vphi)e_0,-k_{IK}e_0+\ga_{IKC}e_C)\\
&+\g(\D_{(e_I\vphi)e_0+k_{IC}e_C}e_J,e_K)+\g(-k_{IJ}e_0+\ga_{IJC}e_C,(e_K\vphi)e_0)\\
=&\R(e_K,e_J,e_0,e_I)-(e_J\vphi)k_{IK}+(e_K\vphi)k_{IJ}+k_{IC}\ga_{CJK}.
\end{align*}
Recalling the Codazzi equations, that is,
\begin{align}\label{Codazzi}
    \nab_Kk_{JI}-\nab_Jk_{KI}=\R(e_K,e_J,e_0,e_I),
\end{align}
we then infer that
\begin{align*}
    e_0(\ga_{IJK})&=\nab_Kk_{JI}-\nab_Jk_{KI}-(e_J\vphi)k_{IK}+(e_K\vphi)k_{IJ}+k_{IC}\ga_{CJK}\\
    &=e_K(k_{JI})-\ga_{KJC}k_{CI}-\ga_{KIC}k_{JC}-e_J(k_{KI})+\ga_{JKC}k_{IC}+\ga_{JIC}k_{KC}\\
    &-(e_J\vphi)k_{IK}+(e_K\vphi)k_{IJ}+k_{IC}\ga_{CJK}.
\end{align*}
This gives \eqref{2.22b}. 
\vspace{3mm}

Notice that the transport equation of $e_I^i$, namely \eqref{2.23} directly follows from
\begin{align}\label{eq2.35}
(\pr_te_I^c)\pr_c=[\pr_t,e_I^c\pr_c]=[\pr_t,e_I]=\D_{\pr_t}e_I-\D_{e_I}(ne_0)=nk_{IC}e_C=nk_{IC}e_C^c\pr_c.
\end{align} 
Regarding the transport equation of $e_I^i$, by contracting \eqref{2.22a} with $\g^{IJ}$, we obtain 
\begin{align*}
    e_0\tr k=&-e_C(e_C\vphi)-(e_C\vphi)e_C\vphi+e_C(\ga_{DDC})-e_D(\ga_{CDC})-\frac{1}{t}\tr k+\ga_{DDC}e_C\vphi\\
    &-\ga_{DEC}\ga_{CED}-\ga_{DDC}\ga_{EEC}-(e_C\psi)e_C\psi-2\left(F_{C\a}{F_C}^{\a}-\frac{3}{4}F_{\a\b}F^{\a\b}\right)-T_{CC}.
\end{align*}
By virtue of the CMC condition $\tr k=-\frac{1}{t}$ and the fact that $T_{00}=T_{CC}$,\footnote{Notice that 
\begin{equation*}
    -T_{00}+T_{CC}=\g^{\mu\nu}T_{\mu\nu}=\int_{P(t, x)} f\g^{\mu\nu}p_{\mu}p_{\nu}\dvol=0.
\end{equation*}} we infer
\begin{align*}
     e_C(e_C\vphi)+(e_C\vphi)e_C\vphi&=\frac{1-n^{-1}}{t^2}+e_C(\ga_{DDC})-e_D(\ga_{CDC})+\ga_{DDC}e_C\vphi-\ga_{DEC}\ga_{CED}-\ga_{DDC}\ga_{EEC}\\
    &-(e_C\psi)e_C\psi-2\left(F_{\b\a}{F_\b}^{\a}-F_{0\a}{F_0}^\a-\frac{3}{4}F_{\a\b}F^{\a\b}\right)-T_{00},
\end{align*}
which implies \eqref{2.25}. 
\vspace{3mm}

To derive the reduced Hamiltonian equation \eqref{2.26a} and the reduced momentum constraint equation \eqref{2.26b}, we appeal to the following Hamiltonian constraint equations and the momentum constraint equation along the spacelike hypersurface $\Si_t$:
\begin{align}
    R-|k|^2+(\tr k)^2=&(e_0\psi)^2+(e_C\psi)e_C\psi+4\left(F_{0\a}{F_{0}}^\a+\frac{1}{4}F_{\a\b}F^{\a\b}\right)+2T_{00}, \label{Hamiltonian}\\
    \nab_Ck_{IC}-\nab_Ik_{CC}=&-(e_0\psi)e_I\psi-2F_{0C}F_{IC}-T_{0I}. \label{momentum}
\end{align}
Incorporating \eqref{Hamiltonian} with \eqref{sRiceIeJ}, we have
\begin{align*}
    &e_C(\ga_{DDC})-\ga_{EED}\ga_{CCD}-e_D(\ga_{CDC})-\ga_{CED}\ga_{DEC}-k_{CD}k_{CD}+t^{-2}\\
    =&(e_0\psi)^2+(e_C\psi)e_C\psi+4\left(F_{0\a}{F_{0}}^\a+\frac{1}{4}F_{\a\b}F^{\a\b}\right)+2T_{00},
\end{align*}
which implies \eqref{2.26a}. 
\vspace{3mm}

Meanwhile, inserting the following identities
\begin{align*}
    \nab_Ck_{IC}=e_Ck_{IC}-\ga_{CID}k_{CD}-\ga_{CCD}k_{ID},\qquad \nab_Ik_{CC}=e_Ik_{CC}=0.
\end{align*}
into \eqref{momentum}, we deduce \eqref{2.26b}. 
\vspace{3mm}

It remains to derive the reduced equations for the matter fields.  In view of \eqref{Wave} we have
\begin{align*}
    -e_0(e_0\psi)+e_C(e_C\psi)&=\bo_\g\psi-(\D_{e_0}e_0)^\a\D_\a\psi+(\D_{e_C}e_C)^\a\D_\a\psi\\
    &=-(e_C\vphi)e_C\psi+t^{-1}e_0\psi+\ga_{CCD}e_D\psi,
\end{align*}
which gives the reduced wave equation for $\psi$ as in \eqref{2.24}. 

As for the Maxwell equations \eqref{Maxwell1}--\eqref{Maxwell3}, employing  \eqref{Maxwell} we get
\begin{align*}
    -\D_{e_0}F_{0I}+\D_{e_C}F_{CI}=0,
\end{align*}
Expanding its $I$-th component, we find
\begin{align*}
    -e_0(F_{0I})+F(\D_{e_0}e_0,e_I)+F(e_0,\D_{e_0}e_I)+e_C(F_{CI})-F(\D_{e_C}e_C,e_I)-F(e_C,\D_{e_C}e_I)=0.
\end{align*}
This consequently renders
\begin{align*}
    -e_0(F_{0I})+(e_C\vphi)F_{CI}+e_C(F_{CI})-t^{-1}F_{0I}-\ga_{CCD}F_{DI}+k_{CI}F_{C0}-\ga_{CID}F_{CD}=0,
\end{align*}
which is equivalent to \eqref{Maxwell1}. 
\vspace{3mm}

From \eqref{Maxwell} we also have
\begin{align}\label{dF=0}
    \D_{\a}F_{\b\ga}+\D_{\b}F_{\ga\a}+\D_{\ga}F_{\a\b}=0.
\end{align}
Setting $\a=0$, $\b=I$ and $\ga=J$ in \eqref{dF=0}, we deduce
\begin{align*}
    \D_{e_0}F_{IJ}+\D_{e_I}F_{J0}+\D_{e_J}F_{0I}=0,
\end{align*}
which yields
\begin{align*}
    &e_0(F_{IJ})-F(\D_{e_0}e_I,e_J)-F(e_I,\D_{e_0}e_J)+e_I(F_{J0})-F(\D_{e_I}e_J,e_0)-F(e_J,\D_{e_I}e_0)\\
    +&e_J(F_{0I})-F(\D_{e_J}e_0,e_I)-F(e_0,\D_{e_J}e_I)=0.
\end{align*}
Then we obtain
\begin{align*}
    e_0(F_{IJ})+k_{IC}F_{JC}+k_{JC}F_{CI}=e_I(F_{0J})-e_J(F_{0I})-(\ga_{IJC}+\ga_{JCI})F_{0C}+(e_I\vphi)F_{0J}+(e_J\vphi)F_{I0},
\end{align*}
which corresponds to \eqref{Maxwell2}. 
\vspace{3mm}

Finally, choosing mutually distinct indices $I,J,K$ for \eqref{dF=0}, we have
\begin{align*}
    \D_{e_I}F_{JK}+\D_{e_J}F_{KI}+\D_{e_K}F_{IJ}=0.
\end{align*}
It then follows
\begin{align*}
&e_I(F_{JK})+e_J(F_{KI})+e_K(F_{IJ})\\
=&(\ga_{IJC}+\ga_{JCI})F_{CK}+(\ga_{ICK}+\ga_{KIC})F_{CJ}
+(\ga_{JKC}+\ga_{KCJ})F_{CI},
\end{align*}
which implies \eqref{Maxwell3}.
\vspace{3mm}

Now consider the Vlasov equation \eqref{Vlasov}:
\begin{align*}
    p^\mu e_\mu(f)+\frac{dp^\mu}{ds}\pr_{p^\mu}(f)=0,
\end{align*}
where $s$ denotes the affine parameter along the geodesic spray $X$. Restricting it on the mass shell $P_{(t,x)}$, we deduce
\begin{align*}
    p^\mu e_\mu(f)+\frac{dp^C}{ds}\pr_{p^C}(f)=0.
\end{align*}
Employing the following geodesic equation of $p^I$ along the geodesic spray $X$:
\begin{align*}
    \frac{dp^I}{ds}+p^\mu p^\nu\g(\D_{e_\mu}e_\nu,e_I)=0,
\end{align*}
and noting that $p^0=\frac{dt}{ds}$, we thus derive
\begin{align*}
    e_0(f)+\frac{p^C}{p^0}e_C(f)-\frac{p^\mu p^\nu}{p^0}\g(\D_{e_\mu}e_\nu,e_C)\pr_{p^C}(f)=0.
\end{align*}
By expanding the expression of $\g(\D_{e_\mu}e_\nu,e_C)$, we then arrive at
\begin{align*}
    e_0(f)+\frac{p^C}{p^0}e_C(f)-\frac{p^Dp^E}{p^0}\ga_{DEC}\pr_{p^C}(f)+p^Dk_{DC}\pr_{p^C}(f)-p^0(e_C\vphi)\pr_{p^C}(f)=0,
\end{align*}
which gives the reduced Vlasov equation \eqref{Vlasoveq}. This concludes the proof of Proposition \ref{basicequations}.
\end{proof}
\subsection{Kasner Variables}
In the following proposition, we derive the corresponding reduced variables for the exact generalized Kasner solution as in \eqref{kasnerg} and \eqref{kasnerq}.
\begin{prop}\label{kasnerquantities}
The reduced variables of the generalized Kasner solution in \eqref{kasnerg}, \eqref{kasnerq} read
\begin{align*}
\nt&=1,\qquad\qquad\et_I^i=t^{-\widetilde{q_I}}\de_I^i,\qquad\quad\kt_{IJ}=-\frac{\qIt}{t}\de_{IJ},\qquad\qquad\gat_{IJK}=0,\\
\psit&=\widetilde{B}\log t,\quad\;\;\;\Ft=0,\qquad\qquad\quad\quad\widetilde{f}=0.
\end{align*}
\end{prop}
\begin{proof}
A direct computation implies
\begin{align*}
    \kt_{IJ}&=-\gt(\D_{\et_I}\et_0,\et_J)=-t^{-\qIt-\qJt}\gt(\D_{\pr_I}\pr_t,\pr_J)\\
    &=-\frac{1}{2}t^{-\qIt-\qJt}\left(\frac{\pr\gt_{IJ}}{\pr t}+\frac{\pr\gt_{tJ}}{\pr x^I}-\frac{\pr\gt_{It}}{\pr x^J}\right)\\
    &=-\frac{1}{2}t^{-\qIt-\qJt}\pr_t\left(t^{2\qIt}\de_{IJ}\right)=-\frac{\qIt}{t}\de_{IJ}.
\end{align*}
Moreover, we have
\begin{align*}
    \gat_{IJK}&=\gt(\D_{\et_I}\et_J,\et_K)=t^{-\qIt-\qJt-\qKt}\gt(\D_{\pr_I}\pr_J,\pr_K)\\
    &=\frac{1}{2}t^{-\qIt-\qJt-\qKt}\left(\frac{\pr\gt_{IK}}{\pr x^J}+\frac{\pr\gt_{JK}}{\pr x^I}-\frac{\pr\gt_{IJ}}{\pr x^K}\right)=0.
\end{align*}
The remaining  equalities follow readily from \eqref{kasnerg}. This completes the proof of Proposition \ref{kasnerquantities}.
\end{proof}
\section{Main Theorem}\label{secmain}
In this paper, we aim to establish our main theorems by a continuity argument for the reduced system in Proposition \ref{basicequations}. We begin by imposing bootstrap assumptions for various norms of the perturbed solution over a time interval $(T_*,1)$ with bootstrap time $T_*\in(0,1)$. The core step is to derive a priori estimates for the perturbed solution. We hope these estimates lead to a strict improvement of the bootstrap assumptions on $(T_*,1]$, which allows us, via standard arguments, to extend the perturbed solution beyond the bootstrap time interval $(T_*,1]$. Iteratively, we can show that this solution must exist on $(0,1]\times \mathbb{T}^3$ and satisfy the desired a priori estimates on $(0,1]$.  As a consequence, using the existence result and the precise controls of reduced variables as in the a priori estimates, we are able to prove various curvature blow-ups as $t\to 0$. This reveals the quantitative information of the Big Bang singularity, and the details on this are given in Section \ref{secconclusions}.
\subsection{Sobolev Norms}\label{secSobolev}
Given a scalar function $v$ on $\Si_t$, we define its $L^2$-norm as
\begin{align}\label{dfL2}
    \|v\|_{L^2(\Si_t)}^2\coloneqq\int_{\Tt^3}v^2(t,x)dx^1dx^2dx^3,
\end{align}
where $dx^1dx^2dx^3$ denotes the Euclidean volume form on $\Si_t$.

Similarly, for a scalar function $h$ on $T\Si_t$, we define its $L^2$--norm by
\begin{align}\label{dfL2f}
\|h\|_{L^2(T\Si_t)}^2\coloneqq\int_{\Tt^3\times\mathbb{R}^3}h^2(t,x,p)dx^1dx^2dx^3dp^1dp^2dp^3.
\end{align}
We now introduce the following schematic differential operators:
\begin{align*}
    \pr\coloneqq\{\pr_{x^1},\pr_{x^2},\pr_{x^3}\},\qquad\qquad p\pr_p\coloneqq\bigcup_{i,j=1,2,3}\left\{p^i\pr_{p^j}\right\}.
\end{align*}
For any triplet $\io\coloneqq(\io^1,\io^2,\io^3)$, we define
\begin{align*}
    \pr^\io\coloneqq\pr_{x^1}^{\io^1}\pr_{x^2}^{\io^2}\pr_{x^3}^{\io^3}.
\end{align*}
And for any $3\times 3$ matrix $\io\coloneqq(\io^{ij})_{i,j=1,2,3}$, we define
\begin{align*}
    (p\pr_p)^{\io}\coloneqq\prod_{i,j=1,2,3}(p^i\pr_{p^j})^{\io^{ij}}.
\end{align*}
It is also convenient to introduce the following conventions:
\begin{align}\label{io1io2df}
h^{(\io_1)}\coloneqq\pr^{\io_1}h,\qquad\qquad h^{(\io_1,\io_2)}\coloneqq\pr^{\io_1}(p\pr_p)^{\io_2}h.
\end{align}
where $\io_1$ is a triplet and $\io_2$ is a $3\times 3$ matrix. With the above notations, we define the standard $H^M(\Si_t)$, $\dot{H}^M(\Si_t)$, $W^{M,\infty}(\Si_t)$ and $\dot{W}^{M,\infty}(\Si_t)$ norms of a scalar function $v$ as follows:
\begin{align*}
\|v\|_{H^M(\Si_t)}&\coloneqq\max_{|\io|\leq M}\|\pr^\io v\|_{L^2(\Si_t)},\qquad\qquad\|v\|_{\dot{H}^M(\Si_t)}\coloneqq\max_{|\io|=M}\|\pr^\iota v\|_{L^2(\Si_t)},\\
\|v\|_{W^{M,\infty}(\Si_t)}&\coloneqq\max_{|\io|\leq M}\|\pr^\io v\|_{L^\infty(\Si_t)},\qquad\;\;\|v\|_{\dot{W}^{M,\infty}(\Si_t)}\coloneqq\max_{|\io|=M}\|\pr^\iota v\|_{L^\infty(\Si_t)}.
\end{align*}
Furthermore, if $v$ is a $\Si_t$--tangent tensorfield, then we regard it as the vector-valued function with components relative to the spatial frame $\{e_{I}\}_{I=1,2,3}$. And we define its $L^2(\Si_t)$, $H^M(\Si_t)$, $\dot{H}^M(\Si_t)$ and $W^{M,\infty}(\Si_t)$ and $\dot{W}^{M,\infty}(\Si_t)$ norms, by summing over all frame indices.
\subsection{Choice of Parameters}\label{seckey}
This section is devoted to choosing the key parameters that are involved in our bootstrap assumptions. Notice that the strong sub-critical condition \eqref{strongstability} gives
\begin{align}\label{strongq}
0<\max_{I,J,K=1,2,3}\big\{1-\qIt,\qIt+\qJt-\qKt\big\}<1,\qquad \max_{I=1,2,3}\{\qIt\}<\frac{3}{5}.
\end{align}
We first define $q,\si\in (0, 1)$ by the following relation:
\begin{align}\label{dfqsi}
    \max_{I,J,K=1,2,3}\big\{1-\qIt,\qIt+\qJt-\qKt\big\}=q-4\si=1-6\si.
\end{align}
We then set
\begin{align}\label{dfqM}
    q_M\coloneqq\max_{I=1,2,3}\{\qIt\}+\si,\qquad \de q\coloneqq \max_{I,J=1,2,3}\{\qIt-\qJt\}+\si.
\end{align}
By definition this implies
\begin{align}\label{AVTD}
    q_M+\de q=\max_{I=1,2,3}\{\qIt\}+\max_{I,J=1,2,3}\{\qIt-\qJt\}+2\si\leq q-2\si.
\end{align}

Next, we select parameters $A_*, \kl$ such that
\begin{align}\label{dfkl}
    \max_{I,J=1,2,3}\{\qIt-\qJt\}<\frac{A_*-5}{\kl+5}<\frac{A_*+5}{\kl-5}<\de q.
\end{align}
This can be achieved by taking $A_*$ large enough and then choosing $\kl\in\mathbb{N}$. It is worth mentioning that $A_*$ and $\kl$ represent respectively the blow-up rate in $t$ (of order $t^{-A_*}$) and the top-order regularity of the perturbed solution in $L^2$.
\begin{rk}
    Throughout this paper, we denote $A\les B$ for $A\leq CB$ with $C$ being a constant that depends only on $q$, $q_M$, $\de q$, and $\si$. Moreover, $A\ll B$ stands for $CA<B$ with $C$ being the largest universal constant among all the constants involved in the proof by $\les$. Similarly, we denote $A\les_*B$ for $A\leq C_*B$, where $C_*$ is a constant depending on $A_*$, $\kl$ and $C$. Moreover, $A\ll_*B$ represents that $C_*A<B$, where $C_*$ is the largest constant involved in the proof through $\les_*$.
\end{rk}
Finally,  we pick two smallness constants $\ep_0,\ep>0$ satisfying
\begin{align*}
    \ep_0\ll_*\ep\ll_*\frac{1}{\kl}<\frac{1}{A_*}\ll\si.
\end{align*}
Here $\ep_0$ measures the size of the initial perturbation and $\ep$ corresponds to the bootstrap bounds that will be improved.
\subsection{Auxiliary Function \texorpdfstring{$\TT$}{}}\label{Subsec:defTT}
To control the energy-momentum tensor $T$ associated with the distribution function $f=f(t,x,p)$, we introduce the following function of $(t,x)$:
    \begin{equation}\label{dfTT}
        \TT(t,x)\coloneqq \left\|(p^0)^\frac{1}{2}\f(t,x,p)\right\|_{L^2_p(\mathbb{R}^3)}.
    \end{equation}
Moreover, we denote\footnote{Since $\TT$ is a function only depends on $(t,x)$, the notation $\TT^{(\io_1,\io_2)}$ will not cause confusion with that in \eqref{io1io2df}.}
\begin{equation}\label{dfTTiota}
    \TT^{(\io_1,\io_2)}(t,x)\coloneqq \left\|(p^0)^\frac{1}{2}\f^{(\io_1,\io_2)}(t,x,p)\right\|_{L^2_p(\mathbb{R}^3)},
\end{equation}
where $\f^{(\io_1,\io_2)}$ is defined according to \eqref{io1io2df}. We also define
\begin{align}\label{dfTTk}
    \TT^{(k)}(t,x)\coloneqq \max_{|\io_1|+|\io_2|\leq k}\TT^{(\io_1,\io_2)}(t,x).
\end{align}
\begin{rk}\label{rkTT}
    As an immediate consequence of \eqref{dfTT}, we obtain
    \begin{align*}
        |T_{\mu\nu}(t,x)|\leq T_{00}(t,x)=\int_{\mathbb{R}^3}fp^0dp^1dp^2dp^3=\TT^2(t,x),\qquad \forall \mu,\nu=0,1,2,3.
    \end{align*}
\end{rk}
\subsection{Fundamental Norms}
In this subsection, we introduce the fundamental $t$-weighted norms for the reduced quantities that we work with throughout the rest of the paper. These norms also indicate the desired blow-up rates for all reduced quantities, which allow us to close the bootstrap argument.
\begin{df}\label{Def:Norm}
For any $\Si_t$--tangent $k$--tensor field $X_{I_1...I_k}$, we denote the associated linearized quantity (in components) as follows:
\begin{align*}
    \Xc_{I_1...I_k}\coloneqq X_{I_1...I_k}-\Xt_{I_1...I_k}.
\end{align*}
Here $\Xt_{I_1...I_k}\coloneqq \Xt(\et_{I_1},...,\et_{I_k})$ corresponds to the value of the exact generalized Kasner solution present in \eqref{kasnerquantities}.

We define the lower-order norms:
\begin{align*}
    \Ll_{e}(t)&\coloneqq t^{q_M}\|\ec\|_{W^{1,\infty}(\Si_t)},\\
    \Ll_n(t)&\coloneqq t^{-2\si}\|\vphi\|_{W^{1,\infty}(\Si_t)}+t^{q_M-2\si}\|\ev\vphi\|_{L^{\infty}(\Si_t)},\\
    \Ll_{\ga}(t)&\coloneqq t^q\|\ga\|_{W^{1,\infty}(\Si_t)},\\
    \Ll_k(t)&\coloneqq t\|\kc\|_{W^{1,\infty}(\Si_t)},\\
    \Ll_\psi(t)&\coloneqq t^q\|\ev\psi\|_{W^{1,\infty}(\Si_t)}+t\|\tpsic\|_{W^{1,\infty}(\Si_t)},\\
    \Ll_F(t)&\coloneqq t^q\|F\|_{W^{1,\infty}(\Si_t)},\\
    \Ll_\TT(t)&\coloneqq t^{\frac{1+q_M}{2}}\|\TT\|_{L^{\infty}(\Si_t)}.
\end{align*}
and the higher-order norms:
\begin{align*}
\Hh_{e}(t)&\coloneqq t^{A_*+q_M}\|\ec\|_{\Hdot^\kl(\Si_t)},\\
\Hh_n(t)&\coloneqq t^{A_*}\|\vphi\|_{\Hdot^\kl(\Si_t)}+t^{A_*+1}\|\ev\vphi\|_{\Hdot^\kl(\Si_t)},\\
\Hh_{\ga}(t)&\coloneqq t^{A_*+1}\|\ga\|_{\Hdot^\kl(\Si_t)}\\
\Hh_k(t)&\coloneqq t^{A_*+1}\|k\|_{\Hdot^\kl(\Si_t)},\\
\Hh_\psi(t)&\coloneqq t^{A_*+1}\|\vec{e}\psi\|_{\dot{H}^\kl(\Si_t)}+t^{A_*+1}\|e_0\psi\|_{\dot{H}^\kl(\Si_t)},\\
\Hh_F(t)&\coloneqq t^{A_*+1}\|F\|_{\Hdot^\kl(\Si_t)},\\
\Hh_\TT(t)&\coloneqq t^{A_*+\de q+\frac{1+q_M}{2}}\|\TT^{(k_*)}\|_{L^2(\Si_t)},
\end{align*}
where $\TT^{(k_*)}$ is defined as in \eqref{dfTTk}. We also define the following total norms for the dynamic variables\footnote{The dynamic variables contain all the geometric quantities and matter fields, except the lapse function $n$, since there is no evolution equation for $n$.}
\begin{align*}
    \Ll(t)&\coloneqq \Ll_{e}(t)+\Ll_{\ga}(t)+\Ll_{k}(t)+\Ll_\psi(t)+\Ll_F(t)+\Ll_\TT(t),\\
    \Hh(t)&\coloneqq \Hh_{e}(t)+\Hh_{\ga}(t)+\Hh_k(t)+\Hh_\psi(t)+\Hh_F(t)+\Hh_\TT(t).
\end{align*}
and
\begin{align*}
    \Dd(t)\coloneqq \Ll(t)+\Hh(t).
\end{align*}
\end{df}
\subsection{Local Well-Posedness of EMSVS}
To prove the stable Big Bang formation of Kasner solutions toward $t=0$, we need to utilize the local existence result for the Einstein--Maxwell--scalar field--(massless) Vlasov system. In this current paper, we mainly adapt to the local existence results \cite{Svedberg} by Svedberg in harmonic gauge and \cite{FU24} by Fajman--Urban in CMC gauge. These results are summarized as below:\footnote{Although they only considered the Einstein--scalar field--Vlasov system, it is not difficult to include the Maxwell field in the associated local well-posedness result based on their proofs.}
\begin{thm}[local well-posedness for the EMSVS \cite{Svedberg,FU24}]\label{Thm:local}
    Let $l\in \mathbb{N}, l\ge 3$ and $\mu\ge 3$, and let $(\Si_{t_0},\go,\ko,\psio,\phio,\Fo,\fo)$ be an initial data set for the Einstein--Maxwell--scalar field--Vlasov system \eqref{EMS}--\eqref{Vlasov} satisfying
    \begin{equation*}
        \|\go\|_{H^{l+1}(\Si_{t_0})}+ \|\ko\|_{H^{l}(\Si_{t_0})}+ \|\psio\|_{H^{l}(\Si_{t_0})}+ \|\phio\|_{H^{l}(\Si_{t_0})}+\|\Fo\|_{H^{l}(\Si_{t_0})}+\|\fo\|_{H^l_{\mu}(T\Si_{t_0})}<\infty.
    \end{equation*}
    Here for any $\xi=\xi(t, x, p)$,
    \begin{equation*}
        \|\xi\|_{H^l_{\mu}(T\Si_t)}\coloneqq \max\limits_{|\io_1|+|\io_2|\le l}\left\|\frac{\jp^{{\mu+l}}}{|p|{^l}} \pr_x^{\io_1} (p\pr_{p})^{\io_2}\xi\right\|_{L_x^{2}L^{2}_p(T\Si_t)}
    \end{equation*}
    with $\jp=\sqrt{|p|^2+1}$. Then, this data gives rise to a unique solution $(\MM, g, k, \psi, F ,f)$ to the EMSVS \eqref{EMS}--\eqref{Vlasov} in a short time interval $J=(t_0-h, t_0]$ for some $h>0$, with
    \begin{equation}\label{Cond:local}
        g\in C\big(J; H^{l+1}(\Si_t)\big), \qquad k,\pr_t\psi, \nab\psi, F\in C\big(J; H^{l}(\Si_t)\big), \qquad  f\in C\big(J; H^{l}_{\mu}(T\Si_t)\big).
    \end{equation}
    Furthermore, for the past-maximal existence interval $(T_-,t_0]$, we have either $T_-=0$ or 
    \begin{equation}\label{Cond:extend}
        \limsup\limits_{t\searrow T_-} \left( \|g\|_{C^2(\Si_t)}+\|(k,\pr_t\psi, \nab\psi, F)\|_{C^1(\Si_t)}+\left\|\jp^{\mu+l} |p|^{-l} (\pr_x, p\pr_{p})^{\le 1}f\right\|_{L^{\infty}_xL^2_p(T\Si_t)}\right)=\infty.
    \end{equation}
\end{thm}
\begin{rk}
    Compared to the local existence result for the massive Vlasov field, one has to work with the singular weight $(\frac{\jp}{|p|})^{l}$ for the particle distribution function $f$. This is due to the singularity of the mass shell $P$ on the apex of the cone.
\end{rk}
\subsection{Statement of the Main Theorem}
Our main goal is to establish the following theorem.
\begin{thm}\label{mainestimates}
Consider an initial data set $(\Si_1,g,k,\psi,\mathring{\phi},\Fo,\fo)$ for the Einstein--Maxwell--scalar field--Vlasov system \eqref{EMS}--\eqref{Vlasov}. There exists a constant $\ep_0>$ sufficiently small, so that if the initial data for the reduced variables $(k,\ga, e, \psi, F, f)$ along $\Si_1$ satisfy $\fo\in H^{k_*}_{\mu}(T\Si_1)$ with $\mu\ge 3$ and
\begin{align}\label{initial}
\Dd(1)\le \ep_0,
\end{align}
then the reduced system in Proposition \ref{basicequations} admits a unique solution on the slab $(0,1]\times\Tt^3$. Moreover, the following estimate holds for all $t\in(0,1]$:
\begin{equation}\label{finalestimates}
    \Dd(t)+\Ll_n(t)+\Hh_n(t)\les\ep_0.
\end{equation}
\end{thm}
\begin{rk}
We do not need the initial assumption for $\Ll_n(1)$ and $\Hh_n(1)$ along $\Si_1$, as both quantities can be controlled by $\Dd(1)$ via the lapse equation. And it is straightforward to verify that the corresponding initial data set satisfies \eqref{Cond:local} with $l=k_*$.
\end{rk}
\begin{rk}
The strong stability condition \eqref{strongstability} is only used when deriving the $L^2$-energy estimates for the Vlasov equation.
\end{rk}
\subsection{Bootstrap Assumptions and Main Intermediate Results}\label{ssecintermdediate}
For a small $\ep>0$, we make the following bootstrap assumption
\begin{align}\label{B1}
    \Ll_n(t)+\Hh_n(t)+\Dd(t)\leq\ep
\end{align}
for all $t\in(T_*, 1]$ with $T_*\in[0, 1)$ being a bootstrap time. To improve this bootstrap bound, we aim to establish three main intermediate results as stated below.
\begin{thm}\label{M1}
Under the initial condition \eqref{initial} in Theorem \ref{mainestimates} and the bootstrap assumption \eqref{B1}, for the lapse $n$ the following estimate holds
\begin{align*}
    \Ll_n(t)+\Hh_n(t)\les \Dd(t).
\end{align*}
\end{thm}
We will prove Theorem \ref{M1} in Section \ref{secelliptic}. The main idea is to apply the maximum principle and $L^2$-energy estimate for the elliptic lapse equation \eqref{2.25}.
\begin{thm}\label{M2}
Under the same assumptions in \Cref{M1}, for the lower-order dynamical variables, the following estimate holds
\begin{align*}
    \Ll(t)^2\les \ep_0^2+\int_t^1s^{-1+2\si}\Dd(s)^2ds.
\end{align*}
\end{thm}
The proof of Theorem \ref{M2} is provided in Section \ref{sectransport}, by treating the main reduced equations as evolution equations and integrating it from $t$ to $1$.
\begin{thm}\label{M3}
   Under the same assumptions in \Cref{M1}, for the higher-order dynamical variables, the following estimate holds
    \begin{align*}
        \Hh(t)^2\les_*\ep_0^2+\int_t^1s^{-1+2\si}\Dd(s)^2ds.
    \end{align*}
\end{thm}
The proof of Theorem \ref{M3} is postponed in Section \ref{secenergy}, based on the $t$-weighted $L^2$-energy estimates.
\vspace{2mm}

Furthermore, to extend the bootstrap spacetime beyond $t=T_*$, we have to verify the continuation criterion for the EMSVS as stated in \eqref{Cond:extend} in Theorem \ref{Thm:local}, namely,
\begin{thm}\label{M4}
    Under the same assumptions in \Cref{M1}, we have 
    \begin{equation}\label{Est:localcriterion}
        \sup\limits_{t\in (T_*, 1]} \left( \|g\|_{C^2(\Si_t)}+\|(k,\pr_t\psi, \nab\psi, F)\|_{C^1(\Si_t)}+\left\|\jp^{\mu+l} |p|^{-l} (\pr_x, p\pr_{p})^{\le 1}f\right\|_{L^{\infty}_xL^2_p(T\Si_t)}\right)<\infty.
    \end{equation}
\end{thm}
Theorem \ref{M4} is proved in a similar manner to that of Theorem \ref{M3}. The proof details are put in Section \ref{Sec:M4proof}.

\subsection{Proof of the Main Theorem}\label{ssecmainproof}
Building on the main intermediate results, namely Theorems \ref{M1}--\ref{M3}, we are ready to prove \Cref{mainestimates}. Let $\mathcal{U}$ be the set of all $T_*\in[0, 1)$ such that \eqref{Est:localcriterion} and the following bootstrap bound hold
\begin{align}
    \Ll_n(t)+\Hh_n(t)+\Dd(t)\leq\ep  \qquad \qquad \text{for all} \quad t\in(T_*, 1].\label{B2}
\end{align}
First, utilizing the initial condition \eqref{initial} and Theorem \ref{M1}, we infer that \eqref{B2} holds if $T_*$ is sufficiently close to $1$. Consequently, we have $\inf\,\mathcal{U}<1$. 

Assume that
\begin{align}
    T_0\coloneqq \inf\mathcal{U}\ne 0.
\end{align}
Then, we have $T_0\in \mathcal{U}$. Within the spacetime slab $(T_0,1]\times\mathbb{T}^3$, from Theorem \ref{M2} and Theorem \ref{M3} we obtain
\begin{align*}
    \Dd(t)^2\les_*\ep_0^2+\int_t^1 s^{-1+2\si}\Dd(s)^2ds.
\end{align*}
Employing Gr\"onwall's inequality, we then derive
\begin{align*}
    \Dd(t)\les_*\ep_0.
\end{align*}
Combining with Theorem \ref{M1},  this yields
\begin{align*}
    \Ll_n(t)+\Hh_n(t)+\Dd(t)\les_*\ep_0 \qquad \qquad \text{for all} \quad t\in(T_0, 1].
\end{align*}
According to our choice of $\ep_0,\ep$, i.e., $\ep_0\ll_*\ep$ and Theorem \ref{M4},  utilizing the local existence result in Theorem \ref{Thm:local}, we deduce that for $\de>0$ small enough, $T_0-\de\in\mathcal{U}$.  This contradicts the definition of $T_0$. Therefore, it holds $T_0=0$, which completes the proof of Theorem \ref{mainestimates}.
\section{First Consequences of Bootstrap Assumptions}\label{secbootass}
In this section, we provide some basic inequalities and commutation formulae that we will frequently use in Sections \ref{secelliptic}--\ref{secenergy}.
\subsection{Blow-up Properties of Bootstrap Assumptions}
Recall our bootstrap assumption
\begin{align}\label{boot}
    \Dd(t)+\Ll_n(t)+\Hh_n(t)\leq \ep,\qquad \forall\; t\in(T_*,1].
\end{align}
We first introduce the following conventions for the reduced variables, which are rather helpful when estimating the error terms present in the reduced system from Proposition \ref{basicequations}.
\begin{df}\label{gammag}
We group the dynamic variables as below:
\begin{align*}
    \Gag\coloneqq \{\ec\},\qquad \Gab\coloneqq \{\ga,\,\ev\psi,\, F\},\qquad\Gaw\coloneqq \{\kc,\,\tpsic\}.
\end{align*}
For $k\in\mathbb{N}$, we also denote
\begin{align*}
\Gag^{(k)}\coloneqq \max_{|\io|\leq k}|\pr^{\io}\Gag|,\qquad\Gab^{(k)}\coloneqq \max_{|\io|\leq k}|\pr^{\io}\Gab|,\qquad\Gaw^{(k)}\coloneqq \max_{|\io|\leq k}|\pr^{\io}\Gaw|.
\end{align*}
\end{df}
\begin{lem}\label{decayGagGabGaw}
Under the bootstrap assumption \eqref{boot}, we have the following blow-up estimates:
\begin{align*}
\|\Gag^{(1)}\|_{L^\infty(\Si_t)}&\leq\frac{\Dd(t)}{t^{q_M}}\leq\frac{\ep}{t^{q_M}},\qquad\qquad\;\,\|\Gag^{(\kl)}\|_{L^2(\Si_t)}\leq\frac{\Dd(t)}{t^{A_*+q_M}}\leq\frac{\ep}{t^{A_*+q_M}},\\
\|\Gab^{(1)}\|_{L^\infty(\Si_t)}&\leq\frac{\Dd(t)}{t^q}\leq\frac{\ep}{t^q},\qquad\, \quad\qquad\|\Gab^{(\kl)}\|_{L^2(\Si_t)}\leq\frac{\Dd(t)}{t^{A_*+1}}\leq\frac{\ep}{t^{A_*+1}},\\
\|\Gaw^{(1)}\|_{L^\infty(\Si_t)}&\leq\frac{\Dd(t)}{t}\leq\frac{\ep}{t},\qquad\;\; \quad\qquad\|\Gaw^{(\kl)}\|_{L^2(\Si_t)}\leq\frac{\Dd(t)}{t^{A_*+1}}\leq\frac{\ep}{t^{A_*+1}},\\
\|\vphi^{(1)}\|_{L^\infty(\Si_t)}&\leq\Ll_n(t)t^{2\si}\leq\ep t^{2\si},\qquad\quad\,\|\vphi^{(\kl)}\|_{L^2(\Si_t)}\leq\frac{\Hh_n(t)}{t^{A_*}}\leq\frac{\ep}{t^{A_*}},\\
\|(\ev\vphi)^{(1)}\|_{L^\infty(\Si_t)}&\leq\frac{\Ll_n(t)}{t^{q-2\si}}\leq\frac{\ep}{t^{q-2\si}},\quad\quad\;\|(\ev\vphi)^{(\kl)}\|_{L^2(\Si_t)}\leq\frac{\Hh_n(t)}{t^{A_*+1}}\leq\frac{\ep}{t^{A_*+1}},\\
\|\TT\|_{L^\infty(\Si_t)}&\leq\frac{\Dd(t)}{t^\frac{1+q_M}{2}}\leq\frac{\ep}{t^\frac{1+q_M}{2}},\qquad\;\;\;\|\TT^{(\kl)}\|_{L^2(\Si_t)}\leq\frac{\Dd(t)}{t^{A_*+\dq+\frac{1+q_M}{2}}}\leq\frac{\ep}{t^{A_*+\dq+\frac{1+q_M}{2}}}.
\end{align*}
\end{lem}
\begin{proof}
    This follows readily from \eqref{boot} and Definition \ref{gammag}.
\end{proof}
\begin{rk}
For two quantities $X_1$ and $X_2$, we write
\begin{align*}
    X_1\preceq X_2
\end{align*}
if $X_1^{(\io)}$ blows up slower than $X_2^{(\io)}$ for any $|\io|\leq\kl$. In particular, Lemma \ref{decayGagGabGaw} implies
\begin{align*}
    \Gag\preceq\Gab\preceq\Gaw\qquad\quad\mbox{ and }\qquad\quad n\c\DD_i\preceq\DD_i\quad\;\mbox{ for }\; i=g,b,w.
\end{align*}
\end{rk}
\begin{rk}\label{Gawrk}
    In the sequel, we adopt the following notations:
\begin{itemize}
    \item For a quantity $h$ that exhibits the same blow-up behavior as $\DD_i$ with $i\in \{g,b,w\}$ in Lemma \ref{decayGagGabGaw}, we also write
    \begin{equation*}
        h\in\DD_i,\qquad i=g,b,w.
    \end{equation*}
    \item If $X_{(1)}\preceq X_{(2)}$, we schematically write
    \begin{align*}
        X_{(1)}+X_{(2)}=X_{(2)}.
    \end{align*}
    For example,
    \begin{equation*}
        \Gag+\Gab=\Gab,\qquad\Gag\c\Gab+\Gag\c\Gaw=\Gag\c\Gaw.
    \end{equation*}
\end{itemize}
\end{rk}
\subsection{Interpolation and Product Inequalities}
In this subsection, we list several useful inequalities on $\Si_t$ and on $T\Si_t$. These are helpful when controlling various error terms for the reduced EMSVS.
\begin{lem}\label{Sobolevinterpolation}
Consider $\Si_t$--tangent tensorfields $v, v_1, \cdots, v_R$ with $R\ge 1$. Let $M,M_1,M_2\geq 0$ and $\iota_1,...,\iota_R$ be multi-indices such that $\sum_{r=1}^R|\io_r|=M$. Then the following inequalities hold
\begin{align}
    \|v\|_{W^{M_1,\infty}(\Si_t)}&\les_{M_1,M_2}\|v\|_{L^\infty(\Si_t)}^{1-\frac{M_1}{M_2}}\|v\|_{\dot{W}^{M_2,\infty}(\Si_t)}^\frac{M_1}{M_2}+\|v\|_{L^\infty(\Si_t)},\quad\;\quad M_2\geq M_1,\label{eqinterpolation}\\
    \|v\|_{W^{M_1,\infty}(\Si_t)}&\les_{M_1,M_2}\|v\|_{H^{M_1+2}(\Si_t)}\nonumber\\
    &\les_{M_1,M_2}\|v\|_{L^\infty(\Si_t)}+\|v\|_{\dot{H}^{M_2}(\Si_t)},\qquad\qquad\qquad\quad\;\, M_2\geq M_1+2,\label{eqsobolev}\\
\|\pr^{\iota_1}v_1...\pr^{\iota_R}v_R\|_{L^2(\Si_t)}&\les_{M_1,M_2}\sum_{r=1}^R\|v_r\|_{\Hdot^{M}(\Si_t)}\prod_{s\ne r}\|v_s\|_{L^\infty(\Si_t)}.\label{eqproduct}
\end{align}
\end{lem}
\begin{proof}
    Note that \eqref{eqinterpolation} and \eqref{eqsobolev} follow directly from Sobolev and interpolation. For the proof of \eqref{eqproduct}, see Lemma 6.16 in \cite{Ringstrom09}.
\end{proof}
\begin{lem}\label{Generalinterpolation}
Let $h, h_2$ be scalar functions defined on $T\Si_t$ and $h_1$ be a scalar function defined on $\Si_t$, and let $M_2\ge M_1\geq 0$. Then, the following inequalities hold:
\begin{align*}
\|h^{(M_1)}\|_{L^{\infty}_x(\Si_t)L^2_p(\mathbb{R}^3)}&\les_{M_1, M_2} \|h\|_{L^{\infty}_x(\Si_t)L^2_p(\mathbb{R}^3)}^{1-\frac{M_1}{M_2}}\|h^{(M_2)}\|_{L^{\infty}_x(\Si_t)L^2_p(\mathbb{R}^3)}^{\frac{M_1}{M_2}},\\
\|h^{(M_1)}\|_{L^{\infty}_x(\Si_t)L^2_p(\mathbb{R}^3)}&\les_{M_1,M_2} \|h\|_{L^{\infty}_x(\Si_t)L^2_p(\mathbb{R}^3)}+\|h\|_{\Hdot^{M_1+2}_{x,p}(T\Si_t)},\\
\left\|h_1^{(M_1)}\c h_2^{(M_2)}\right\|_{L^2_{x,p}(T\Si_t)}&\les_{M_1,M_2}\left\|h_1\right\|_{H^{M_1+M_2}_{x}(\Si_t)} \|h_2\|_{L^{\infty}_x(\Si_t)L^2_p(\mathbb{R}^3)}\\
&+ \|h_1\|_{L^{\infty}_x(\Si_t)}\left\|h_2\right\|_{H^{M_1+M_2}_{x,p}(T\Si_t)}.
\end{align*}
Here $h^{(M)}\coloneqq \{h^{(\io_1, \io_2)}: \ |\io_1|+|\io_2|\le M\}$.
\end{lem}
\begin{proof}
This follows directly from the mixed Sobolev inequality and interpolation.
\end{proof}
The next lemma is a direct implication of Lemma \ref{Sobolevinterpolation} and Lemma \ref{Generalinterpolation}. It establishes the quantitative decaying estimates for the reduced variables with one more derivative than is used at lower orders in Definition \ref{Def:Norm}.
\begin{lem}\label{interpolation}
The following estimates hold for $t\in(T_*,1]$ and $\kl$ large enough:
   \begin{align*}
        \|\Gag\|_{W^{2,\infty}(\Si_t)}&\les t^{-q}\Dd(t),\\
        \|\Gab\|_{W^{2,\infty}(\Si_t)}&\les t^{-q-\dq}\Dd(t),\\
        \|\Gaw\|_{W^{2,\infty}(\Si_t)}&\les t^{-1-\dq}\Dd(t),\\
        \|\Gaw\|_{W^{3,\infty}(\Si_t)}&\les t^{-1-2\dq}\Dd(t),\\
        \|\ev\ga\|_{W^{1,\infty}(\Si_t)}+\|\ev(\ev\psi)\|_{W^{1,\infty}(\Si_t)}+\|\ev F\|_{W^{1,\infty}(\Si_t)}&\les t^{-2q}\Dd(t),\\
        \|\TT\|_{W^{1,\infty}(\Si_t)}&\les t^{-\frac{1+q_M}{2}-\dq}\Dd(t).
    \end{align*}
\end{lem}
\begin{proof}
Employing Lemma \ref{decayGagGabGaw} and Lemma \ref{Sobolevinterpolation}, we obtain
\begin{align*}
    \|\Gag\|_{W^{2,\infty}(\Si_t)}&\les\|\Gag\|_{W^{1,\infty}(\Si_t)}\|^{1-\frac{1}{k_*-3}}\|\Gag\|_{\dot{W}^{k_*-2,\infty}(\Si_t)}^\frac{1}{k_*-3}+\|\Gag\|_{W^{1,\infty}(\Si_t)}\\
    &\les\|\Gag\|_{W^{1,\infty}(\Si_t)}\|^{1-\frac{1}{k_*-3}}\|\Gag\|_{\Hdot^{k_*,\infty}(\Si_t)}^\frac{1}{k_*-3}+\|\Gag\|_{W^{1,\infty}(\Si_t)}\\
    &\les\left(\frac{\Dd(t)}{t^{q_M}}\right)^{1-\frac{1}{k_*-3}}\left(\frac{\Dd(t)}{t^{A_*+q_M}}\right)^{\frac{1}{\kl-3}}+\frac{\Dd(t)}{t^{q_M}}\\
    &\les\frac{\Dd(t)}{t^{\frac{A_*}{\kl-3}+q_M}}+\frac{\Dd(t)}{t^{q_M}}\les\frac{\Dd(t)}{t^{q_M+\dq}}\les t^{-q}\Dd(t),
\end{align*}
where in the last line we use \eqref{AVTD} and \eqref{dfkl}. Proceeding in a similar manner as above, we also derive
\begin{align*}
    \|\Gab\|_{W^{2,\infty}(\Si_t)}\les t^{-q-\dq}\Dd(t),\quad\;\|\Gaw\|_{W^{2,\infty}(\Si_t)}\les t^{-1-\dq}\Dd(t),\quad\;\|\Gaw\|_{W^{3,\infty}(\Si_t)}\les t^{-1-2\dq}\Dd(t).
\end{align*}
We move to estimate $\ev\ga, \ev(\ev\psi)$ and $\ev F$. For $\ev\ga$, by utilizing  Lemma \ref{decayGagGabGaw}, Lemma \ref{Sobolevinterpolation}, and noting \eqref{AVTD} and \eqref{dfkl} again, we deduce
\begin{align*}
    \|\ev\ga\|_{W^{1,\infty}(\Si_t)}&\les t^{-q_M}\|\ga\|_{W^{2,\infty}(\Si_t)}\\
    &\les t^{-q_M}\|\ga\|_{W^{1,\infty}(\Si_t)}^{1-\frac{1}{\kl-3}}\|\ga\|_{\Hdot^{\kl}(\Si_t)}^\frac{1}{\kl-3}+t^{-q_M}\|\ga\|_{W^{1,\infty}(\Si_t)}\\
    &\les\frac{1}{t^{q_M}}\left(\frac{\Dd(t)}{t^{q}}\right)^{1-\frac{1}{k_*-3}}\left(\frac{\Dd(t)}{t^{A_*+1}}\right)^{\frac{1}{\kl-3}}+\frac{\Dd(t)}{t^{q_M+q}}\\
    &\les\frac{\Dd(t)}{t^{\frac{A_*+1-q}{\kl-3}+q_M+q}}\les t^{-2q}\Dd(t).
\end{align*}
The similar argument also gives
\begin{align*}
    \|\ev(\ev\psi)\|_{W^{1,\infty}(\Si_t)}+\|\ev F\|_{W^{1,\infty}(\Si_t)}\les t^{-2q}\Dd(t).
\end{align*}
Finally, applying Lemma \ref{Generalinterpolation}, together with \eqref{dfkl}, we get the desired bound of $\|\TT\|_{W^{1,\infty}(\Si_t)}$ as below:
\begin{align*}
    \|\TT\|_{W^{1,\infty}(\Si_t)}&\les\|\TT\|_{L^{\infty}(\Si_t)}^{1-\frac{1}{\kl-2}}\|\TT\|_{\dot{W}^{\kl-2,\infty}(\Si_t)}^{\frac{1}{\kl-2}}+\|\TT\|_{L^{\infty}(\Si_t)}\\
    &\les\|\TT\|_{L^{\infty}(\Si_t)}^{1-\frac{1}{\kl-2}}\|\TT\|_{\Hdot^{\kl,\infty}(\Si_t)}^{\frac{1}{\kl-2}}+\|\TT\|_{L^{\infty}(\Si_t)}\\
    &\les\left(\frac{\Dd(t)}{t^\frac{1+q_M}{2}}\right)^{1-\frac{1}{\kl-2}}\left(\frac{\Dd(t)}{t^{A_*+\dq+\frac{1+q_M}{2}}}\right)^{\frac{1}{\kl-2}}+\frac{\Dd(t)}{t^\frac{1+q_M}{2}}\\
    &\les\frac{\Dd(t)}{t^{\frac{A_*+\dq}{\kl-2}+\frac{1+q_M}{2}}}+\frac{\Dd(t)}{t^\frac{1+q_M}{2}}\\
    &\les t^{-\frac{1+q_M}{2}-\dq}\Dd(t).
\end{align*}
This finishes the proof of Lemma \ref{interpolation}.
\end{proof}
Based on the decay estimate of $\TT$ as shown in Lemma \ref{decayGagGabGaw} and Lemma \ref{interpolation}, we are able to prove the following bound for the energy-momentum tensor $T$.
\begin{lem}\label{topT}
For the energy-momentum tensor $T$, the following estimates hold true
\begin{align}
t^{1+q}\|T\|_{W^{1,\infty}(\Si_t)}&\les\Dd(t)^2,\label{T1}\\
t^{A_*+1+q}\|T^{(k_*)}\|_{L^2(\Si_t)}&\les_*\Dd(t)^2.\label{Tk*}
\end{align}
\end{lem}
\begin{proof}
In view of the definitions of $T$ and $\TT$ in \eqref{dfT} and \eqref{dfTT}, via applying Lemma \ref{decayGagGabGaw} and Lemma \ref{interpolation}, we derive
\begin{align*}
    |T^{(1)}|\les\left\|\f^{(1)}(p^0)^\frac{1}{2}\right\|_{L^2_p(\mathbb{R}^3)}\left\|(p^0)^\frac{1}{2}\f\right\|_{L^2_p(\mathbb{R}^3)}\les|\TT^{(1)}|\c|\TT|\les\frac{\Dd(t)}{t^{\frac{1+q_M}{2}+\dq}}\frac{\Dd(t)}{t^{\frac{1+q_M}{2}}}\les\frac{\Dd(t)^2}{t^{1+q}},
\end{align*}
which implies \eqref{T1}. Next, from \eqref{dfT}, \eqref{dfTT} again, and using \eqref{eqproduct} in Lemma \ref{Sobolevinterpolation}, we deduce
\begin{align*}
|T^{(\kl)}|\les_*\left\|\f^{(\kl)}(p^0)^\frac{1}{2}\right\|_{L^2_p(\mathbb{R}^3)}\left\|(p^0)^\frac{1}{2}\f\right\|_{L^2_p(\mathbb{R}^3)}\les_*|\TT^{(k_*)}|\c|\TT|.
\end{align*}
Combining with Lemma \ref{decayGagGabGaw} and \eqref{AVTD}, we thus arrive at
\begin{align*}
    \|T^{(k_*)}\|_{L^2(\Si_t)}\les_* \|\TT^{(k_*)}\|_{L^2(\Si_t)}\|\TT\|_{L^\infty(\Si_t)}\les_*\frac{\Dd(t)}{t^{A_*+\dq+\frac{q_M+1}{2}}}\frac{\Dd(t)}{t^{\frac{q_M+1}{2}}}\les_*\frac{\Dd(t)^2}{t^{A_*+q+1}},
\end{align*}
which gives \eqref{Tk*} as desired.
\end{proof}
\subsection{Commutation Formulae}
To derive estimates for the derivatives of reduced variables, we need to commute the reduced equations with the transported spatial coordinate derivative $\{\pr_i\}_{i=1,2,3}$ up to the top order. Note that by the expansion $e_I=e_I^c\pr_c$, we readily get the following commutation identity
\begin{align}\label{commutation1}
    [\pr^\io,e_I]=\sum_{\io_1\cup\io_2=\io,|\io_2|<|\io|}(\pr^{\io_1}e_I^c)\pr^{\io_2}\pr_c.
\end{align}
We also make use of the commutation identity involving $\partial_t$ as established in \eqref{eq2.35}, i.e.,
\begin{equation}\label{commutation2}
    [\pr_t,e_I]=nk_{IC}e_C^c\pr_c,
\end{equation}
For the commutator of $e_I$ with higher-order spatial coordinate derivatives, by employing \eqref{commutation1} repeatedly, we obtain
\begin{prop}\label{commutation}
The following schematic commutation formula holds true
\begin{align*}
    [\pr^\io,e_I]v=\left(\Gag^{(1)}\c v^{(1)}\right)^{(\io-1)}.
\end{align*}
\end{prop}
\section{Estimates for the Lapse Function \texorpdfstring{$n$}{}}\label{secelliptic}
In this section, we prove Theorem \ref{M1} using elliptic estimates for the lapse equation \eqref{2.25}.
\subsection{Decay Estimates for \texorpdfstring{$\vphi$}{} (Maximum Principle)}
The following basic lemma of elliptic theory will be useful to estimate $\vphi$ in lower order.
\begin{lem}\label{maximalprinciple}
Let $u$ be the solution of the following elliptic equation on $\Si_t$:
\begin{align}\label{ellipticeq}
    e_C(e_Cu)-t^{-2}u=f.
\end{align}
Then we have
\begin{align}\label{ellipticconclusion}
    \|u\|_{L^\infty(\Si_t)}\leq t^2\|f\|_{L^\infty(\Si_t)}.
\end{align}
\end{lem}
\begin{proof}
Assume that $u$ reaches its maximum and minimum, respectively, at $x_M$ and $x_m$. Hence it holds\footnote{Note that $\Si_t\simeq \mathbb{T}^3$.}
\begin{equation*}
    e_C e_C u(x_M)\le 0\le e_C e_C u(x_m).
\end{equation*}
Utilizing \eqref{ellipticeq}, we then obtain
\begin{align*}
    t^{-2}u(x_M)\leq& t^{-2}u(x_M)-e_Ce_Cu(x_M)\leq|f(x_M)|, \\
    t^{-2}u(x_m)\geq& t^{-2}u(x_m)-e_Ce_Cu(x_m)\geq-|f(x_m)|.
\end{align*}
Therefore, for all $x\in \Si_t$ we conclude
\begin{align*}
   -\sup_{\Si_t}|f|\le t^{-2}u(x_m) \le t^{-2}u(x) \le t^{-2}u(x_M)\leq\sup_{\Si_t}|f|.
\end{align*}
This concludes the proof of Lemma \ref{maximalprinciple}.
\end{proof}
We are now ready to control the lower-order norm of the lapse $n$ as follows.
\begin{prop}\label{decayn}
For the lower-order norm $\Ll_n(t)$, we have the following estimate:
\begin{align*}
    \Ll_n(t)\les\Dd(t)+\ep\Hh_n(t).
\end{align*}
\end{prop}
\begin{proof}
We rewrite \eqref{2.25} in the below systematic form
\begin{align}\label{Eqn:elliptic n 1}
e_C(e_C\vphi)-t^{-2}\vphi&=\ev\vphi\c\ev\vphi+2e_C(\ga_{DDC})+t^{-2}\vphi\c\vphi+\Gab\c\Gab+T.
\end{align}
Commuting with $\pr^\io$ for $|\io|\leq 1$ and applying Proposition \ref{commutation}, we deduce
\begin{align}\label{Eqn:pr1n}
e_C(e_C(\pr^\io\vphi))-t^{-2}\pr^\io\vphi&=\ev\left(\Gag^{(1)}\c\vphi^{(1)}\right)+\Gag^{(1)}\c(\ev\vphi)^{(1)}+(\ev\vphi\c\ev\vphi)^{(1)}\\
&+(\ev\ga)^{(1)}+t^{-2}(\vphi\c\vphi)^{(1)}+(\Gab\c\Gab)^{(1)}+T^{(1)}.
\end{align}
By utilizing Lemma \ref{decayGagGabGaw}, Lemma \ref{interpolation} and Lemma \ref{topT}, we estimate the terms on the right of \eqref{Eqn:pr1n} and hence get
\begin{align*}
    \left\|e_C(e_C(\pr^\io\vphi))-t^{-2}\pr^\io\vphi\right\|_{L^\infty(\Si_t)}\les\frac{\Dd(t)}{t^{1+q}}+\frac{\ep(\Ll_n(t)+\Hh_n(t))}{t^{2q}}.
\end{align*}
Consequently, applying Lemma \ref{maximalprinciple} for $\pr^\io\vphi$, we deduce
\begin{align*}
    \|\vphi\|_{W^{1,\infty}(\Si_t)}\les t^{1-q}(\Dd(t)+\ep(\Ll_n(t)+\Hh_n(t)))\les t^{2\si}\Dd(t)+\ep t^{2\si}(\Ll_n(t)+\Hh_n(t)).
\end{align*}
Combining with Lemma \ref{decayGagGabGaw}, this implies
\begin{align*}
    \|\ev\vphi\|_{L^\infty(\Si_t)}&\les t^{-q_M+2\si}\Big(\Dd(t)+\ep\big(\Ll_n(t)+\Hh_n(t)\big)\Big).
\end{align*}
This concludes the proof of Proposition \ref{decayn}.
\end{proof}
\subsection{Top-order Estimates for \texorpdfstring{$\vphi$}{} (Energy Estimates)}
Then we turn to establish the desired bound for the higher-order norm of the lapse $n$.
\begin{prop}\label{topn}
For the higher-order norm $\Hh_n(t)$, we have the following estimate:
\begin{align*}
    \Hh_n(t)\les\Dd(t).
\end{align*}
\end{prop}
\begin{proof}
Recall \eqref{Eqn:elliptic n 1}
\begin{align*}
e_C(e_C\vphi)-t^{-2}\vphi&=\ev\vphi\c\ev\vphi+2e_C(\ga_{DDC})+t^{-2}\vphi\c\vphi+\Gab\c\Gab+T.
\end{align*}
Differentiating  by $\pr^\io$ with $|\io|=\kl$ and applying Proposition \ref{commutation}, we deduce
\begin{align*}
e_C\pr^\io e_C\vphi-t^{-2}\pr^\io\vphi&=2e_D\pr^\io\ga_{CCD}+\left(\Gab^{(1)}\c(\ev\vphi)^{(1)}\right)^{(\io-1)}+t^{-2}(\vphi\c\vphi)^{(\io)}\\
&+(\ev\vphi\c\ev\vphi)^{(\io)}+\left(\Gab^{(1)}\c\Gab^{(1)}\right)^{(\io-1)}+T^{(\io)}.
\end{align*}
Multiplying by $-\pr^\io\vphi$, we obtain
\begin{align*}
-\pr^\io\vphi\left(e_C\pr^\io e_C\vphi\right)+t^{-2}(\pr^\io\vphi)^2&=-2(\pr^\io\vphi)e_D\pr^\io\ga_{CCD}+\vphi^{(\io)}\c\left(\Gab^{(1)}\c\Gab^{(1)}+\Gab^{(1)}\c(\ev\vphi)^{(1)}\right)^{(\io-1)}\\
&+\vphi^{(\io)}\c(\ev\vphi\c\ev\vphi)^{(\io)}+t^{-2}\vphi^{(\io)}\c(\vphi\c\vphi)^{(\io)}+T^{(\io)}.
\end{align*}
Integrating over $\Si_t$, we thus infer
\begin{align*}
\int_{\Si_t}|\pr^\io(\ev\vphi)|^2+t^{-2}|\pr^\io\vphi|^2&=\int_{\Si_t}2(\pr^\io e_D\vphi)\pr^\io\ga_{CCD}+\vphi^{(\io)}\c(\ev\vphi\c\ev\vphi)^{(\io)}+\vphi^{(\io)}\c(\vphi\c\vphi)^{(\io)}\\
&+\int_{\Si_t}\vphi^{(\io)}\c\left(\Gab^{(1)}\c\Gab^{(1)}+\Gab^{(1)}\c(\ev\vphi)^{(1)}\right)^{(\io-1)}+T^{(\io)}.
\end{align*}
Note that from Lemma \ref{decayGagGabGaw} and \eqref{eqproduct} we have
\begin{align*}
    \left\|\left(\Gab^{(1)}\c\Gab^{(1)}\right)^{(\io-1)}\right\|_{L^2(\Si_t)}\les_*\frac{\Dd(t)}{t^{A_*+1}}\frac{\Dd(t)}{t^{q}}\les_*\frac{\Dd(t)^2}{t^{A_*+q+1}}.
\end{align*}
Also, employing Lemma \ref{decayGagGabGaw}, Proposition \ref{decayn} and \eqref{eqproduct}, we estimate
\begin{align*}
    \left\|\left(\Gab^{(1)}\c(\ev\vphi)^{(1)}\right)^{(\io-1)}\right\|_{L^2(\Si_t)}&\les_*\frac{\Dd(t)}{t^{A_*+1}}\frac{\Ll_n(t)+\Hh_n(t)}{t^{q}}\les_*\frac{\Dd(t)(\Dd(t)+\Hh_n(t))}{t^{A_*+q+1}},\\
    t^{-2}\left\|(\vphi\c\vphi)^{(\io)}\right\|_{L^2(\Si_t)}&\les_*\frac{(\Ll_n(t)+\Hh_n(t))^2}{t^{A_*+2-2\si}}\les_*\frac{(\Dd(t)+\Hh_n(t))^2}{t^{A_*+q+1}},\\
    \left\|(\ev\vphi\c\ev\vphi)^{(\io)}\right\|_{L^2(\Si_t)}&\les_*\frac{(\Ll_n(t)+\Hh_n(t))^2}{t^{A_*+1+q-2\si}}\les_*\frac{(\Dd_n(t)+\Hh_n(t))^2}{t^{A_*+1+q-2\si}},\\
    \left\|\vphi^{(\io)}\right\|_{L^2(\Si_t)}&\les\frac{\ep}{t^{A_*}},\\
    \|T^{(\io)}\|_{L^2(\Si_t)}&\les\frac{\ep}{t^{A_*+q+1}}.
\end{align*}
Combining all above estimates, we arrive at\footnote{Here and below, we use $C_*$ to denote a constant that depends on $\kl$.}
\begin{align*}\int_{\Si_t}|\pr^\io(\ev\vphi)|^2+t^{-2}|\pr^\io\vphi|^2\les\int_{\Si_t}(\pr^\io e_D\vphi)\pr^\io\ga_{CCD}+\ep C_*\c\frac{\Hh_n(t)\Dd(t)+\Dd(t)^2}{t^{2A_*+q+1}}.
\end{align*}
Multiplying this by $t^{2A_*+2}$ yields
\begin{align*}
\int_{\Si_t}t^{2A_*+2}|\pr^\io(\ev\vphi)|^2+t^{2A_*}|\pr^\io\vphi|^2\les\int_{\Si_t}t^{2A_*+2}(\pr^\io e_D\vphi)\pr^\io\ga_{CCD}+\ep C_*t^{2\si}\left(\Hh_n(t)\Dd(t)+\Dd(t)^2\right).
\end{align*}
As a result, it follows from Cauchy--Schwarz inequality that
\begin{align*}
\int_{\Si_t}t^{2A_*+2}|\pr^\io(\ev\vphi)|^2+t^{2A_*}|\pr^\io\vphi|^2&\les t^{2A_*+2}\|\pr^\io\ga\|_{L^2(\Si_t)}^2+\ep C_*t^{2\si}\left(\Hh_n(t)^2+\Dd(t)^2\right).
\end{align*}
Therefore, by summing over $|\io|=\kl$, we conclude
\begin{align*}
\Hh_n(t)^2\les\Hh_\ga(t)^2+\ep C_*(\Hh_n(t)^2+\Dd(t)^2).
\end{align*}
Choosing $\ep$ small enough, this implies
\begin{align*}
    \Hh_n(t)\les\Dd(t).
\end{align*}
as stated.
\end{proof}
Finally, incorporating Proposition \ref{decayn} with Proposition \ref{topn}, we conclude the proof of Theorem \ref{M1}.
\begin{rk}
    As a consequence of Theorem \ref{M1}, in the sequel we can systematically write
    \begin{align*}
        t^{-1}\vphi\in\Gab,\qquad\qquad\ev\vphi\in\Gab.
    \end{align*}
\end{rk}
\section{Blow-up Estimates for \texorpdfstring{$\Dd(t)$}{}}\label{sectransport}
The goal of this section is to prove Theorem \ref{M2}, namely, the lower-order estimates for the dynamical reduced variables.
\subsection{Evolution Lemma}
The following evolution lemma enables us to deal with the transport equations for dynamical variables, except the Vlasov field.
\begin{lem}\label{evolution}
Let $U$ and $V$ be two functions satisfying the following evolution equation:
\begin{align}\label{evolutioneq}
    \pr_tU+\frac{\la_0}{t}U=V.
\end{align}
Then, for $\la>\la_0$ the following estimate holds
\begin{equation}\label{subcritialevolution}
    |t^\la U(t)|^2+(\la-\la_0)\int_t^1s^{2\la-1}|U(s)|^2ds\leq|U|^2(1)+\frac{1}{\la-\la_0}\int_t^1s^{2\la+1}|V(s)|^2ds,
\end{equation}
Moreover, in the case $\la_0=1$, we have
\begin{align}\label{criticalevolution}
    |tU(t)|^2\leq |U(1)|^2+2\int_t^1 s^2|U(s)V(s)|ds.
\end{align}
\end{lem}
\begin{proof}
From \eqref{evolutioneq} we obtain
\begin{align*}
    \pr_t(t^{\la}U)+(\la_0-\la)t^{\la-1}U=t^{\la}V.
\end{align*}
Multiplying both sides by $t^\la U$, we infer
    \begin{align}\label{prtUV}
    \frac{1}{2}\pr_t\left(|t^\la U|^2\right)+(\la_0-\la)(t^\la U)(t^{\la-1}U)=(t^\la U)t^\la V.
    \end{align}
Then the integration from $t$ to $1$ gives
\begin{align*}
    |U|^2(1)-|t^\la U|^2+2(\la_0-\la)\int_t^1s^{-1}|s^\la U|^2ds=2\int_t^1(s^\la U)s^\la Vds.
\end{align*}
Applying Cauchy--Schwarz inequality and noting that $\la-\la_0>0$, we deduce
\begin{align*}
    &|t^\la U(t)|^2+2(\la-\la_0)\int_t^1s^{-1}|s^\la U|^2ds\\
    \leq&|U|^2(1)+(\la-\la_0)\int_t^1s^{-1}|s^\la U|^2ds+\frac{1}{\la-\la_0}\int_t^1s^{2\la+1}|V|^2ds,
\end{align*}
which readily implies the desired inequality \eqref{subcritialevolution}. 

Next, by selecting $\la=\la_0=1$ in \eqref{prtUV}, we derive
\begin{align*}
    \pr_t\left(|tU(t)|^2\right)=2t^2UV.
\end{align*}
Integrating it from $t$ to $1$, we have
\begin{align*}
    |tU(t)|^2\leq |U|^2(1)+2\int_t^1s^2|UV|ds,
\end{align*}
which is exactly \eqref{criticalevolution}. This completes the proof of Lemma \ref{evolution}.
\end{proof}

\subsection{Estimate for \texorpdfstring{$\ga$}{}}
In this subsection, we control the $L^\infty$-norm of the connection coefficients $\ga_{IJK}$. Rather than using their evolution equations directly, we work with the evolution equations for the structure coefficients
\begin{align*}
    S_{IJK}\coloneqq \ga_{IJK}+\ga_{JKI}=\g([e_I,e_J],e_K).
\end{align*}
The next lemma provides the evolution equations for $S_{IJK}$, which are decoupled at the linear level.
\begin{lem}\label{gaeq}
    The structure coefficient $S_{IJK}$ obeys the following evolution equation
    \begin{align*}
        e_0(S_{IJK})+\frac{\qIt+\qJt-\qKt}{t}S_{IJK}=t^{-1}O(\ev\vphi)+t^{-q_M}\Gaw^{(1)}+\Gag\c\Gaw^{(1)}+\Gab\c\Gaw.
    \end{align*}
\end{lem}
\begin{proof}
In light of \eqref{2.22b}, we get
\begin{align*}
    e_0(\ga_{IJK})&=-\ga_{KJC}k_{CI}-\ga_{KIC}k_{JC}+\ga_{JKC}k_{IC}+\ga_{JIC}k_{BC}+\ga_{CJK}k_{IC}+O(\ev k)+\ev\vphi\c k.
\end{align*}
Applying Proposition \ref{kasnerquantities}, we convert the RHS of the above equation into the form
\begin{equation}\label{Eqn:e0ga1}
\begin{aligned}
    e_0(\ga_{IJK})&=\frac{\qIt}{t}\ga_{KJI}+\frac{\qJt}{t}\ga_{KIJ}-\frac{\qIt}{t}\ga_{JKI}-\frac{\qKt}{t}\ga_{JIK}-\frac{\qIt}{t}\ga_{IJK}+\err\\
    &=\frac{\qIt-\qJt}{t}\ga_{KJI}-\frac{\qIt-\qKt}{t}\ga_{JKI}-\frac{\qIt}{t}\ga_{IJK}+\err,
\end{aligned}
\end{equation}
where the error terms $\err$ have the expression
\begin{align*}
    \err\coloneqq t^{-1}O(\ev\vphi)+t^{-q_M}\Gaw^{(1)}+\Gag\c\Gaw^{(1)}+\Gab\c\Gaw.
\end{align*}
Similarly, we deduce the transport equation for $S_{JKI}$:
\begin{align}\label{Eqn:e0ga2}
    e_0\ga_{JKI}=\frac{\qJt-\qKt}{t}\ga_{IKJ}-\frac{\qJt-\qIt}{t}\ga_{KIJ}-\frac{\qJt}{t}\ga_{JKI}+\err.
\end{align}
Then adding \eqref{Eqn:e0ga1} and \eqref{Eqn:e0ga1} renders
\begin{align*}
    e_0S_{IJK}&=-\frac{\qIt-\qKt}{t}\ga_{JKI}-\frac{\qIt}{t}\ga_{IJK}+\frac{\qKt-\qJt}{t}\ga_{IJK}-\frac{\qJt}{t}\ga_{JKI}+\err\\
    &=\frac{-\qIt-\qJt+\qKt}{t}\ga_{IJK}+\frac{-\qIt-\qJt+\qKt}{t}\ga_{JKI}+\err\\
    &=\frac{-\qIt-\qJt+\qKt}{t}S_{IJK}+\err.
\end{align*}
This concludes the proof of Lemma \ref{gaeq}.
\end{proof}
We proceed to derive the lower-order estimates for $S_{IJK}$ and $\ga_{IJK}$.
\begin{prop}\label{estga}
For the connection coefficients $\ga$, the following estimate holds
\begin{align*}
    |t^q\ga^{(1)}|^2+\int_t^1s^{2q-1}|\ga^{(1)}|^2ds\les\ep_0^2+\ep^2\int_t^1s^{2q-1}\left|(\Gag,\Gab)^{(1)}\right|^2ds+\int_t^1s^{-1+2\si}\Dd(s)^2ds.
\end{align*}
\end{prop}
\begin{proof}
Employing Lemma \ref{gaeq}, we have
\begin{align*}
    \pr_t(S_{IJK})+\frac{\qIt+\qJt-\qKt}{t}S_{IJK}=t^{-1}O(\ev\vphi)+t^{-q_M}\Gaw^{(1)}+\Gag\c\Gaw^{(1)}+\Gab\c\Gaw.
\end{align*}
Commuting with $\pr^\io$ for $|\io|\leq 1$ and applying Proposition \ref{commutation}, we then deduce
\begin{align}\label{Eqn:prS}
    &\pr_t(\pr^\io S_{IJK})+\frac{\qIt+\qJt-\qKt}{t}(\pr^\io S_{IJK})\\
    =\;&t^{-q_M}\Gaw^{(2)}+\Gag\c\Gaw^{(2)}+\Gag^{(1)}\c\Gaw^{(1)}+(\Gab\c\Gaw)^{(1)}+t^{-1}(\ev\vphi)^{(1)}.
\end{align}
Note that from Lemma \ref{interpolation} and Theorem \ref{M1} we obtain
\begin{equation}\label{Est:S L Right}
\begin{aligned}
    \int_t^1 s^{2q+1}\left(s^{-q_M}\Gaw^{(2)}+\Gag\c\Gaw^{(2)}\right)^2ds&\les \int_t^1 s^{2q+1-2q_M-2\de q-2}\Dd(s)^2ds\\
    &\les\int_t^1s^{-1+2\si}\Dd(s)^2ds,\\
    \int_t^1 s^{2q+1}\left(\Gag^{(1)}\c\Gaw^{(1)}+(\Gab\c\Gaw)^{(1)}\right)^2ds&\les \ep^2\int_t^1s^{2q-1}\left|(\Gag,\Gab)^{(1)}\right|^2ds,\\
    \int_t^1 s^{2q-1}\left|(\ev\vphi)^{(1)}\right|^2&\les \int_t^1s^{2q-1-2q+4\si}(\Ll_n(s)+\Hh_n(s))^2ds\\
    &\les \int_t^1s^{-1+2\si}\Dd(s)^2ds.
\end{aligned}
\end{equation}
Thus, adapting Lemma \ref{evolution} to \eqref{Eqn:prS}  with $\la=q$ and $\la_0=\qIt+\qJt-\qKt$ and injecting the estimates \eqref{Est:S L Right}, we derive
\begin{align*}
    |t^q(\pr^\io S_{IJK})|^2+\int_t^1s^{2q-1}|\pr^\io S_{IJK}|^2ds&\les\ep_0^2+\ep^2\int_t^1s^{2q-1}\left|(\Gag,\Gab)^{(1)}\right|^2ds+\int_t^1s^{-1+2\si}\Dd(s)^2ds.
\end{align*}
Taking the maximum for $|\io|\leq 1$ and combining with the following Koszul formula
\begin{align*}
    \ga_{IJK}=\frac{1}{2}(S_{IJK}+S_{KJI}+S_{KIJ}),
\end{align*}
we arrive at the desired lower-order estimate for $\ga$.
\end{proof}

\subsection{Estimate for \texorpdfstring{$\ec$}{}}
Next we derive the lower-order estimate of $\ec$.
\begin{prop}\label{esteom}
    The following estimate holds for $\ec$:
    \begin{align*}
       |t^{q_M}\ec^{(1)}|^2+\int_t^1s^{2q_M-1}|\ec^{(1)}|^2ds\les\ep_0^2+\ep^2\int_t^1s^{2q_M-1}|\Gag^{(1)}|^2ds+\int_t^1s^{-1+2\si}\Dd(s)^2ds.
    \end{align*}
\end{prop}
\begin{proof}
    Utilizing \eqref{2.23} and Proposition \ref{kasnerquantities}, we obtain
    \begin{align*}
        \pr_t(e_I^i)&=nk_{IC}e_C^i=-\frac{\qIt}{t}e_I^i+n\kc_{IC}\left(t^{-\qit}\de_C^i+\ec_C^i\right)=-\frac{\qIt}{t}e_I^i+t^{-\qit}\Gaw+\Gag\c\Gaw,\\
        \pr_t(\et_I^i)&=-\frac{\qIt}{t}(\et_I^i).
    \end{align*}
    Taking the difference, we then get the evolution equation for the linearized reduced variable $\ec_I^i$, i.e.,
    \begin{align}\label{Eqn:linear e}
    \pr_t(\ec_I^i)+\frac{\qIt}{t}(\ec_I^i)=t^{-\qit}\Gaw+\Gag\c\Gaw.
    \end{align}
    Commuting the above equation with $\pr^\io$ for $|\io|\leq 1$ and applying Proposition \ref{commutation}, we deduce
    \begin{align*}
    \pr_t(\pr^\io\ec_I^i)+\frac{\qIt}{t}(\pr^\io\ec_I^i)=t^{-\qit}\Gaw^{(1)}+(\Gag\c\Gaw)^{(1)}.
    \end{align*}
    Hence, employing Lemma \ref{evolution} with $\la=q_M$ and $\la_0=\qIt$ and combining with Lemma \ref{decayGagGabGaw} and \eqref{dfqM}, we conclude
    \begin{align*}
    |t^{q_M}\pr^\io\ec_I^i|^2+\int_t^1s^{2q_M-1}|\pr^\io\ec_I^i|^2ds&\les\ep_0^2+\int_t^1 s^{2q_M+1-2\qit}|\Gaw^{(1)}|^2ds+\int_t^1s^{2q_M+1}|(\Gag\c\Gaw)^{(1)}|^2ds\\
    &\les\ep_0^2+\int_t^1s^{-1+2\si}\Dd(s)^2ds+\ep^2\int_t^1s^{2q_M-1}|\Gag^{(1)}|^2ds.
    \end{align*}
This finishes the proof of Proposition \ref{esteom}.
\end{proof}

\subsection{Estimate for \texorpdfstring{$k$}{}}
We turn to control $k$ at the lower order.
\begin{prop}\label{estk}
For the second fundamental form $k$, the following estimate holds
\begin{align}\label{kcesteq}
    |t\kc^{(1)}|^2\les\ep_0^2+\int_t^1s^{-1+2\si}\Dd(s)^2ds.
\end{align}
\end{prop}
\begin{proof}
   From \eqref{2.22a} in Proposition \ref{basicequations}, we have
    \begin{align}\label{Eqn:prt k}
        \pr_t(k_{IJ})+\frac{1}{t}k_{IJ}&=-e_I(e_J\vphi)+e_C(\ga_{IJC})-e_I(\ga_{CJC})+O(t^{-1})\Gab+\Gab\c\Gaw+T.
    \end{align}
    Also, by Proposition \ref{kasnerquantities}, there holds
    \begin{align*}
        \pr_t(\kt_{IJ})+\frac{1}{t}\kt_{IJ}=0.
    \end{align*}
    Taking the difference, commuting with $\pr^\io$ for $|\io|\leq 1$ and applying Proposition \ref{commutation}, we thus obtain
    \begin{align*}
        \pr_t(\pr^\io\kc_{IJ})+\frac{1}{t}(\pr^\io\kc_{IJ})=-\left(\ev(\ev\vphi)\right)^{(1)}+(\ev\ga)^{(1)}+t^{-1}\Gab^{(1)}+(\Gab\c\Gaw)^{(1)}+T^{(1)}.
    \end{align*}
    Hence, applying \eqref{criticalevolution} in Lemma \ref{evolution}, along with Lemma \ref{decayGagGabGaw}, Lemma \ref{interpolation} and Lemma \ref{topT}, we deduce
    \begin{align*}
    |t\pr^\io\kc_{IJ}|^2&\les\ep_0^2+\int_t^1s^2\frac{\Dd(s)}{s}\left(\frac{\Dd(s)}{s^{2q}}+\frac{\Dd(s)}{s^{1+q}}\right)ds\les\ep_0^2+\int_t^1s^{-1+2\si}\Dd(s)^2ds.
    \end{align*}
    This concludes the proof of Proposition \ref{estk}.
\end{proof}
\subsection{Estimates for \texorpdfstring{$e_0\psi$}{} and \texorpdfstring{$\ev\psi$}{}}
We now establish the lower-order estimates for the time and spatial derivatives of the scalar field $\psi$.
\begin{prop}\label{estpsi}
The following estimates hold for $\widecheck{e_0\psi}$ and $\ev\psi$:
\begin{align}
\left|t\left(\widecheck{e_0\psi}\right)^{(1)}\right|^2&\les\ep_0^2+\int_t^1s^{-1+2\si}\Dd(s)^2ds,\label{este0psi}\\
\left|t^q\left(\ev\psi\right)^{(1)}\right|^2+\int_t^1s^{2q-1}\left|\left(\ev\psi\right)^{(1)}\right|^2ds&\les\ep_0^2+\ep^2\int_t^1s^{2q-1}|\Gab^{(1)}|^2ds+\int_t^1s^{-1+2\si}\Dd(s)^2ds.\label{estevpsi}
\end{align}
\end{prop}
\begin{proof}
In view of \eqref{2.24} in Proposition \ref{basicequations}, we obtain
\begin{align*}
    e_0(e_0\psi)+\frac{1}{t}e_0\psi=e_C(e_C\psi)+\Gab\c\Gab=e_C(e_C\psi)+\Gab\c\Gab.
\end{align*}
This implies
\begin{align*}
    \pr_t(e_0\psi)+\frac{n}{t}e_0\psi=O(\ev(\ev\psi))+\Gab\c\Gab.
\end{align*}
Note that from Proposition \ref{kasnerquantities} we have
\begin{align*}
    \pr_t\left(\widetilde{e_0\psi}\right)+\frac{1}{t}\widetilde{e_0\psi}=\pr_t\left(\frac{\Bt}{t}\right)+\frac{\Bt}{t^2}=0.
  \end{align*}
Subtracting these two equations, then commuting with $\pr^\io$ for $|\io| \leq 1$ and applying Proposition \ref{commutation}, we infer
\begin{align}\label{prtprtpsi}
    \pr_t(\pr^\io\widecheck{e_0\psi})+\frac{1}{t}(\pr^\io\widecheck{e_0\psi})=\left(\ev(\ev\psi)\right)^{(1)}+(\Gab\c\Gaw)^{(1)}.
\end{align}
Applying \eqref{criticalevolution} in Lemma \ref{evolution}, together with Lemma \ref{decayGagGabGaw} and Lemma \ref{interpolation}, we hence derive
\begin{align*}
    |t\widecheck{e_0\psi}|^2\les\ep_0^2+\int_t^1 s^2\frac{\Dd(s)}{s}\left(\frac{\Dd(s)}{s^{2q}}+\frac{\ep\Dd(s)}{s^{q+1}}\right)ds\les\ep_0^2+\int_t^1s^{-1+2\si}\Dd(s)^2,
\end{align*}
which is exactly \eqref{este0psi}. 
\vspace{2mm}

Next, using \eqref{commutation2} we compute
\begin{align*}
    \pr_t(e_I\psi)&=e_I(\pr_t\psi)+[\pr_t,e_I]\psi=e_I(\pr_t\psi)+nk_{IC}e_C\psi\\
    &=e_I(ne_0\psi)-\frac{n\qIt}{t}e_I\psi+\Gab\c\Gaw\\
    &=-\frac{\qIt}{t}e_I\psi+t^{-1}O(\ev\vphi)+t^{-q_M}\Gaw^{(1)}+t^{-q+2\si}\Gaw+\Gab\c\Gaw.
\end{align*}
Commuting with $\pr^\io$ for $|\io|\leq 1$ and employing Proposition \ref{commutation}, we then deduce
\begin{align*}
    \pr_t(\pr^\io e_I\psi)+\frac{\qIt}{t}(\pr^\io e_I\psi)=t^{-1}(\ev\vphi)^{(1)}+t^{-q_M}\Gaw^{(2)}+t^{-q+2\si}\Gaw^{(1)}+(\Gab\c\Gaw)^{(1)}.
\end{align*}
Consequently, utilizing Lemma \ref{evolution} with $\la=q$ and $\la_0=\qIt$, and incorporating with Lemma \ref{decayGagGabGaw} and Lemma \ref{interpolation}, we arrive at
\begin{align*}
    |t^q\pr^\io e_I\psi|^2+\int_t^1s^{2q-1}|\pr^\io e_I\psi|^2ds&\les\ep_0^2+\int_t^1s^{2q+1}\left(\frac{\Dd(s)^2}{s^{2q_M+2+2\dq}}+\frac{\Dd(s)^2}{s^{2q+2-2\si}}+\frac{\ep^2|\Gab^{(1)}|^2}{s^{2}}\right)ds\\
    &\les\ep_0^2+\ep^2\int_t^1s^{2q-1}|\Gab^{(1)}|^2ds+\int_t^1s^{-1+2\si}\Dd(s)^2ds,
\end{align*}
as stated in \eqref{estevpsi}.
\end{proof}
\subsection{Estimate for \texorpdfstring{$F$}{}}
In a similar manner to the previous analysis, we proceed to estimate the lower-order norm of the Maxwell field $F$.
\begin{prop}\label{estF}
    For the Maxwell field $F$, the following estimate holds
    \begin{align*}
        |t^qF^{(1)}|^2+\int_t^1s^{2q-1}|F^{(1)}|^2ds&\les\ep_0^2+\ep^2\int_t^1s^{2q-1}|\Gab^{(1)}|^2ds+\int_t^1s^{-1+2\si}\Dd(s)^2ds.
    \end{align*}
\end{prop}
\begin{proof}
We first write \eqref{Maxwell1} in Proposition \ref{basicequations} as
\begin{align*}
    \pr_t(F_{0I})+\frac{1-\qIt}{t}F_{0I}=O(\ev F)+\Gab\c\Gaw.
\end{align*}
Commuting with $\pr^\io$ for $|\io|\leq 1$ and applying Proposition \ref{commutation}, we deduce
\begin{align*}
    \pr_t(\pr^\io F_{0I})+\frac{1-\qIt}{t}(\pr^\io F_{0I})=(\ev F)^{(1)}+(\Gab\c\Gaw)^{(1)}.
\end{align*}
Now, employing Lemma \ref{evolution} with $\la=q$ and $\la_0=1-\qIt$, in view of Lemma \ref{decayGagGabGaw} and Lemma \ref{interpolation}, we thus derive
\begin{align*}
|t^qF_{0I}|^2+\int_t^1s^{2q-1}|F_{0I}|^2ds&\les\ep_0^2+\int_t^1s^{2q+1}\left(\frac{\Dd(s)^2}{s^{4q}}+\frac{\ep^2|\Gab^{(1)}|^2}{s^2}\right)ds\\
&\les\ep_0^2+\int_t^1s^{-1+2\si}\Dd(s)^2ds+\ep^2\int_t^1s^{2q-1}|\Gab^{(1)}|^2ds.
\end{align*}
This gives the desired estimate for $F_{0I}$. 

Note that from \eqref{Maxwell2} in Proposition \ref{basicequations},  $F_{IJ}$ also satisfies the equation in the form
\begin{equation*}
     \pr_t(F_{IJ})+\frac{\qIt+\qJt}{t}F_{IJ}=O(\ev F)+\Gab\c\Gaw.
\end{equation*}
The corresponding estimate for $F_{IJ}$ hence follows analogously.
\end{proof}
\subsection{Estimate for \texorpdfstring{$\TT$}{}}
The lower-order norm of the Vlasov part $\TT$ is controlled in a different way, based on a new approach via the conservation law for the associated energy-momentum tensor $T$.
\begin{prop}\label{estT}
The following estimate holds for $\TT$:
\begin{align*}
t^{1+q_M}|\TT|^2+\int_{t}^1s^{q_M}|\TT|^2ds\les\ep_0^2+\int_t^1s^{-1+2\si}\Dd(s)^2ds.
\end{align*}
\end{prop}
\begin{proof}
We start with the conservation law for $T_{\mu\nu}=T^{(V)}_{\mu\nu}$:
    \begin{align*}
        \D_0T_{00}=\D_CT_{C0}.
    \end{align*}
   This can be expanded to
    \begin{align*}
    e_0(T_{00})-2T(\D_{e_0}e_0,e_0)=e_C(T_{C0})-T(\D_{e_C}e_C,e_0)-T(e_C,\D_{e_C}e_0).
    \end{align*}
   Combining with \eqref{De0e0}, \eqref{DeIe0} and \eqref{DeIeJ}, we obtain
    \begin{align*}
        e_0(T_{00})-2(e_C\vphi)T_{0C}=e_C(T_{0C})+(\tr k)T_{00}-\ga_{CCD} T_{0D}+k_{CD}T_{CD},
    \end{align*}
   which can be systematically written as
    \begin{align}\label{conservationlaw1}
        e_0(T_{00})+\frac{1}{t}T_{00}+\sum_{I=1}^3\frac{\qIt}{t}T_{II}=e_C(T_{0C})+\Gaw\c T.
    \end{align}
    Observe that from \eqref{dfT} we have that
\begin{align*}
    T_{\mu\mu}=\int_{P(t,x)}f(p_\mu)^2\dvol\geq 0,\qquad\quad T_{00}=\sum_{I=1}^3T_{II}.
\end{align*}
Together with \eqref{conservationlaw1}, this implies
    \begin{align}\label{conservationlaw}
        e_0(T_{00})+\frac{1+\max_{I=1,2,3}\{\qIt\}}{t}T_{00}\geq e_C(T_{0C})+\Gaw\c T.
    \end{align}
    Moreover, using Lemma \ref{topT} we estimate
    \begin{align*}
        |e_C(T_{0C})|\les t^{-q_M}\|T\|_{W^{1,\infty}(\Si_t)}\les\frac{\Dd(t)^2}{t^{1+q+q_M}}.
    \end{align*}
    Now, multiplying \eqref{conservationlaw} by $nt^{1+q_M}$ and integrating it from $t$ to 1, we hence infer
    \begin{align*}
        t^{1+q_M}T_{00}+\int_{t}^1 s^{q_M}T_{00}ds&\les\ep_0^2+\int_t^1\left(s^{-q}\Dd(s)^2+\ep s^{q_M}T_{00}\right) ds\\
        &\les\ep_0^2+\int_t^1\left(s^{-1+2\si}\Dd(s)^2+\ep s^{q_M}T_{00}\right)ds.
    \end{align*}
    Here we employ Lemma \ref{decayGagGabGaw}, Lemma \ref{interpolation} and Remark \ref{rkTT} to control the terms on the right. 

    Thus, by picking $\ep>0$ small enough, we deduce
    \begin{align*}
        t^{1+q_M}T_{00}+\int_{t}^1 s^{q_M}T_{00}ds\les \ep_0^2+\int_t^1s^{-1+2\si}\Dd(s)^2.
    \end{align*}
    Recalling from Remark \ref{rkTT} that $T_{00}=\TT^2$, this concludes the proof of Proposition \ref{estT}.
\end{proof}
\subsection{End of the Proof of Theorem \ref{M2}}
We are prepared to establish Theorem \ref{M2}. Collecting Propositions \ref{estga}--\ref{estT} above, we derive
    \begin{align*}
        &\Ll(t)^2+\int_t^1\left(s^{2q_M-1}|\Gag^{(1)}|^2+s^{2q-1}|\Gab^{(1)}|^2\right)ds\\
        \les\;&\ep_0^2+\ep^2\int_t^1\left(s^{2q_M-1}|\Gag^{(1)}|^2+s^{2q-1}|\Gab^{(1)}|^2\right)ds+\int_t^1s^{-1+2\si}\Dd(s)^2ds,
    \end{align*}
Therefore, by choosing $\ep>0$ sufficiently small, we conclude the desired inequality
    \begin{align*}
        \Ll(t)^2+\int_t^1\left(s^{2q_M-1}|\Gag^{(1)}|^2+s^{2q-1}|\Gab^{(1)}|^2\right)ds\les\ep_0^2+\int_t^1s^{-1+2\si}\Dd(s)^2ds.
    \end{align*}
This completes the proof of Theorem \ref{M2}.
\section{Top-Order Estimates for \texorpdfstring{$\Dd(t)$}{}}\label{secenergy}
The goal of this section is to prove the top-order energy estimates stated in Theorem \ref{M3}. To this end, we conduct the $t$-weighted energy estimates for the Bianchi pairs $(k, \ga), (e_0\psi, \ev \psi), (F_{0I}, F_{IJ})$. The Vlasov field will be handled separately. In addition, we establish \Cref{M4}, which validates the continuation criterion for the EMSVS in Theorem \ref{Thm:local}.
\subsection{Estimates for \texorpdfstring{$k$}{} and \texorpdfstring{$\ga$}{}}
We begin with estimating the Bianchi pair $(k, \ga)$. 
\begin{prop}\label{estkga}
The following estimate holds for $k$ and $\ga$:
\begin{align*}
&\max_{1\leq|\io|\leq\kl}\sum_{I,J,K}\int_{\Si_t}t^{2A_*+2}\left(|\pr^\io k_{IJ}|^2+|\pr^\io\ga_{IJK}|^2\right)\\
+&A_*\max_{1\leq|\io|\leq\kl}\sum_{I,J,K}\int_t^1\int_{\Si_s}s^{2A_*+1}\left(|\pr^\io k_{IJ}|^2+|\pr^\io\ga_{IJK}|^2\right)ds\\
\les\;&\ep_0^2+\int_t^1\int_{\Si_s}s^{-1}\Hh(s)^2ds+\int_{t}^1s^{-1+2\si}\Dd(s)^2ds.
\end{align*}
\end{prop}
\begin{proof}
Recall from Proposition \ref{basicequations} the evolution equations for $k$ and $\ga$, namely \eqref{2.22a}, \eqref{2.22b}, and the momentum constraint \eqref{2.26b}, which schematically read:
\begin{align*}
    e_0(k_{IJ})+t^{-1}k_{IJ}&=-e_I(e_J\vphi)+e_C(\ga_{IJC})-e_I(\ga_{CJC})+\Gab\c\Gab+T,\\
    e_0(\ga_{IJK})&=e_K(k_{JI})-e_J(k_{KI})+t^{-1}\Gab+\Gab\c\Gaw,\\
    e_Ck_{CI}&=t^{-1}\Gab+\Gab\c\Gaw+T.
\end{align*}
Commuting these equations with $\pr^\io$ and applying Proposition \ref{commutation}, we infer
\begin{align}
    e_0(\pr^\io k_{IJ})+t^{-1}(\pr^\io k_{IJ})&=-e_I(\pr^\io e_J\vphi)+e_C(\pr^\io\ga_{IJC})-e_I(\pr^\io\ga_{CJC})\nonumber\\
    &+t^{-1}\Gab^{(\io)}+(\Gab^{(1)}\c\Gaw^{(1)})^{(\io-1)}+T^{(\io)},\label{kenid}\\
    e_0(\pr^\io\ga_{IJK})&=e_K(\pr^\io k_{JI})-e_J(\pr^\io k_{KI})+t^{-1}\Gab^{(\io)}+(\Gab^{(1)}\c\Gaw^{(1)})^{(\io-1)},\label{gaenid}\\
    e_C(\pr^\io k_{CI})&=t^{-1}\Gab^{(\io)}+(\Gab^{(1)}\c\Gaw^{(1)})^{(\io-1)}+T^{(\io)}\label{kdiv}.
\end{align}
Multiplying \eqref{kenid} and \eqref{gaenid} by $2\pr^\io k_{IJ}$ and $\pr^\io\ga_{IJK}$ respectively yields
\begin{align*}
e_0\left(|\pr^\io k_{IJ}|^2\right)+\frac{2}{t}|\pr^\io k_{IJ}|^2&=-2\pr^\io k_{IJ}e_I(\pr^\io e_J\vphi)+2\pr^\io k_{IJ}e_C(\pr^\io\ga_{IJC})-2\pr^\io k_{IJ}e_I(\pr^\io\ga_{CJC})\\
&+t^{-1}\Gab^{(\io)}\c\Gaw^{(\io)}+\Gaw^{(\io)}\c(\Gab^{(1)}\c\Gaw^{(1)})^{(\io-1)}+T^{(\io)},\\
\frac{1}{2}e_0\left(|\pr^\io\ga_{IJC}|^2\right)&=(\pr^\io\ga_{IJC})e_C(\pr^\io k_{IJ})+(\pr^\io\ga_{ICJ})e_J(\pr^\io k_{IC})\\
&+t^{-1}\Gab^{(\io)}\c\Gab^{(\io)}+\Gab^{(\io)}\c(\Gab^{(1)}\c\Gaw^{(1)})^{(\io-1)}+T^{(\io)}.
\end{align*}
Summing the above two equations and utilizing \eqref{kdiv}, we deduce
\begin{align*}
&e_0\left(|\pr^\io k_{IJ}|^2+\frac{1}{2}|\pr^\io\ga_{IJK}|^2\right)+\frac{2}{t}|\pr^\io k_{IJ}|^2\\
=&-2e_I\left(\pr^\io k_{IJ}\pr^\io e_J\vphi+\pr^\io k_{IJ}\pr^\io\ga_{CJC}\right)+2e_I(\pr^\io k_{IJ})\left(\pr^\io e_J\vphi+\pr^\io\ga_{CJC}\right)+2e_C(\pr^\io k_{IJ}\pr^\io\ga_{IJC})\\
&+t^{-1}\Gab^{(\io)}\c\Gaw^{(\io)}+\Gaw^{(\io)}\c(\Gab^{(1)}\c\Gaw^{(1)})^{(\io-1)}+\Gaw^{(\io)}\c T^{(\io)}\\
=&-2e_I\left(\pr^\io k_{IJ}\pr^\io e_J\vphi+\pr^\io k_{IJ}\pr^\io\ga_{CJC}\right)+2e_C(\pr^\io k_{IJ}\pr^\io\ga_{IJC})\\
&+t^{-1}\Gab^{(\io)}\c\Gaw^{(\io)}+\Gaw^{(\io)}\c(\Gab^{(1)}\c\Gaw^{(1)})^{(\io-1)}+\Gaw^{(\io)}\c T^{(\io)}.
\end{align*}
Multiplying by $t^{2A_*+2}n$ and then integrating over $\Si_t$, we thus derive
\begin{align}
\begin{split}\label{prtkga}
&\pr_t\left(\int_{\Si_t}t^{2A_*+2}|\pr^\io k_{IJ}|^2+\frac{1}{2}t^{2A_*+2}|\pr^\io\ga_{IJK}|^2\right)\\
-&2A_*\int_{\Si_t}t^{2A_*+1}|\pr^\io k_{IJ}|^2-(2A_*+2)\int_{\Si_t}t^{2A_*+1}|\pr^\io\ga_{IJK}|^2\\
=&\int_{\Si_t}t^{2A_*+1}\Gab^{(\io)}\c\Gaw^{(\io)}+t^{2A_*+2}\Gaw^{(\io)}\c(\Gab^{(1)}\c\Gaw^{(1)})^{(\io-1)}+t^{2A_*+2}\Gaw^{(\io)}\c T^{(\io)}.
\end{split}
\end{align}
To bound the terms on the right, we appeal to Lemma \ref{decayGagGabGaw} and Lemma \ref{topT} and obtain
\begin{align*}
    \int_{\Si_t}t^{2A_*+1}\Gab^{(\io)}\c\Gaw^{(\io)}&\les t^{-1}\|t^{A_*+1}\Gab^{(\io)}\|_{L^2(\Si_t)}\|t^{A_*+1}\Gaw^{(\io)}\|_{L^2(\Si_t)}\\
    &\les t^{-1}\Hh(t)^2,\\
    \int_{\Si_t}t^{2A_*+2}\Gaw^{(\io)}\c(\Gab^{(1)}\c\Gaw^{(1)})^{(\io-1)}&\les_*\ep t^{-q}\|t^{A_*+1}\Gaw^{(\io)}\|_{L^2(\Si_t)}\|t^{A_*+1}\Gaw^{(\io)}\|_{L^2(\Si_t)}\\
    &+\ep t^{-1}\|t^{A_*+1}\Gab^{(\io)}\|_{L^2(\Si_t)}\|t^{A_*+1}\Gaw^{(\io)}\|_{L^2(\Si_t)},\\
    &\les t^{-1}\Hh(t)^2\\
    \int_{\Si_t}t^{2A_*+2}\Gaw^{(\io)}\c T^{(\io)}&\les t^{-q}\|t^{A_*+1}\Gaw^{(\io)}\|_{L^2(\Si_t)}\|t^{A_*+1+q}T^{(\io)}\|_{L^2(\Si_t)}\\
    &\les_* t^{-q}\Hh(t)\Dd(t)^2\\
    &\les t^{-1+2\si}\Dd(t)^2.
\end{align*}
Consequently, it follows
\begin{align}\label{Est:kgnonlinear}
&\int_{\Si_t}t^{2A_*+1}\Gab^{(\io)}\c\Gaw^{(\io)}+t^{2A_*+2}\Gaw^{(\io)}\c(\Gab^{(1)}\c\Gaw^{(1)})^{(\io-1)}+t^{2A_*+2}\Gaw^{(\io)}\c T^{(\io)}\\
\les\;& t^{-1}\Hh(t)^2+t^{-1+2\si}\Dd(t)^2.
\end{align}
Finally, integrating \eqref{prtkga} from $t$ to $1$ and taking the maximum for $1\leq|\io|\leq\kl$, combining with \eqref{Est:kgnonlinear}, we conclude
\begin{align*}
&\max_{1\leq|\io|\leq\kl}\sum_{I,J}\int_{\Si_t}t^{2A_*+2}|\pr^\io k_{IJ}|^2+\max_{1\leq|\io|\leq\kl}\sum_{I,J,K}\int_{\Si_t}t^{2A_*+2}|\pr^\io\ga_{IJK}|^2\\
&+A_*\max_{1\leq|\io|\leq\kl}\int_t^1\int_{\Si_s}s^{2A_*+1}\left(\sum_{I,J}|\pr^\io k_{IJ}|^2+\sum_{I,J,K}|\pr^\io\ga_{IJK}|^2\right)ds\\
\les\;&\ep_0^2+\int_t^1s^{-1}\Hh(s)^2ds+\int_t^1s^{-1+2\si}\Dd(s)^2ds.
\end{align*}
as desired.
\end{proof}
\subsection{Estimate for \texorpdfstring{$\ec$}{}}
We proceed to control the top-order norm of $e^i_I$. Since the evolution equation for $e^i_I$ does not involve any loss of derivatives, we can directly perform the standard weighted energy estimate for $e^i_I$.
\begin{prop}\label{estecomc}
The following estimate holds for $e^i_I$:
\begin{align*}
&\max_{1\leq|\io|\leq\kl}\sum_{I,i=1}^3\int_{\Si_t}|t^{A_*+q_M}\pr^\io\ec_I^i|^2+A_*\max_{1\leq|\io|\leq\kl}\sum_{I,i=1}^3\int_t^1\int_{\Si_s}s^{-1}|s^{A_*+q_M}\pr^\io\ec_I^i|^2ds\\
\les\;&\ep_0^2+\int_t^1s^{-1}\Hh(s)^2ds+\int_t^1s^{-1+2\si}\Dd(s)^2ds.
\end{align*}
\end{prop}
\begin{proof}
Recall \eqref{Eqn:linear e} as derived in the proof of Proposition \ref{esteom}:
 \begin{align*}
    \pr_t(\ec_I^i)+\frac{\qIt}{t}(\ec_I^i)=t^{-\qit}\Gaw+\Gag\c\Gaw.
    \end{align*}
Commuting this with $\pr^\io$ and applying Proposition \ref{commutation}, we obtain
\begin{align*}
    \pr_t(\pr^\io\ec_I^i)+\frac{\qIt}{t}(\pr^\io\ec_I^i)=t^{-q_M}\Gab^{(\io)}+(\Gag\c\Gaw)^{(\io)}.
\end{align*}
Multiplying by $t^{2A_*+2q_M}\pr^\io\ec_I^i$ gives
\begin{align}
\begin{split}\label{prte}
&\frac{1}{2}\pr_t\left(|t^{A_*+q_M}\pr^\io\ec_I^i|^2\right)-\frac{A_*+q_M-\qIt}{t}|t^{A_*+q_M}\pr^\io\ec_I^i|^2\\
=\;&t^{2A_*+q_M}\Gag^{(\io)}\c\Gab^{(\io)}+t^{2A_*+2q_M}\Gag^{(\io)}\c(\Gag\c\Gaw)^{(\io)}.
\end{split}
\end{align}
Thus, integrating over $\Si_t$ and then in $s$ from $t$ to $1$, we deduce
\begin{align*}
&\sum_{I,i=1}^3\left(\int_{\Si_t}|t^{A_*+q_M}\pr^\io\ec_I^i|^2+A_*\int_t^1\int_{\Si_s}s^{-1}|s^{A_*+q_M}\pr^\io\ec_I^i|^2ds\right)\\
\les\;&\ep_0^2+\int_t^1 \left(s^{2A_*+q_M}\Gag^{(\io)}\c\Gab^{(\io)}+s^{2A_*+2q_M}\Gag^{(\io)}\c(\Gag\c\Gaw)^{(\io)}\right)ds.
\end{align*}
Note that from Lemma \ref{decayGagGabGaw} we have
\begin{equation}\label{Est:e nonlinear 2}
\begin{aligned}
\int_{\Si_t}t^{2A_*+q_M}\Gag^{(\io)}\c\Gab^{(\io)}&\les t^{-q}\|t^{A_*+q_M}\Gag^{(\io)}\|_{L^2(\Si_t)}\|t^{A_*+q}\Gab^{(\io)}\|_{L^2(\Si_t)},\\
&\les t^{-1+2\si}\Dd(t)^2,\\
\int_{\Si_t}t^{2A_*+2q_M}\Gag^{(\io)}\c(\Gag\c\Gaw)^{(\io)}&\les_*\ep t^{-1}\|t^{A_*+q_M}\Gag^{(\io)}\|_{L^2(\Si_t)}\|t^{A_*+q_M}\Gag^{(\io)}\|_{L^2(\Si_t)}\\
&+\ep t^{-q_M}\|t^{A_*+q_M}\Gag^{(\io)}\|_{L^2(\Si_t)}\|t^{A_*+1}\Gaw^{(\io)}\|_{L^2(\Si_t)}\\
&\les t^{-1}\Hh(t)^2.
\end{aligned}
\end{equation}
Combining with \eqref{Est:e nonlinear 2}, this renders
\begin{align*}
&\sum_{I,i=1}^3\left(\int_{\Si_t}|t^{A_*+q_M}\pr^\io\ec_I^i|^2+A_*\int_t^1\int_{\Si_s}s^{-1}|s^{A_*+q_M}\pr^\io\ec_I^i|^2ds\right)\\
\les\;&\ep_0^2+\int_t^1s^{-1}\Hh(s)^2ds+\int_t^1s^{-1+2\si}\Dd(s)^2ds.
\end{align*}
Taking the maximum for $1\leq|\io|\leq\kl$, we finish thus the proof of Proposition \ref{estecomc}.
\end{proof}
\subsection{Estimates for \texorpdfstring{$e_0\psi$}{} and \texorpdfstring{$\ev\psi$}{}}
Next we move to derive the top-order estimates for the scalar field $\psi$, which satisfies the wave equation $\square_{\g}\psi=0$.
\begin{prop}\label{estepsi}
The following estimate holds for $e_0\psi$ and $\ev\psi$:
\begin{align*}
    &\max_{1\leq|\io|\leq\kl}\int_{\Si_t}|t^{A_*+1}\pr^\io e_0\psi|^2+|t^{A_*+1}\pr^\io\ev\psi|^2+A_*\max_{1\leq|\io|\leq\kl}\int_t^1\int_{\Si_s}s^{2A_*+1}\left(|\pr^\io e_0\psi|^2+|\pr^\io\ev\psi|^2\right)ds\\
    \les\;&\ep_0^2+\int_t^1s^{-1}\Hh(s)^2ds+\int_t^1s^{-1+2\si}\Dd(s)^2ds.
\end{align*}
\end{prop}
\begin{proof}
From \eqref{2.24} in Proposition \ref{basicequations} we have
\begin{align*}
 ne_0(e_0\psi)+t^{-1}ne_0\psi=ne_C(e_C\psi)+\Gab\c\Gab.
\end{align*}
    Commuting with $\pr^\io$ implies
    \begin{align*}
        \pr_t(\pr^\io e_0\psi)+t^{-1}n(\pr^\io e_0\psi)=ne_C(\pr^\io e_C\psi)+\left(\Gab^{(1)}\c\Gaw^{(1)}\right)^{(\io-1)}.
    \end{align*}
    Multiplying it by $t^{2A_*+2}\pr^\io e_0\psi$, we get
    \begin{align*}
        &\quad\;\,\frac{1}{2}\pr_t\left(|t^{A_*+1}\pr^\io e_0\psi|^2\right)+t^{2A_*+1}|\pr^\io e_0\psi|^2\\
        &=t^{2A_*+2}(\pr^\io e_0\psi)(ne_C(\pr^\io e_C\psi))+t^{2A_*+2}\Gaw^{(\io)}\c\left(\Gab^{(1)}\c\Gaw^{(1)}\right)^{(\io-1)}\\
        &=t^{2A_*+2}ne_C\left(\pr^\io e_0\psi(\pr^\io e_C\psi)\right)-t^{2A_*+2}ne_C(\pr^\io e_0\psi)\pr^\io e_C\psi+t^{2A_*+2}\Gaw^{(\io)}\c\left(\Gab^{(1)}\c\Gaw^{(1)}\right)^{(\io-1)}\\
        &=t^{2A_*+2}ne_C\left(\pr^\io e_0\psi(\pr^\io e_C\psi)\right)-\frac{1}{2}t^{2A_*+2}\pr_t(|\pr^\io\ev\psi|^2)+t^{2A_*+2}\Gaw^{(\io)}\c\left(\Gab^{(1)}\c\Gaw^{(1)}\right)^{(\io-1)}.
    \end{align*}
Integrating over $\cup_{s\in[t, 1]} \{s\}\times\Si_s$, we hence deduce
\begin{align*}
&\int_{\Si_t}|t^{A_*+1}\pr^\io e_0\psi|^2+|t^{A_*+1}\pr^\io\ev\psi|^2+A_*\int_t^1\int_{\Si_s}s^{2A_*+1}\left(|\pr^\io e_0\psi|^2+|\pr^\io\ev\psi|^2\right)ds\\
\les\;&\ep_0^2+\int_t^1s^{2A_*+2}\Gaw^{(\io)}\c\left(\Gab^{(1)}\c\Gaw^{(1)}\right)^{(\io-1)}ds\\
\les\;&\ep_0^2+\ep C_*\int_t^1s^{-1+2\si}\Dd(s)^2ds+\ep C_*\int_t^1\int_{\Si_s}s^{2A_*+1}\Gaw^{(\io)}\c\Gab^{(\io)}ds\\
\les\;&\ep_0^2+\int_t^1s^{-1+2\si}\Dd(s)^2ds+\int_t^1s^{-1}\|s^{A_*+1}\Gaw^{(\io)}\|_{L^2(\Si_s)}\|s^{A_*+1}\Gab^{(\io)}\|_{L^2(\Si_s)}ds\\
\les\;&\ep_0^2+\int_t^1s^{-1+2\si}\Dd(s)^2ds+\int_t^1s^{-1}\Hh(s)^2ds.
\end{align*}
This completes the proof of Proposition \ref{estepsi} by taking the maximum for $1\leq|\io|\leq\kl$.
\end{proof}
\subsection{Estimate for \texorpdfstring{$F$}{}}
Now we estimate the last Bianchi pair $(F_{0I}, F_{IJ})$ via the Maxwell equation.
\begin{prop}\label{topF}
    The following estimate holds for the Maxwell field $F$:
    \begin{align*}
    &\max_{1\leq|\io|\leq\kl}\int_{\Si_t}|t^{A_*+1}\pr^\io F|^2+A_*\max_{1\leq|\io|\leq\kl}\int_t^1\int_{\Si_s}|t^{A_*+1}\pr^\io F|^2ds\\
    \les\;&\ep_0^2+\int_t^1s^{-1}\Hh(s)^2ds+\int_t^1s^{-1+2\si}\Dd(s)^2ds.
    \end{align*}
\end{prop}
\begin{proof}
Employing \eqref{Maxwell1} and \eqref{Maxwell2} in Proposition \ref{basicequations}, we get
\begin{align*}
e_0(F_{0I})+\frac{1-\qIt}{t}F_{0I}&=e_C(F_{CI})+\Gab\c\Gaw,\\
e_0(F_{IJ})+\frac{\qIt+\qJt}{t}F_{IJ}&=e_I(F_{0J})-e_J(F_{0I})+\Gab\c\Gaw.
\end{align*}
Commuting with $\pr^\io$, we deduce
\begin{align*}
e_0(\pr^\io F_{0I})+\frac{1-\qIt}{t}(\pr^\io F_{0I})&=e_C(\pr^\io F_{CI})+\left(\Gab^{(1)}\c\Gaw^{(1)}\right)^{(\io-1)},\\
e_0(\pr^\io F_{IJ})+\frac{\qIt+\qJt}{t}(\pr^\io F_{IJ})&=e_I(\pr^\io F_{0J})-e_J(\pr^\io F_{0I})+\left(\Gab^{(1)}\c\Gaw^{(1)}\right)^{(\io-1)}.
\end{align*}
Multiplying these two equations by $2\pr^\io F_{0I}$ and $\pr^\io F_{IJ}$ and taking the sum in $I$ and $I,J$ respectively, we obtain\footnote{Note here we utilize the fact that $F_{IJ}=-F_{JI}$.}
\begin{align*}
\sum_{I} e_0\left(|\pr^\io F_{0I}|^2\right)+\sum_{I}\frac{2(1-\qIt)}{t}|\pr^\io F_{0I}|^2&=\sum_{I} 2(\pr^\io F_{0I})e_C(\pr^\io F_{CI})+\Gab^{(\io)}\c\left(\Gab^{(1)}\c\Gaw^{(1)}\right)^{(\io-1)},\\
\sum_{I, J}\frac{1}{2}e_0(|\pr^\io F_{IJ}|^2)+\sum_{I, J}\frac{\qIt+\qJt}{t}|\pr^\io F_{IJ}|^2&=\sum_{I, J}2(\pr^\io F_{IJ})e_I(\pr^\io F_{0J})+\Gab^{(\io)}\c\left(\Gab^{(1)}\c\Gaw^{(1)}\right)^{(\io-1)}.
\end{align*}
Multiplying by $t^{2A_*+2}$ then implies
\begin{align*}
&\sum_{I}e_0\left(|t^{A_*+1}\pr^\io F_{0I}|^2\right)-\sum_{I}\frac{2(A_*+1-(1-\qIt))}{t}|t^{A_*+1}\pr^\io F_{0I}|^2\\
=\;&\sum_{I}2(t^{A_*+1}\pr^\io F_{0I})e_C(t^{A_*+1}\pr^\io F_{CI})+t^{2A_*+2}\Gab^{(\io)}\c\left(\Gab^{(1)}\c\Gaw^{(1)}\right)^{(\io-1)},\\
&\sum_{I, J}\frac{1}{2}e_0(|t^{A_*+1}\pr^\io F_{IJ}|^2)-\sum_{I, J}\frac{A_*+1-(\qIt+\qJt)}{t}|t^{A_*+1}\pr^\io F_{IJ}|^2\\
=\;&\sum_{I, J}2(t^{A_*+1}\pr^\io F_{IJ})e_I(t^{A_*+1}\pr^\io F_{0J})+t^{2A_*+2}\Gab^{(\io)}\c\left(\Gab^{(1)}\c\Gaw^{(1)}\right)^{(\io-1)}.
\end{align*}
Adding them up, we thus derive
\begin{align*}
    &e_0\left(\sum_{I}|t^{A_*+1}\pr^\io F_{0I}|^2+\sum_{I,J}\frac{1}{2}|t^{A_*+1}\pr^\io F_{IJ}|^2\right)\\
    -&\sum_{I}\frac{2(A_*+1-(1-\qIt))}{t}|t^{A_*+q}\pr^\io F_{0I}|^2-\sum_{I,J}\frac{A_*+1-(\qIt+\qJt)}{t}|t^{A_*+1}\pr^\io F_{IJ}|^2\\
    =\;&2e_C\left(t^{2A_*+2}\pr^\io F_{0D}(\pr^\io F_{CD})\right)+t^{2A_*+2}\Gab^{(\io)}\c\left(\Gab^{(1)}\c\Gaw^{(1)}\right)^{(\io-1)}.
\end{align*}
Therefore, integrating it first on $\Si_t$ and then in $s$ from $t$ to $1$, we obtain
\begin{align*}
&\sum_{I,J}\int_{\Si_t}|t^{A_*+1}\pr^\io F_{0I}|^2+|t^{A_*+1}\pr^\io F_{IJ}|^2+A_*\sum_{I,J}\int_t^1\int_{\Si_s}|t^{A_*+1}\pr^\io F_{0I}|^2+|t^{A_*+1}\pr^\io F_{IJ}|^2ds\\
\les\;&\ep_0^2+\int_t^1\int_{\Si_s}s^{2A_*+2}\Gag^{(\io)}\c\left(\Gab^{(1)}\c\Gaw^{(1)}\right)^{(\io-1)}ds\\
\les\;&\ep_0^2+\ep C_*\int_t^1s^{-1}\Hh(s)^2ds+\ep C_*\int_t^1s^{-1+2\si}\Dd(s)^2ds\\
\les\;&\ep_0^2+\int_t^1s^{-1}\Hh(s)^2ds+\int_t^1s^{-1+2\si}\Dd(s)^2ds.
\end{align*}
Taking the maximum for $1\leq|\io|\leq\kl$ then concludes the proof of Proposition \ref{topF}.
\end{proof}
\subsection{Estimate for \texorpdfstring{$\TT$}{}}
In this subsection, we aim to derive the top-order estimate for the Vlasov field, i.e., $\TT^{(k_*)}$. Here we recall from \eqref{io1io2df} and \eqref{dfTTiota} that 
\begin{align}\label{Def:TT new}
f^{(\io_1,\io_2)}\coloneqq \pr_x^{\io_1}(p\pr_p)^{\io_2}f, \qquad \quad
    \TT^{(\io_1,\io_2)}(t,x)\coloneqq &\left\|(p^0)^\frac{1}{2}\f^{(\io_1,\io_2)}(t,x,p)\right\|_{L^2_p(\mathbb{R}^3)}.
\end{align}
\begin{prop}\label{topf}
The following estimate holds for $\TT^{(k_*)}$:
\begin{align*}
    &t^{2A_*+2\de q+q_M+1}\int_{\Si_t}|\TT^{(k_*)}|^2+A_*\int_t^1s^{2A_*+2\de q+q_M}|\TT^{(k_*)}|^2ds\\
    \les_*\;&\ep_0^2+\ep\int_t^1s^{-1}\Hh(s)^2ds+\int_t^1s^{-1+2\si}\Dd(s)^2ds.
\end{align*}
\end{prop}
\begin{proof}
    Observing that $\f$ also satisfies the Vlasov equation $X(\f)=0$, by employing \eqref{Vlasoveq} we get
    \begin{align}
    \begin{split}\label{schematicVlasov}
        &\pr_t(\f)+\sum_{I=1}^3\frac{p^I}{p^0}t^{-\qIt}\pr_I(\f)-\sum_{I=1}^3\frac{\qIt}{t}p^I\pr_{p^I}(\f)\\
        =&O_p(1)\c\Gaw\c p\pr_p(\f)+O_p(1)\c\Gag\c\pr_x\f,
    \end{split}
    \end{align}
Here $O_p(1)$ denotes a certain homogeneous function of $(p^I)_{I=1,2,3}$ with degree $0$ and it satisfies 
\begin{equation*}
    |(p\pr_p)^{\le N} O_p(1)|\lesssim_N 1 \qquad \quad \text{for any} \quad N\in \mathbb{Z}_{\ge0}.
\end{equation*}
   Note that  a direct computation yields for any $I,J,K\in\{1,2,3\}$,
\begin{align*}
[p^J\pr_{p^K},p^I\pr_{p^I}]=p^I[p^J,\pr_{p^I}]\pr_{p^K}+p^J[\pr_{p^K},p^I]\pr_{p^I}=-\de_{IJ}p^J\pr_{p^K}+\de_{IK}p^J\pr_{p^I}.
\end{align*}
This implies
\begin{align}\label{commutationpp}
    \left[p^J\pr_{p^K},\sum_{I=1}^3\frac{\qIt}{t}p^I\pr_{p^I}\right]=\sum_{I=1}^3\frac{\qIt}{t}\left(-\de_{IJ}p^J\pr_{p^K}+\de_{IK}p^J\pr_{p^I}\right)=\frac{\qKt-\qJt}{t}p^J\pr_{p^K}.
\end{align}
   Thus, commuting \eqref{schematicVlasov} with $\pr_x^{\io_1}(p\pr_p)^{\io_2}$ for $|\io_1|+|\io_2|\leq\kl$, in view of \eqref{commutationpp} we obtain
    \begin{align}
    \begin{split}\label{Vlasovpartial}
        &\pr_t(\f^{(\io_1,\io_2)})+\frac{C_{\io_2}}{t}\f^{(\io_1,\io_2)}+\sum_{I=1}^3\frac{p^I}{p^0}t^{-\qIt}\pr_I(\f^{(\io_1,\io_2)})-\sum_{I=1}^3\frac{\qIt}{t}p^I\pr_{p^I}(\f^{(\io_1,\io_2)})\\
        =&O_p(1)\c\Gaw\c p\pr_p(\f^{(\io_1,\io_2)})+O_p(1)\c\Gag\c\pr_x(\f^{(\io_1,\io_2)})\\
        +&O_p(t^{-q_M})\f^{(\kl)}+O_p(1)\c\left(\Gaw^{(1)}\c\f^{(1)}\right)^{(\kl-1)},
    \end{split}
    \end{align}
    where $f^{(k)}\coloneqq \max_{|\io_1|+|\io_2|\leq k}f^{(\io_1,\io_2)}$ and $C_{\io_2}$ is a constant obeying
    \begin{align}\label{Cio}
        |C_{\io_2}|\leq |\io_2|\max_{I,J=1,2,3}\{\qIt-\qJt\}.
    \end{align}
    Multiplying \eqref{Vlasovpartial} by $2p^0\f^{(\io_1,\io_2)}$, we then deduce
    \begin{align}
    \begin{split}\label{p0f2}
        &\pr_t\left(p^0|\f^{(\io_1,\io_2)}|^2\right)+\frac{2C_{\io_2}}{t}\left(p^0|\f^{(\io_1,\io_2)}|^2\right)+\sum_{I=1}^3\frac{p^I}{p^0}t^{-\qIt}\pr_I\left(p^0|\f^{(\io_1,\io_2)}|^2\right)\\
        -&\sum_{I=1}^3\frac{\qIt}{t}p^I\pr_{p^I}\left(p^0|\f^{(\io_1,\io_2)}|^2\right)+\sum_{I=1}^3\frac{\qIt}{t}\frac{(p^I)^2}{(p^0)^2}\left(p^0|\f^{(\io_1,\io_2)}|^2\right)\\
        =&O_p(1)\c\Gaw\c p\pr_p(p^0|\f^{(\io_1,\io_2)}|^2)+O_p(1)\c\Gaw\c p^0|\f^{(\kl)}|^2\\
        +&O_p(1)\c\Gag\c\pr_x(p^0|\f^{(\io_1,\io_2)}|^2)+O_p(t^{-q_M})p^0|\f^{(\kl)}|^2\\
        +&O_p(1)\c\left(\Gaw^{(1)}\c\f^{(1)}\right)^{(\kl-1)}\c p^0\f^{(\kl)}.
    \end{split}
    \end{align}
    Then we multiply \eqref{p0f2} by $t^{2A_*+2\de q+q_M+1}$ and infer that
    \begin{align*}
    &\pr_t\left(t^{2A_*+2\dq+q_M+1}p^0|\f^{(\io_1,\io_2)}|^2\right)+\frac{2(C_{\io_2}-A_*-\de q)-q_M-1}{t}\left(t^{2A_*+2\dq+q_M+1}p^0|\f^{(\io_1,\io_2)}|^2\right)\\
        +&\sum_{I=1}^3\frac{p^I}{p^0}t^{-\qIt}\pr_I\left(t^{2A_*+2\dq+q_M+1}p^0|\f^{(\io_1,\io_2)}|^2\right)-\sum_{I=1}^3\frac{\qIt}{t}p^I\pr_{p^I}\left(t^{2A_*+2\dq+q_M+1}p^0|\f^{(\io_1,\io_2)}|^2\right)\\
        +&\sum_{I=1}^3\frac{\qIt}{t}\frac{(p^I)^2}{(p^0)^2}\left(t^{2A_*+2\dq+q_M+1}p^0|\f^{(\io_1,\io_2)}|^2\right)\\
        =&O_p(t^{2A_*+2\de q+q_M+1})\c\Gaw\c p\pr_p\left(p^0|\f^{(\io_1,\io_2)}|^2\right)+O_p(t^{2A_*+2\de q+q_M+1})\c\Gag\c\pr_x\left(p^0|\f^{(\io_1,\io_2)}|^2\right)\\
        +&O_p(t^{2A_*+2\de q+1})\c p^0|\f^{(\kl)}|^2+O(t^{2A_*+2\de q+q_M+1})\left(\Gaw^{(1)}\c\f^{(1)}\right)^{(\kl-1)}\c p^0\f^{(\kl)}.
    \end{align*}
    Integrating it on $T\Si_t$ and using \eqref{Def:TT new}, by integration by parts we derive
    \begin{equation}\label{Est:T high}
    \begin{aligned}
        &\pr_t\left(t^{2A_*+2\de q+q_M+1}\int_{\Si_t}|\TT^{(\io_1,\io_2)}|^2\right)+\frac{2(C_{\io_2}-A_*-\de q)}{t}t^{2A_*+2\de q+q_M+1}\int_{\Si_t}|\TT^{(\io_1,\io_2)}|^2\\
        \geq&\int_{T\Si_t}O_p(t^{2A_*+2\de q+q_M+1})\c\left(\Gaw,\Gag^{(1)}\right)\c p^0|\f^{(\io_1,\io_2)}|^2+O_p(t^{2A_*+2\de q+1})p^0|\f^{(\kl)}|^2\\
        +&\int_{T\Si_t}t^{2A_*+2\de q+q_M+1}\left(\Gaw^{(1)}\c (p^0)^\frac{1}{2}\f^{(1)}\right)^{(\kl-1)}\c (p^0)^\frac{1}{2}\f^{(\kl)}.
    \end{aligned}
        \end{equation}
    Here we utilize the facts that
    \begin{align*}
        \sum_{I=1}^3\qIt=1,\qquad \max_{1\leq I\leq 3}\qIt\leq q_M,\qquad\sum_{I=1}^3(p^I)^2=(p^0)^2.
    \end{align*}
    Notice that from our choice of $A_*, k_*$ as in \eqref{dfkl}, along  with \eqref{Cio} there holds
    \begin{equation}\label{Cond:C}
    \begin{aligned}
        2(C_{\io_2}-A_*-\de q)&\leq 2|\io_2|\max_{I,J=1,2,3}\{\qIt-\qJt\}-2A_*-2\de q\\
        &\leq 2\left(\kl\max_{I,J=1,2,3}\{\qIt-\qJt\}-A_*-\de q\right)\\
        &\leq -10.
    \end{aligned}
     \end{equation}
    Meanwhile, by virtue of Lemma \ref{decayGagGabGaw} we have
    \begin{align}\label{Est:GagGaw1}
\left\|\Gag^{(1)},\Gaw\right\|_{L^\infty(\Si_t)}\les\frac{\ep}{t}.
    \end{align}
    Injecting \eqref{Cond:C} and \eqref{Est:GagGaw1} into \eqref{Est:T high}, we then obtain
    \begin{align*}
        &\pr_t\left(t^{2A_*+2\de q+q_M+1}\int_{\Si_t}|\TT^{(\io_1,\io_2)}|^2\right)-9t^{2A_*+2\de q+q_M}\int_{\Si_t}|\TT^{(\io_1,\io_2)}|^2\\
        \geq&\int_{T\Si_t}O_p(t^{2A_*+2\de q+1})p^0|\f^{(\kl)}|^2+t^{2A_*+2\de q+q_M+1}\left(\Gaw^{(1)}\c (p^0)^\frac{1}{2}\f^{(1)}\right)^{(\kl-1)}\c (p^0)^\frac{1}{2}\f^{(\kl)}.
    \end{align*}
    Integrating from $t$ to 1 gives
    \begin{align}
    \begin{split}\label{estTT}
        &t^{2A_*+2\de q+q_M+1}\int_{\Si_t}|\TT^{(\io_1,\io_2)}|^2+\int_t^1s^{2A_*+2\de q+q_M}\int_{\Si_s}|\TT^{(\io_1,\io_2)}|^2ds\\
        \les_*&\ep_0^2+\int_t^1\int_{T\Si_s}s^{2A_*+2\de q+q_M+1}\left(\Gaw^{(1)}\c (p^0)^\frac{1}{2}\f^{(1)}\right)^{(\kl-1)}\c (p^0)^\frac{1}{2}\f^{(\kl)}ds\\
        +&\int_{t}^1s^{2A_*+2\de q+1}\int_{\Si_s}|\TT^{(k_*)}|^2ds.
    \end{split}
    \end{align}
    To control the rest of the error terms, applying Lemma \ref{decayGagGabGaw}, Lemma \ref{Generalinterpolation} and Lemma \ref{interpolation}, we deduce
\begin{align*}
    \left\|\left(\Gaw^{(1)}\c (p^0)^\frac{1}{2}\f^{(1)}\right)^{(\kl-1)}\right\|_{L^2(T\Si_t)}&\les_*\left\|\Gaw^{(k_*)}\right\|_{L^2_x}\left\|(p^0)^\frac{1}{2}\f^{(1)}\right\|_{L^\infty_xL^2_p}+\left\|D_w^{(1)}\right\|_{L^\infty_x}\left\|(p^0)^\frac{1}{2}\f^{(k_*)}\right\|_{L^2_xL^2_p}\\
&\les_*\|\Gaw^{(k_*)}\|_{L^2_x}\|\TT^{(1)}\|_{L^\infty_x}+\|D_w^{(1)}\|_{L^{\infty}_x} \|\TT^{(k_*)}\|_{L^2_x}\\
    &\les_*\frac{\ep\Hh(t)}{t^{A_*+1+\dq+\frac{q_M+1}{2}}}.
\end{align*}
Substituting this into \eqref{estTT}, we hence derive
\begin{align*}
    &t^{2A_*+2\de q+q_M+1}\int_{\Si_t}|\TT^{(\io_1,\io_2)}|^2+\int_t^1s^{2A_*+2\de q+q_M}\int_{\Si_s}|\TT^{(\io_1,\io_2)}|^2ds\\
    \les_*\;&\ep_0^2+\int_t^1\int_{\Si_s}s^{2A_*+2\de q+q_M+1}\left\|\left(\Gaw^{(1)}\c (p^0)^\frac{1}{2}\f^{(1)}\right)^{(\kl-1)}\right\|_{L^2_p}\TT^{(\kl)}+s^{2A_*+2\de q+1}|\TT^{(k_*)}|^2ds\\
    \les_*\;&\ep_0^2+\int_t^1 s^{2A_*+2\dq+q_M+1}\left\|\left(\Gaw^{(1)}\c (p^0)^\frac{1}{2}\f^{(1)}\right)^{(\kl-1)}\right\|_{L^2(T\Si_s)}\c\left\|\TT^{(\kl)}\right\|_{L^2(\Si_s)}ds\\
    &+\int_t^1 s^{2A_*+2\dq+1}\left\|\TT^{(k_*)}\right\|_{L^2(\Si_s)}^2ds\\
\les_*\;&\ep_0^2+\ep\int_t^1s^{-1}\Hh(s)^2ds+\ep\int_t^1s^{2A_*+2\dq+q_M}\left\|\TT^{(k_*)}\right\|^2_{L^2(\Si_s)}ds+\int_{t}^1s^{-1+2\si}\Dd(s)^2ds.
\end{align*}
    Finally, choosing $\ep$ small enough, we conclude
    \begin{align*}
    &t^{2A_*+2\de q+q_M+1}\int_{\Si_t}|\TT^{(k_*)}|^2+A_*\int_t^1s^{2A_*+2\de q+q_M}|\TT^{(k_*)}|^2ds\\
    \les_*\;&\ep_0^2+\ep\int_t^1s^{-1}\Hh(s)^2ds+\int_t^1s^{-1+2\si}\Dd(s)^2ds.
    \end{align*}
    This completes the proof of Proposition \ref{topf}.
\end{proof}
\subsection{End of the Proof of Theorem \ref{M3}}
Now we finish the proof of Theorem \ref{M3}. Combining Propositions \ref{estkga}--\ref{topf}, for $\varepsilon$ sufficiently small, we obtain
\begin{align*}
\Hh(t)^2+A_*\int_t^1s^{-1}\Hh(s)^2ds\le C_*\ep_0^2+C\int_t^1s^{-1}\Hh(s)^2ds+C_*\int_t^1s^{-1+2\si}\Dd(s)^2ds.
\end{align*}
Since the constant $C$ involved in the inequality above is independent of $A_*$, choosing $A_*$ large enough, we absorb the bulk term $\int_t^1 s^{-1}\Hh(s)ds$ on the right. The desired top-order estimate in Theorem \ref{M3} thus follows
\begin{align*}\Hh(t)^2\les_*\ep_0^2+\int_t^1s^{-1+2\si}\Dd(s)^2ds.
\end{align*}

\subsection{Proof of Theorem \ref{M4}}\label{Sec:M4proof}
In view of the bootstrap assumption \eqref{B1} and the Sobolev embedding, we readily get the uniform boundedness of $\|g\|_{C^2(\Si_t)}$ and $\|(k,\pr_t\psi, \nab\psi, F)\|_{C^1(\Si_t)}$ for $t\in (T_*, 1]$. It remains to show
\begin{equation*}
    \sup\limits_{t\in (T_*, 1]}\left\|\jp^{\mu+\kl} |p|^{-\kl} (\pr_x, p\pr_{p})^{\le 1}f\right\|_{L^{\infty}_xL^2_p(T\Si_t)}<\infty.
\end{equation*}
In fact, we will establish the below stronger quantitative energy estimate for $f$:
\begin{prop}\label{fH3}
    Under the same assumptions in \Cref{M4}, there exists $C>0$ such that for all $t\in (T_*, 1]$ it holds
    \begin{equation*}
        \|f\|_{H^{3}_{\mu}(T\Si_t)}\les e^{Ct^{-C}}\|f\|_{H^{3}_{\mu}(T\Si_1)}.
    \end{equation*}
\end{prop}
\begin{proof}
    We conduct the similar energy estimate as in the proof of Proposition \ref{topf}. By expanding the Vlasov equation $X(f)=0$ and commuting this with $(\pr_x)^{\io_1} (p\pr_p)^{\io_2}$ with $|\io_1|+|\io_2|\le 3$, mimicking the derivation of \eqref{Vlasovpartial} we get
     \begin{align*}
    \begin{split}
    &\pr_t(f^{(\io_1, \io_2)})+\frac{C_{\io_2}}{t}f^{(\io_1, \io_2)}+\sum_{I=1}^3\frac{p^I}{p^0}t^{-\qIt}\pr_I(f^{(\io_1, \io_2)})-\sum_{I=1}^3\frac{\qIt}{t}p^I\pr_{p^I}(f^{(\io_1, \io_2)})\\
        =&O_p(1)\c\Gaw\c p\pr_p(f^{(\io_1, \io_2)})+O_p(1)\c\Gag\c\pr_x(f^{(3)})\\
        +&O_p(t^{-q_M})f^{(3)}+O_p(1)\c\left(\Gaw^{(1)}\c f^{(1)}\right)^{(2)},
    \end{split}
    \end{align*}
    where $C_{\io_2}$ is a constant obeying
    \begin{align}\label{Cio1}
        |C_{\io_2}|\leq 3\max_{I,J=1,2,3}\{\qIt-\qJt\}.
    \end{align}
Multiplying by $wf^{(\io_1, \io_2)}$ with $w\coloneqq \jp^{2\mu+2\kl}|p|^{-2\kl}$ then yields
\begin{align}
    \begin{split}\label{Eqn:localf}
        &\pr_t\left(w|f^{(\io_1, \io_2)}|^2\right)+\frac{2C_{\io_2}}{t}\c \left(w|f^{(\io_1, \io_2)}|^2\right)+\sum_{I=1}^3 \frac{p^I}{p^0}t^{-\qIt} \pr_{I}\left(w|f^{(\io_1, \io_2)}|^2\right)\\-&\sum_{I=1}^3\frac{\qIt}{t}p^I\pr_{p^I}\left(w|f^{(\io_1, \io_2)}|^2\right)+\sum_{I=1}^3 \frac{\qIt}{t}p^I\pr_{p^I}(\log w)\c \left(w|f^{(\io_1, \io_2)}|^2\right)
        \\=&O_p(1)\c\Gaw\c p\pr_p(w|f^{(\io_1,\io_2)}|^2)+O_p(1)\c\Gaw\c w|f^{(3)}|^2\\
+&O_p(1)\c\Gag\c\pr_x(p^0|f^{(\io_1,\io_2)}|^2)+O_p(t^{-q_M})w|f^{(3)}|^2\\
        +&O_p(1)\c\left(\Gaw^{(1)}\c f^{(1)}\right)^{(2)}\c w f^{(3)}.
    \end{split}
\end{align}
Notice that
\begin{align*}
    \sum_{I=1}^3 \frac{\qIt}{t}p^I\pr_{p^I}(\log w)&=\sum_{I=1}^3 \frac{\qIt}{t}p^I\c\left( (\kl+\mu)\c \frac{2p^I}{\jp^2}-\kl\frac{2p^I}{|p|^2} \right)\leq 2\mu\sum_{I=1}^3\frac{\qIt}{t}\frac{|p^I|^2}{\jp^2}\leq\frac{2\mu q_M}{t}.
\end{align*}
Thus, multiplying \eqref{Eqn:localf} by $t^{P}$ with $P=6\de q+(2\mu+1)q_M+2$ and integrating it over $T\Si_t$, via integration by parts we derive\footnote{The derivation is analogous to that of \eqref{Est:T high}.}
    \begin{align*}
       &\pr_t\left(t^P\int_{T\Si_t}w|f^{(\io_1, \io_2)}|^2\right)+\frac{2C_{\io_2}+(2\mu+1)q_M+1-P}{t}\c \left(t^P\int_{T\Si_t}w|f^{(\io_1, \io_2)}|^2\right)
        \\\ge&\int_{T\Si_t}O_p(t^P)\c(\Gaw, \Gag^{(1)})\c (w|f^{(\io_1,\io_2)}|^2)
+O_p(t^{P-q_M})w|f^{(3)}|^2
        +O_p(t^P)\c\left(\Gaw^{(1)}\c f^{(1)}\right)^{(4)}\c w f^{(3)}.
    \end{align*}
    Integrating it from $t$ to $1$, thanks to \eqref{Est:GagGaw1} and \eqref{Cio1}, we hence infer 
    \begin{equation}\label{Est:localfintegral}
    \begin{aligned}
       &t^P\int_{T\Si_t}w|f^{(\io_1, \io_2)}|^2+\int_t^1s^{P-1}\int_{T\Si_s}w|f^{(\io_1, \io_2)}|^2 ds\\
    \les& \|f\|_{H^{3}_{\mu}(T\Si_1)}+\int_t^1 s^{P-q_M}\int_{T\Si_s}w|f^{(3)}|^2ds+\int_t^1 s^{P}\int_{T\Si_s}\left(\Gaw^{(1)}\c f^{(1)}\right)^{(2)}\c w f^{(3)}.
    \end{aligned}
    \end{equation}
    Regarding the last error term, from Lemma \ref{interpolation} we have
    \begin{equation*}
        \left|\left(\Gaw^{(1)}\c f^{(1)}\right)^{(2)}\right|\les|\Gaw^{(3)}|\c |f^{(3)}|\les \varepsilon t^{-1-2\de q} |f^{(3)}|.
    \end{equation*}
    As a consequence, summing \eqref{Est:localfintegral} for all $|\io_1|+|\io_2|\leq 3$, we deduce
    \begin{align*}
        t^P\|f\|_{H^{3}_{\mu}(T\Si_t)}^2\les& \|f\|_{H^{3}_{\mu}(T\Si_1)}^2+\int_t^1 \left(s^{-q_M}+\varepsilon s^{-1-2\de q} \right) \c \left(s^P\|f\|_{H^{3}_{\mu}(T\Si_s)}^2\right)ds\\
        \les&\|f\|_{H^{3}_{\mu}(T\Si_1)}^2+\int_t^1 s^{-1-2\de q} \c \left(s^P\|f\|_{H^{3}_{\mu}(T\Si_s)}^2\right)ds.
    \end{align*}
    Therefore, employing Gr\"onwall's inequality, we arrive at
    \begin{equation*}
        t^P\|f\|_{H^{3}_{\mu}(T\Si_t)}^2\les e^{\frac{t^{-2\de q}}{2\de q}}\|f\|_{H^{3}_{\mu}(T\Si_1)}^2.
    \end{equation*}
    This concludes the proof of Proposition \ref{fH3}.
\end{proof}
Finally, utilizing the Sobolev embedding on $\mathbb{T}^3$, we conclude
\begin{equation*}
   \sup\limits_{t\in (T_*, 1]}\left\|\jp^{\mu+\kl} |p|^{-\kl} (\pr_x, p\pr_{p})^{\le 1}f\right\|_{L^{\infty}_xL^2_p(T\Si_t)}\les \sup\limits_{t\in (T_*, 1]}\|f\|_{H^{3}_{\mu}(T\Si_t)}\les e^{CT_*^{-C}}\|f\|_{H^{3}_{\mu}(T\Si_1)}<\infty.
\end{equation*}
This finishes the proof of Theorem \ref{M4}.
\section{Physical Conclusions}\label{secconclusions}
In this section, we discuss several important physical implications of Theorem \ref{mainestimates}, which are analogous to Section 6 in \cite{FRS}.
\subsection{Limiting Functions and Kasner-like Behavior}
The following proposition shows that $\{tk_{IJ}\}_{I,J=1,2,3}$ and $t\pr_t\psi$ have limits in $W^{1,\infty}(\Tt^3)$ as $t\to 0$. This indicates that our perturbed spacetimes converge to a nearby ``Kasner-like'' spacetime as approaching the Big Bang singularity at $t=0$.
\begin{prop}\label{Kasnerlimit}
Within the perturbed spacetime $(\MM\simeq (0,1]\times \mathbb{T}^3, \g)$ solved in Theorem \ref{mainestimates}, the following limits exist:
\begin{align*}
\ka_{IJ}^{(\infty)}(x)\coloneqq \lim_{t\to 0}tk_{IJ}(t,x),\qquad B^{(\infty)}(x)\coloneqq \lim_{t\to 0}t\pr_t\psi(t,x).
\end{align*}
Furthermore, we have the below estimates:
\begin{align}
\|tk_{IJ}(t,\c)-\ka_{IJ}^{(\infty)}\|_{W^{1,\infty}(\Tt^3)}&\les\ep_0 t^\si,\qquad\qquad\|\ka_{IJ}^{(\infty)}+\qIt\de_{IJ}\|_{W^{1,\infty}(\Tt^3)}\les\ep_0,\label{kcontrol}\\
\|t\pr_t\psi(t,\c)-B^{(\infty)}\|_{W^{1,\infty}(\Tt^3)}&\les\ep_0 t^\si,\qquad\qquad\quad\,\|B^{(\infty)}-\Bt\|_{W^{ 1,\infty}(\Tt^3)}\les\ep_0.\label{Bcontrol}
\end{align}
In addition, for each $x\in\Tt$, the symmetric matrix $\left(-\ka_{IJ}^{(\infty)}(x)\right)_{I,J=1,2,3}$ has $3$ eigenvalues $q_I^{(\infty)}(x)$ which are the final Kasner exponents of the perturbed spacetime, that can be ordered such that $q_1^{(\infty)}$, $q_2^{(\infty)}$, $q_3^{(\infty)}\in C^{0,1}(\Tt^3)$ and such that the following estimate holds:
\begin{align}\label{qIcontrol}
    \sum_{I=1}^3\|q_I^{(\infty)}-\qIt\|_{C^{0,1}(\Tt^3)}\les\ep_0.
\end{align}
Moreover, the $\left\{q_I^{(\infty)}(x)\right\}_{I=1,2,3}$ and $B^{(\infty)}(x)$ satisfy the following algebraic relations:
\begin{align}\label{algebraicrelations}
    \sum_{I=1}^3q_I^{(\infty)}(x)=1,\qquad \sum_{I=1}^3\left[q_I^{(\infty)}(x)\right]^2=1-\left[B^{(\infty)}(x)\right]^2.
\end{align}
\end{prop}
\begin{proof}
Recall from \eqref{Eqn:prt k} that $k_{IJ}$ obeys
\begin{equation*}
     \pr_t(k_{IJ})+\frac{1}{t}k_{IJ}=-\ev(\ev n)+\ev\ga+O(t^{-1})\Gab+\Gab\c\Gaw+T.
\end{equation*}
Inserting the hyperbolic estimate \eqref{finalestimates} in Theorem \ref{mainestimates} and integrating over $[a, b]\subset (0, 1]$, we get
\begin{align}\label{Cauchy}
    \|ak_{IJ}(a,\c)-bk_{IJ}(b,\c)\|_{W^{1,\infty}(\Tt^3)}\les\int_a^bs^{-1+\si}\Dd(s)ds\les\ep_0b^\si,
\end{align}
Let $\{t_n\}_{n=1}^\infty\subset(0,1]$ be a decreasing sequence of times such that $\lim_{n\to\infty}t_n=0$. From \eqref{Cauchy} we have that $\{t_nk_{IJ}(t_n,\c)\}_{n=1}^\infty$ is a Cauchy sequence in $W^{1,\infty}(\Tt^3)$. We then denote 
\begin{align*}
    \ka_{IJ}^{(\infty)}\coloneqq \lim_{n\to\infty}\{t_nk_{IJ}(t_n,\c)\}_{n=1}^\infty.
\end{align*}
Thus, for any fixed $t\in(0, 1]$, choosing $(a, b)=(t_n,t)$ in \eqref{Cauchy} with $n$ large enough and letting $n\to \infty$, we derive
\begin{align*}
    \|\ka_{IJ}^{(\infty)}-tk_{IJ}(t,\c)\|_{W^{1,\infty}(\Tt^3)}\les\ep_0 t^\si.
\end{align*}
In particular, taking $t=1$ gives
\begin{align*}
    \|\ka_{IJ}^{(\infty)}-k_{IJ}(1,\c)\|_{W^{1,\infty}(\Tt^3)}\les\ep_0.
\end{align*}
Notice that from \eqref{initial} it holds
\begin{align*}
    \|k_{IJ}(1,\c)+\qIt\de_{IJ}\|_{W^{1,\infty}(\Tt^3)}\les\ep_0.
\end{align*}
Incorporating the above two inequalities, we deduce
\begin{align}\label{matrixperturbation}
    \|\ka_{IJ}^{(\infty)}+\qIt\de_{IJ}\|_{W^{1,\infty}(\Tt^3)}\les\ep_0.
\end{align}
as stated. The estimates \eqref{Bcontrol} can be derived in a similar manner.\\ \\
Next, notice that \eqref{matrixperturbation} implies that the symmetric matrix $(\ka_{IJ}^{(\infty)}(x))_{I,J=1,2,3}$ is $O(\ep_0)$--close to the diagonal matrix $\diag(-\widetilde{q_1},-\widetilde{q_2},-\widetilde{q_3})$. Thus, for each $x\in\Tt^3$, the eigenvalues of the diagonalizable matrix  $(\ka_{IJ}^{(\infty)}(x))_{I,J=1,2,3}$ can be ordered such that\footnote{It is a direct consequence of Weyl's inequality, see also (3.6) in Chapter IV of \cite{StewardSun}.}
\begin{align*}
    \sum_{I=1}^3\left|q_I^{(\infty)}(x)-\qIt\right|&\les\max_{I,J=1,2,3}\left|\ka_{IJ}^{(\infty)}(x)+\qIt\de_{IJ}\right|,\\
    \sum_{I=1}^3\left|q_I^{(\infty)}(x)-q_I^{(\infty)}(y)\right|&\les\max_{I,J=1,2,3}\left|\ka_{IJ}^{(\infty)}(x)-\ka_{IJ}^{(\infty)}(y)\right|.
\end{align*}
Combining these with \eqref{kcontrol}, we obtain that $q_I(x)\in C^{0,1}(\Tt^3)$ and it satisfies \eqref{qIcontrol}.\\ \\
Finally, in light of $\tr k=-\frac{1}{t}$ and \eqref{kcontrol} we have
\begin{align*}
    -1=t\tr k=O(\ep_0 t^\si)+\tr\ka^{(\infty)}(x)=O(\ep_0t^\si)-\sum_{I=1}^3q_I^{(\infty)}(x).
\end{align*}
Sending $t\to 0$ implies
\begin{align*}
    \sum_{I=1}^3q_I^{(\infty)}(x)=1.
\end{align*}
To get the second algebraic constraint in \eqref{algebraicrelations}, employing the Hamiltonian equation \eqref{2.26a} from Proposition \eqref{basicequations}, together with estimates in Theorem \ref{mainestimates}, \eqref{kcontrol} and \eqref{Bcontrol}, we arrive at
\begin{align*}
1&=t^2(e_0\psi)^2+t^2k_{CD}k_{CD}+O(\ep_0t^{\si})\\
&=\ka_{CD}^{(\infty)}(x)\ka_{CD}^{(\infty)}(x)+\left[B^{(\infty)}(x)\right]^2+O(\ep_0t^{\si})\\
&=\sum_{I=1}^3\left[q_I^{(\infty)}(x)\right]^2+\left[B^{(\infty)}(x)\right]^2+O(\ep_0t^{\si}).
\end{align*}
The desired identity hence follows by taking $t\to 0$.
\end{proof}
\subsection{Curvature Blow-up at \texorpdfstring{$t=0$}{}}
Building on the behaviors of limiting fields established in Proposition \ref{Kasnerlimit}, we now prove that the Kretschmann scalar blows up like $t^{-4}$ as $t\to0$ as below. In other words, this demonstrates that the Big Bang singularity exactly occurs at $t=0$.
\begin{prop}\label{Kretschmann}
Within the perturbed spacetime $(\MM\simeq (0,1]\times \mathbb{T}^3, \g)$ solved in Theorem \ref{mainestimates}, the Kretschmann scalar satisfies the following estimate:
\begin{align*}
\R^{\a\mu\b\nu}\R_{\a\mu\b\nu}&=4t^{-4}\left\{\sum_{I=1}^3\left[\left(q_I^{(\infty)}\right)^2-q_I^{(\infty)}\right]^2+\sum_{1\leq I<J\leq 3}\left(q_I^{(\infty)}\right)^2\left(q_J^{(\infty)}\right)^2\right\}+O(\ep_0t^{-4+\si})\\
&=4t^{-4}\left\{\sum_{I=1}^3\left[\qIt^2-\qIt\right]^2+\sum_{1\leq I<J\leq 3}\qIt^2\qJt^2\right\}+O(\ep_0t^{-4}).
\end{align*}
\end{prop}
\begin{proof}
Throughout this proof, we adopt the Einstein summation conventions for $I,J=1,2,3$ indices as well. According to the definition of the Kretschmann scalar, we directly calculate
\begin{align}
\begin{split}\label{Kretschmanncomputation}
\R^{\a\mu\b\nu}\R_{\a\mu\b\nu}&=\R(e_A,e_I,e_B,e_J)\R(e_A,e_I,e_B,e_J)+4\R(e_0,e_I,e_0,e_J)\R(e_0,e_I,e_0,e_J)\\
&-4\R(e_A,e_I,e_0,e_J)\R(e_A,e_I,e_0,e_J).
\end{split}
\end{align}
Using the Gauss equation \eqref{Gauss}, together with estimates in Theorem \ref{mainestimates} and \eqref{kcontrol}, we obtain
\begin{equation}\label{Est:R1}
\begin{aligned}
    &t^2\R(e_A,e_I,e_B,e_J)\\
    =\;&t^2R(e_A,e_I,e_B,e_J)+t^2k_{AB}k_{IJ}-t^2k_{AJ}k_{BI}\\
    =\;&t^2g\left(\nab_{e_A}\nab_{e_I}e_J-\nab_{e_I}\nab_{e_A}e_J-\nab_{[e_A,e_I]}e_J,e_B\right)+(tk_{AB})(tk_{IJ})-(tk_{AJ})(tk_{BI})\\
    =\;&  (tk_{AB})(tk_{IJ})-(tk_{AJ})(tk_{BI})+\ev\ga+\ga\c\ga\\
    =\;&\ka_{AB}^{(\infty)}\ka_{IJ}^{(\infty)}-\ka_{AJ}^{(\infty)}\ka_{BI}^{(\infty)}+O(\ep_0t^\si).
\end{aligned}
\end{equation}
Similarly, employing the Codazzi equation \eqref{Codazzi} and Theorem \ref{mainestimates} we get
\begin{align}\label{Est:R2}
    t^2\R(e_A,e_I,e_0,e_J)=t^2e_A(k_{IJ})-t^2e_I(k_{AJ})+O(\ep_0t^\si)=O(\ep_0t^\si).
\end{align}
Meanwhile, in view of \eqref{2.28}, \eqref{sRiceIeJ}, along with hyperbolic estimates in Theorem \ref{mainestimates}, we infer
\begin{align}\label{Est:R3}
t^2\R(e_0,e_I,e_0,e_J)=-\ka_{IJ}^{(\infty)}-\ka_{IC}^{(\infty)}\ka_{CJ}^{(\infty)}+O(\ep_0t^\si).
\end{align}
Denote $K=\left(\ka_{IJ}^{(\infty)}\right)_{I,J=1,2,3}$. From Proposition \ref{Kasnerlimit} we know that $K$ is a $3\times 3$ symmetric matrix with eigenvalues $(-q_I^{(\infty)})_{I=1,2,3}$. Injecting \eqref{Est:R1}--\eqref{Est:R3} into \eqref{Kretschmanncomputation}, we therefore conclude
\begin{align*}
t^4\R^{\a\mu\b\nu}\R_{\a\mu\b\nu}&=(K_{IJ}K_{AB}-K_{AJ}K_{BI})(K_{IJ}K_{AB}-K_{AJ}K_{BI})\\
&+4(K_{IJ}+K_{IB}K_{BJ})(K_{IJ}+K_{IC}K_{CJ})+O(\ep_0t^{\si})\\
&=2[\tr\left(K^2\right)]^2+4\tr\left(K^2\right)+8\tr\left(K^3\right)+2\tr\left(K^4\right)+O(\ep_0t^{\si})\\
&=2\left[\sum_{I=1}^3(q_I^{(\infty)})^2\right]^2+4\sum_{I=1}^3(q_I^{(\infty)})^2-8\sum_{I=1}^3(q_I^{(\infty)})^3+2\sum_{I=1}^3(q_I^{(\infty)})^4+O(\ep_0t^{\si})\\
&=4\sum_{I=1}^3(q_I^{(\infty)})^4+4\sum_{1\leq I<J\leq 3}(q_I^{(\infty)})^2(q_J^{(\infty)})^2+4\sum_{I=1}^3(q_I^{(\infty)})^2-8\sum_{I=1}^3(q_I^{(\infty)})^3+O(\ep_0t^{\si})\\
&=4\sum_{I=1}^3\left[(q_I^{(\infty)})^2-q_I^{(\infty)}\right]^2+4\sum_{1\leq I<J\leq 3}(q_I^{(\infty)})^2(q_J^{(\infty)})^2+O(\ep_0t^{\si}).
\end{align*}
Combining with \eqref{qIcontrol}, this completes the proof of Proposition \ref{Kretschmann}.
\end{proof}

\vspace{0.1cm}
Xinliang An: Department of Mathematics, National University of Singapore, Singapore 119076. \\
Email: \textit{matax@nus.edu.sg}\\ \\
Taoran He: Department of Mathematics, National University of Singapore, Singapore 119076. \\
Email: \textit{taoran$\_$he@u.nus.edu}\\ \\
Dawei Shen: Department of Mathematics, Columbia University, New York, NY, 10027. \\
Email: \textit{ds4350@columbia.edu}

\begin{thebibliography}{30}
\addcontentsline{toc}{section}{References}
\bibitem{ABIO:2021_Royal_Soc} E.~Ames, F.~Beyer, J.~Isenberg, and T.~A.~Oliynyk, \emph{Stability of asymptotic behavior within polarised {$\mathbb{T}^2$}-symmetric vacuum solutions with cosmological constant}, arXiv:2108.02886. To appear in Phil.~Trans.~R.~Soc.~A.
\bibitem{ABIO:2021} E.~Ames, F.~Beyer, J.~Isenberg, and T.~A.~Oliynyk, \emph{Stability of {AVTD} behavior within the polarized {$\mathbb{T}^2$}-symmetric vacuum spacetimes}, To appear in Ann. Inst. Henri Poincar\'{e}, arXiv:2101.03167, (2021).
\bibitem{AF} N.~Athanasiou and G.~Fournodavlos, \textit{A localized construction of Kasner-like singularities}, arXiv:2412.16630, (2024).
\bibitem{BOZ:subcrit} F.~Beyer, T.~A.~Oliynyk and W.~Zheng, \emph{Localized past stability of the subcritical Kasner-scalar field spacetimes}, arXiv:2502.09210, (2025).
\bibitem{Boothroyd} M.~Boothroyd, \emph{Analysis of the Einstein--Maxwell--Scalar Field Equations in Cosmology}, Master thesis, University of Otago, Dunedin, New Zealand, 2025.
\bibitem{Darmour} T.~Darmour, M.~Henneaux, A.~D.~Rendall and M.~Weaver, \emph{Kasner-Like Behaviour for Subcritical Einstein-Matter Systems}, Ann. Henri Poincar\'e, \textbf{3} (2002), 1049--1111.
\bibitem{Edward} S.~Edward, \emph{Quark--Gluon Plasma, Heavy Ion Collisions and Hadrons}, World Scientific Lecture Notes in Physics. Vol \textbf{85}. Singapore: World Scientific, 2024.
\bibitem{FU22} D.~Fajman and L.~Urban, \emph{Cosmic Censorship near FLRW spacetimes with negative spatial curvature}, Anal. PDE \textbf{18} (2025), No. 7, pp. 1615--1713.
\bibitem{FU24} D.~Fajman and L.~Urban, \emph{On the past maximal development of near-FLRW data for the Einstein scalar-field Vlasov system}, arXiv:2402.08544, (2024).
\bibitem{FRS} G.~Fournodavlos, I.~Rodnianski and J.~Speck, \emph{Stable big bang formation for Einstein’s equations: the complete sub-critical regime}, J. Amer. Math. Soc. \textbf{36} (3) (2023), 827--916.
\bibitem{GOPR} H.~O.~Groeniger, O.~Petersen and H.~Ringstr\"om, \emph{Formation of quiescent big bang singularities}, arXiv:2309.11370, (2023).

\bibitem{Ringstrom09} H.~Ringstrom, \emph{The Cauchy problem in general relativity}, ESI Lectures in Mathematics and Physics, European Mathematical Society (EMS), Z\"urich, 2009.
\bibitem{RSscalar} I.~Rodnianski and J.~Speck, \emph{A regime of linear stability for the Einstein-scalar field system with applications to nonlinear big bang formation}, Ann. of Math. (2), \textbf{187} (1) (2018), 65--156.
\bibitem{RSstiff} I.~Rodnianski and J.~Speck, \emph{Stable Big Bang formation in near-FLRW solutions to the Einstein-scalar field and Einstein-stiff fluid systems}, Selecta Math. (N.S.) \textbf{24} (5) (2018), 4293--4459.
\bibitem{RSvacuum} I.~Rodnianski and J.~Speck, \emph{On the nature of Hawking’s incompleteness for the Einstein-vacuum equations: The regime of moderately spatially anisotropic initial data}, J. Eur. Math. Soc. \textbf{24} (1) (2022), 167--263.
\bibitem{Shen22} D.~Shen, \emph{Stability of Minkowski spacetime in exterior regions}, Pure Appl. Math. Q. \textbf{20} (2) (2024), 757--868.
\bibitem{Speck} J.~Speck, \emph{The maximal development of near-FLRW data for the Einstein-scalar field system with spatial topology $\mathbb{S}^3$}, Comm. Math. Phys. \textbf{364} (3) (2018), 879--979.
\bibitem{StewardSun} G.~W.~Steward and J.~Sun, \emph{Matrix perturbation theory}, Computer Science and Scientifi Computing, Academic Press, Inc. Boston, MA, 1990.
\bibitem{Svedberg} C.~Svedberg, \emph{Future stability of the Einstein-Maxwell-Scalar field system and non-linear wave
equations coupled to generalized massive-massless Vlasov equations}, PhD thesis, KTH, Mathematics
(Div.), 2012. QC 20120503.
\bibitem{Taylor} M.~Taylor, \emph{The global nonlinear stability of Minkowski space for the massless Einstein–Vlasov system}, Ann. PDE, \textbf{3}, Art. 9, 2017.
\bibitem{Urban} L.~Urban, \emph{Quiescent Big Bang formation in $2+1$ dimensions}, arXiv:2412.03396, (2024).
\end{thebibliography}
\end{document}